\documentclass[opre,nonblindrev]{informs4}
\usepackage{eqndefns-left,multirow,threeparttable,bbm} 
\RequirePackage{endnotes}

\definecolor{MyDarkBlue}{RGB}{158,0,0}
\usepackage[colorlinks=TRUE,citecolor=MyDarkBlue,urlcolor=MyDarkBlue,linkcolor=MyDarkBlue,hypertexnames=false]{hyperref}
\usepackage{pgfplots}
\pgfplotsset{compat=1.18}

\usepackage{caption}
\usepackage[labelfont=sf]{subcaption}
\OneAndAHalfSpacedXI

\usepackage{algorithm,booktabs}
\usepackage{algpseudocode}
\usepackage{tikz}

\usepackage{natbib}
 \bibpunct[, ]{(}{)}{,}{a}{}{,}%
 \def\bibfont{\small}%

\newcommand{\E}{\mathbb{E}}
\renewcommand{\P}{\mathbb{P}}

\newcommand{\ind}{\mathbf{1}}
\newcommand{\OPT}{\operatorname{OPT}}

\newcommand{\Rev}{\operatorname{Rev}}

\EquationsNumberedThrough    

\TheoremsNumberedThrough     
\ECRepeatTheorems  %

\MANUSCRIPTNO{}

\begin{document}


\RUNAUTHOR{Sun and Lin}

\RUNTITLE{Online Stochastic Allocation with Increasing Returns}

\TITLE{Online Stochastic Allocation with Increasing Returns}

\ARTICLEAUTHORS{%
\AUTHOR{Shuo Sun}
\AFF{Department of Industrial \& Systems Engineering,
The University of Wisconsin–Madison, \EMAIL{ssun279@wisc.edu}}

\AUTHOR{Yunduan Lin}
\AFF{Department of Decisions, Operations, and Technology, CUHK Business School, The Chinese University of Hong Kong, \EMAIL{yunduanlin@cuhk.edu.hk}}
} 


\ABSTRACT{%
We study online allocation with increasing returns, where an uncertain number of customers arrive sequentially and can be irrevocably assigned to any product with remaining capacity. Each assignment earns a base reward plus a bonus that is nondecreasing in the number of previous assignments to the same product. Prior work has focused primarily on fixed or diminishing returns. Under increasing returns, the central challenge is that concentrating assignments can unlock large future rewards, while a different product may be preferable if arrivals stop early. We show that classical approaches, including greedy and natural LP-based rounding algorithms, can perform arbitrarily poorly, even with homogeneous bonuses. We evaluate performance through competitive ratios relative to the hindsight offline optimum. For homogeneous bonuses, we establish a tight competitive ratio of $1/2$. For arbitrary heterogeneous bonuses, however, no online algorithm can guarantee a competitive ratio greater than $1/m$, where $m$ is the number of products. By imposing discrete concavity on the bonus functions, we limit the concentration of rewards in later assignments and restore constant guarantees. Under this structure, we develop a re-solving framework that repeatedly selects a planning horizon, constructs a feasible allocation plan from the current state, and executes its assignments in a specific order. For arbitrary arrival distributions, we design distribution-oblivious deterministic and randomized algorithms with competitive ratios of $0.2537$ and $0.3864$, respectively. For known increasing-failure-rate (IFR) arrival distributions, we exploit distributional information to obtain a deterministic $1/2$-competitive algorithm.
}%




\KEYWORDS{Online resource allocation, Increasing returns, Online algorithms, Competitive analysis} 

\maketitle


\section{Introduction}
In many contemporary markets, the economic rewards generated by products and services increase with accumulated adoption, making later transactions more valuable than earlier ones. This phenomenon has gained prominence as social media and online marketplaces make market trends more visible. A recent illustration is the rise of ``Labubunomics''  \citep{modernretail2025labubunomics}. Labubu, a collectible character commercialized by Pop Mart, gained global momentum following high-profile public exposure and celebrity adoption. As its popularity surged, secondary market prices rose sharply, with its ``Big Into Energy'' six-pack reselling at an average premium of nearly 60\% over its original retail price \citep{reuters2025comiccon}.

More broadly, increasing returns arise through diverse economic mechanisms beyond consumer markets shaped by social influence. In settings where participants interact and collaborate, network externalities increase the value of participation: larger professional communities offer more opportunities for networking and knowledge exchange, while multiplayer gaming sessions with more players enable richer cooperation and coordination. In two-sided service marketplaces, successful transactions can strengthen a provider's reputation, allowing subsequent matches to generate higher rewards as established providers may command higher fees. In production and service operations, learning by doing can reduce unit costs as experience accumulates, increasing margins from subsequent sales. Similar benefits may arise from other settings involving information accumulation, coordination improvements, and economies of density. Despite their different mechanisms, these settings share a common feature: early activity not only generates immediate rewards but also contributes to the buildup that makes subsequent activity more valuable. 

This intertemporal dependence creates a fundamental challenge for sequential allocation decisions. When customers arrive over time, assigning them to different products or services affects not only immediate rewards but also where future reward potential develops. Motivated by this feature, we revisit online resource allocation, a fundamental problem in operations management that studies how sequentially arriving customers should be assigned to capacity-limited resources. Despite its broad applications in revenue management, online matching, and platform operations, existing research primarily considers fixed rewards or diminishing returns~\citep{KVV1990,MSVV2007,FMMM2009,DevanurHayes2009,Mehta2013}. These models capture settings with independent transaction revenues or diminishing marginal benefits due to congestion or saturation. However, how increasing returns affect online allocation decisions and how to design effective algorithms that account for reward buildup remain largely unexplored.

In this paper, we study online resource allocation with increasing returns. We consider a set of capacity-limited products, where each arriving customer must be assigned immediately and irrevocably to a product with remaining capacity or rejected. Each assignment generates a base reward plus a nonnegative bonus that is nondecreasing in the number of customers previously assigned to that product. Products may differ in their capacities, base rewards, and bonus functions. An online algorithm aims to maximize the expected total reward of all allocations, knowing the product characteristics and the distribution of total arrivals, but not the realized arrival count. To evaluate its performance, we compare against the hindsight offline optimum, which observes the realized arrival count before choosing the best allocation.

Perhaps surprisingly, the seemingly modest change in reward structure fundamentally changes the nature of online allocation. In classical settings with fixed or diminishing returns, a central challenge arises from customer-product compatibility: serving a customer may consume capacity needed by future customers with fewer options. Under full compatibility, a greedy policy achieves the hindsight offline optimum for every arrival count. Under increasing returns, however, a different challenge arises even with full compatibility: each assignment affects not only the immediate reward but also the rewards from future assignments to the same product. The online algorithm must decide whether to prioritize high immediate returns or invest in building up a product to obtain greater future rewards. Uncertainty about the arrival counts makes this choice difficult, as investing in buildup may pay off if enough customers arrive but be unprofitable if arrivals stop early.

This new trade-off challenges classical algorithms for online resource allocation. We show that greedy and natural linear programming (LP)-based rounding algorithms can be arbitrarily bad. A greedy policy may repeatedly favor products with high current marginal rewards, reinforcing their advantages while overlooking opportunities to build up rewards elsewhere. Such self-reinforcement resembles the ``rich-get-richer'' phenomenon and popularity bias in recommendation systems, where excessive emphasis on popular items can disadvantage both consumers and providers \citep{klimashevskaia2024survey}. LP-based methods can account for future rewards, but a formulation based on expected demand may place too much weight on rewards that are unlikely to be reached before arrivals stop.
Even when an LP accounts for arrival probabilities, independently rounding a fractional optimal solution can disperse assignments across products, destroying the coordination needed to realize increasing returns.  These failures reveal that increasing returns reshape the fundamental challenges of online allocation, calling for a deeper understanding of the problem and new algorithms that coordinate assignments to build up rewards while preserving value when arrivals stop early.

\subsection{Our Results and Technical Ideas}
The main contribution of this paper is to provide a systematic understanding of online resource allocation with increasing returns and develop effective algorithmic approaches. We characterize how bonus structures and arrival uncertainty shape the achievable competitive ratios, establishing both positive results and impossibility bounds. In the following, we present our findings progressively, highlighting the key structural insights and technical ideas behind each result.

We first consider \emph{homogeneous bonuses}, where all products share a nondecreasing bonus function but may differ in base rewards and capacities. Although classical approaches can perform arbitrarily poorly, the common bonus structure allows us to establish a tight competitive ratio of $1/2$ by separating base rewards from bonus
rewards. We construct two fixed-order policies, filling products in nonincreasing order of base reward and capacity, respectively. Each policy maximizes its corresponding reward component for every arrival count, so their combined reward dominates the offline optimum. Thus, selecting the better policy in expectation achieves a $1/2$-competitive ratio. We show that this guarantee is
tight through a two-product instance where the first assignment
must choose between earning a higher immediate reward and
initiating buildup that yields a substantial bonus only if a
second customer arrives. This conflict between short- and long-term rewards, which the offline optimum can resolve using the realized arrival count, yields the matching $1/2$ bound.

We next consider \emph{arbitrary heterogeneous bonuses}, where products have different nondecreasing bonus functions. Unlike the homogeneous case, the bonus reward can no longer be maximized
through a common capacity-based ordering, as the potential value of buildup varies across products. We show that the optimal worst-case competitive ratio is exactly $1/m$. The guarantee follows by committing all arrivals to the product with the highest expected standalone reward, since the expected offline optimum is bounded by their sum. To establish the matching upper bound, we extend the conflict underlying the homogeneous example to $m$ products, each offering a substantial bonus after a different number of assignments. This creates a multiway commitment problem in which uncertainty about the arrival count leads to a matching $1/m$ bound. Thus, without additional restrictions on bonus functions, no constant competitive ratio independent of $m$ is possible.

This impossibility result highlights the fundamental difficulty of online allocation with increasing returns. Comparing the upper-bound constructions reveals a common pattern: substantial rewards are concentrated at assignment thresholds, creating competing buildup opportunities across arrival counts. In practice, however, rewards may increase more gradually, with diminishing increments in marginal rewards as adoption accumulates. The worst-case behavior under arbitrary bonuses may be too conservative. Motivated by this distinction, we consider nonnegative, nondecreasing, and \emph{discrete-concave} bonus functions, whose successive increments are nonincreasing. This restriction limits extreme increases in marginal rewards
at late assignment thresholds while preserving increasing
returns. We show that discrete concavity restores constant competitive ratios for heterogeneous products, although classical greedy and LP-based rounding can still perform arbitrarily poorly.

Our algorithms use a common re-solving framework that combines \emph{allocation planning} and \emph{assignment ordering}. The central idea is to use offline optimal allocations for different arrival counts to guide a sequence of irrevocable online decisions. At each epoch, the algorithm selects an offline benchmark to pursue, constructs a feasible plan specifying the additional assignments to each product given the current allocation, and determines their execution order and when to re-solve. This re-solving framework is motivated by \emph{incremental maximization}, which seeks a growing solution whose prefixes approximate the offline optimum at different cardinalities \citep{bernstein2022general}. However, increasing returns violate the structural assumptions underlying existing incremental-maximization guarantees, as rewards may be concentrated in later assignments. Under discrete concavity, we develop new structural bounds for coordinating successive, nonnested allocation plans and obtaining constant competitive ratios. Moreover, when the arrival distribution is known, we exploit the distributional information to guide the selection of planning horizons. We further discuss the connection to incremental maximization in Section~\ref{sec:related-work}.

A key challenge is that arrivals may stop before a plan is completed, so only its initial assignments are executed. Under increasing returns, much of the plan's reward may come from later assignments, making early stopping particularly costly. Although the assignment order does not affect the reward of a completed plan, it determines how much reward is earned if arrivals stop early. We therefore carefully choose the order in which assignments within each plan are made. Specifically, we construct a \emph{backward-greedy ordering}: at each step, we place last an assignment whose removal causes the smallest decrease in the remaining plan's reward. Under discrete concavity, for any current allocation and any feasible plan of \(k\) additional assignments, the first \(h\) assignments in this order earn at least a fraction \(h(h+1)/(k(k+1))\) of the plan’s total incremental reward. This quadratic prefix-value guarantee bounds the reward earned if arrivals stop before completion. The same bound also limits the growth of the offline optimum to at most quadratic in the arrival count, allowing comparisons across planning horizons. We next consider how to select these benchmarks and when to re-solve.

For \emph{general arrival distributions}, we develop distribution-oblivious algorithms that pursue offline benchmarks at geometrically increasing arrival levels. If arrivals continue through an epoch, the algorithm secures its target benchmark; if they stop within the epoch, the reward secured previously, together with the prefix-value guarantee for partial execution, provides a lower bound on the accumulated reward. Geometric spacing allows us to compare these guarantees with the offline optimum at every stopping count. We first develop a deterministic geometric schedule that achieves a $0.2537$-competitive ratio. We improve the ratio to \(0.3864\) through a randomized algorithm with two refinements. First, accumulated reward may allow a benchmark to be attained before its allocation plan is completed; we use this fact to schedule shorter epochs. Second, randomly shifting the geometric grid allows us to average the guarantee over the possible positions of a fixed stopping count within an epoch. We implement these refinements through a fractional allocation construction and jointly round the resulting path to obtain feasible integral assignments. Both algorithms require no knowledge of the arrival distribution. Their guarantees hold for every fixed arrival count, with the randomized guarantee taken in expectation over the random shift and joint rounding.

For known increasing-failure-rate (IFR) arrival distributions, where the conditional probability of another arrival is nonincreasing, we obtain a deterministic \(1/2\)-competitive algorithm. At each epoch, the algorithm uses the conditional arrival distribution to select a planning horizon. We define a potential equal to half the expected residual offline reward and seek a plan whose expected reward covers the decrease in this potential. The prefix-value guarantee bounds the reward earned if arrivals stop before plan completion, while a residual value bound controls the offline reward remaining afterward. We show that the IFR property guarantees a plan satisfying the potential condition at every epoch. Summing the resulting inequalities establishes the \(1/2\) guarantee.

Together, we demonstrate that increasing returns introduce a new trade-off in online allocation, leading to fundamental challenges and even impossibility results that do not arise in classical settings. Addressing this trade-off requires new algorithmic approaches that manage reward buildup through coordinated planning and execution. The achievable guarantees depend critically on the structure of bonus rewards and the available arrival information. Table~\ref{tab:regimes} summarizes this landscape of performance guarantees and impossibility results.
\begin{table}[ht!]
\centering
\begin{threeparttable}
\caption{Performance guarantees relative to the hindsight offline optimum.}
\label{tab:regimes}
\begin{tabular}{@{}p{0.52\textwidth}ll@{}}
\toprule
Bonus and arrival assumptions & Competitive ratio & Main results \\
\midrule
Homogeneous bonuses; general arrivals & Tight $1/2$ & Thm.~\ref{thm:homo-half-approx} \\
Arbitrary heterogeneous bonuses; general arrivals & Tight $1/m$ & Thm.~\ref{thm:hete-decay} \\
Discrete-concave bonuses; general arrivals & $0.2537$ (det.); $0.3864$ (rand.) & Thms.~\ref{thm:general-d-cr} and \ref{thm:general-r-cr} \\
Discrete-concave bonuses; IFR arrivals & $1/2$ & Thm.~\ref{thm:ifr-half} \\
\bottomrule
\end{tabular}
\begin{tablenotes}[flushleft]
    \scriptsize
    \item Notes: Det. and rand. denote performance guarantees from deterministic and randomized algorithms, respectively. The randomized guarantee holds in expectation over the algorithm's internal randomization.
\end{tablenotes}
\end{threeparttable}
\end{table}
\vspace{-5pt}

\subsection{Related Work}\label{sec:related-work}
Our work builds on the extensive literature on online resource allocation. The seminal work of \citet{KVV1990} establishes the optimal $(1-1/e)$-competitive ratio for online bipartite matching under adversarial arrivals. Subsequent research studies variants of online matching \citep{AggarwalGoelKarandeMehta2011,FMMM2009}, online advertising and revenue management \citep{MSVV2007,DevanurHayes2009,JasinKumar2012}. For comprehensive surveys, see \citet{Mehta2013,huang2024online}. Most of these models consider fixed returns, while allocation-dependent returns primarily focus on diminishing returns, including concave and submodular objectives \citep{devanur2012online,kapralov2013online,patton2026online,udwani2026optimality}. In contrast, increasing returns remain underexplored. One exception is \citet{blum2015online}, who study economies of scale, where decreasing marginal production costs increase returns from subsequent sales. They consider heterogeneous buyer valuations drawn from an unknown distribution and learned as buyers arrive. Our setting instead assumes known reward functions, with uncertainty in total arrivals determining whether buildup can be completed.

Classical online allocation largely focuses on the trade-off induced by compatibility constraints, where current assignments may limit options for future customers. Recent studies also consider full compatibility but examine other sources of uncertainty, including unknown rewards and resource consumption \citep{zhalechian2022online} and nonstationary arrivals \citep{zhang2024nonstationary}. Our model likewise considers full compatibility, focusing on the distinct source of difficulty driven by reward buildup rather than compatibility constraints.


Methodologically, our approach draws on re-solving techniques in online allocation and network revenue management, which repeatedly update allocation plans based on remaining capacities \citep{JasinKumar2012,BumpensantiWang2020,vera2021online}. Under increasing returns, re-solving must also account for accumulated reward buildup and nonnested offline allocations. A particularly close connection is incremental maximization, which seeks a single irrevocable sequence whose initial solution segments (prefixes) approximate offline optima at different cardinalities. \citet{bernstein2022general} develop algorithms by combining offline solutions, while \citet{disser2023incremental} introduce continuous formulations and randomized scaling techniques. Although this objective parallels our need to perform well across possible arrival counts, our setting introduces two important distinctions. First, their guarantees rely on \emph{accountability}, which requires some element of every nonempty set to contribute at most its average value, and this property fails under increasing returns. Second, incremental maximization focuses on uniform performance across cardinalities, whereas our stochastic arrival setting allows distributional information to guide the selection and coordination of successive plans. We will account for both distinctions in our algorithms and analysis.

Finally, our work also relates to the economics and operations management literature on increasing returns. We highlight three representative mechanisms. First, network externalities capture how growing adoption enhances subsequent consumers' valuations \citep{katz1985network,economides1996network}. Such effects have attracted substantial attention in operations management, including inventory decisions under social influence \citep{hu2016liking} and pricing and assortment with network effects \citep{candogan2012optimal,wang2017consumer,chen2021duopoly}. Closely related to our setting, \citet{lin2024content} study the allocation of promotional resources across content, balancing immediate gains against future adoption through diffusion. Unlike their data-driven promotion optimization over a fixed horizon, ours considers sequential allocation under uncertain arrivals. Second, social learning and reputation building can enhance subsequent returns as previous transactions convey favorable information and establish credibility \citep{pallais2014inefficient,liu2026broadcast}. Third, learning by doing reduces production costs as cumulative experience improves efficiency \citep{wright1936factors,arrow1962economic,gray2009outsourcing}, while related economies of scale lower average costs and influence sourcing decisions \citep{cachon2002competition}. Together, these studies demonstrate the broad relevance of increasing returns and their implications for operational decision-making.

\section{Model and Benchmarks}
\label{sec:model}
In this section, we introduce the online allocation model with increasing marginal rewards and define the benchmarks used to evaluate online policies. 

We consider a decision maker who allocates sequentially arriving customers among \(m\) capacity-constrained products, indexed by \([m]=\{1,\ldots,m\}\). Product \(i\) has capacity \(b_i\in\mathbb Z_{>0}\), and the total capacity is \(B=\sum_{i\in[m]}b_i\). Let \(L\) denote the number of customer arrivals, drawn from distribution \(\mathcal D_L\) supported on \(\{0,1,\ldots,L_{\max}\}\). Upon each arrival, the decision maker must immediately and irrevocably assign the customer to a product with remaining capacity or reject the customer. We study two distributional settings. In the \emph{general setting}, \(\mathcal D_L\) is arbitrary. In the \emph{increasing-failure-rate (IFR) setting}, we restrict attention to distributions whose discrete hazard rate \(\Pr(L=\ell\mid L\ge \ell)\) is nondecreasing in \(\ell\) over values for which the conditional probability is defined. The IFR class includes binomial distributions and finite truncations of Poisson and geometric distributions \citep{shaked2007stochastic}.

Each assignment to product $i$ generates a base reward $r_i\ge0$ and a bonus that depends on the number of previous assignments to that product. Specifically, let $f_i:\mathbb Z_{\ge0}\to\mathbb R_{\ge0}$ be a nondecreasing function with $f_i(0)=0$. If $k$ customers have been assigned to product $i$, assigning the next customer yields reward $r_i+f_i(k+1)$. Thus, the cumulative reward of assigning $n$ customers to product $i$ is
\begin{equation}
    \label{eq:cum-reward}
    R_i(n):=nr_i+\sum_{k=1}^{n}f_i(k),
    \qquad n=0,1,\ldots,b_i.
\end{equation}
Because $f_i$ is nondecreasing, an assignment can increase the marginal reward of subsequent assignments to the same product. In this work, we focus on the full compatibility setting, in which every arriving customer can be assigned to any product with remaining capacity. 

Our objective is to maximize the expected total reward. An online policy knows the capacities, reward functions, and arrival distribution $\mathcal D_L$, but not the realized value of $L$; distribution-oblivious policies do not use $\mathcal D_L$. For an admissible online policy $\pi$, which may be randomized, let $\Rev^\pi(\ell)$ denote its expected reward when $L=\ell$, where the expectation is over the policy's internal randomization. The overall expected reward is $\Rev^\pi:=\E_{L\sim\mathcal D_L}[\Rev^\pi(L)]$. 
We evaluate online policies against the offline optimum benchmark, which assumes that the realized arrival count $L=\ell$ is known in advance. Under full compatibility, the offline reward depends only on the final allocation counts. Let $\boldsymbol{x}=(x_1,\ldots,x_m)$ denote these counts, with total reward $V(\boldsymbol{x}):=\sum_{i\in[m]}R_i(x_i)$. For a realization $L=\ell$, the offline optimal value is
\begin{equation}
\label{eq:offline-opt}
U(\ell)=\max_{\boldsymbol{x}\in\mathbb Z_{\ge0}^m} V(\boldsymbol{x})
\quad\text{s.t.}\quad
0\le x_i\le b_i,\quad i\in[m],\quad
\|\boldsymbol{x}\|_1\le\ell,
\end{equation}
and its expectation is $\OPT_{\mathrm{off}}:=\E_{L\sim\mathcal D_L}[U(L)]$.
We evaluate online policies against the expected offline optimum. Specifically, an online policy $\pi$ is $c$-competitive if $\Rev^\pi\ge c\,\OPT_{\mathrm{off}}$ for every instance under consideration. 
Problem~\eqref{eq:offline-opt} can be solved by dynamic programming in $O(m(B+1)^2)$ time (see Appendix~\ref{app:benchmark-offline}).

\begin{remark}[Full compatibility]
\label{rem:full-compatibility}
We focus on full compatibility to isolate the challenges arising from increasing returns and arrival uncertainty. With nonincreasing marginal rewards, greedy achieves the offline optimum for every realized arrival count by selecting the largest available marginal reward. Thus, uncertainty about the number of arrivals alone does not cause any performance loss. Under increasing returns, however, an assignment with a low immediate reward may enable higher rewards from subsequent assignments. Consequently, offline optimal allocations for different arrival counts need not be nested, and arrival uncertainty can lead to a performance loss even under full compatibility. This differs from classical online allocation settings, where a central difficulty arises from compatibility constraints: assigning a customer to a product may consume capacity needed by future customers with fewer options.

\end{remark}

\section{Performance Limits under Increasing Returns}
\label{sec:limit}
In this section, we examine how the bonus structure affects the achievable competitive ratios under increasing returns. We first show that greedy and natural LP-based independent-rounding algorithms can perform arbitrarily poorly, even under homogeneous bonuses. Nevertheless, homogeneous bonuses admit a tight $1/2$ competitive ratio, whereas arbitrary heterogeneous bonuses admit only a tight $1/m$ ratio, where $m$ is the number of products. The $1/m$ bound shows that additional structure is necessary for a constant competitive ratio, motivating our study of discrete-concave bonuses in Section~\ref{sec:dc}.

\subsection{Limitations of Greedy and Independent LP-Based Rounding}
\label{sec:limit-fail}
We first examine two natural approaches to online allocation: greedy, which prioritizes immediate marginal rewards, and LP-based rounding, which uses ex-ante optimization to account for future rewards. We show that both can perform arbitrarily poorly even relative to the optimal online policy, and hence also relative to the offline optimum, under homogeneous bonuses.

We begin with greedy. Let $\boldsymbol{s}=(s_1,\ldots,s_m)$ denote the
current allocation, where $s_i$ is the number of customers assigned to
product $i$. Upon each arrival, greedy selects an available product
with the highest immediate marginal reward:
\(
i\in\arg\max_{j:s_j<b_j}\{r_j+f_j(s_j+1)\}.
\)
As discussed in Remark~\ref{rem:full-compatibility}, greedy achieves
the offline optimum under full compatibility and nonincreasing returns.
Under increasing returns, however, an assignment with a low immediate
reward may enable substantially higher rewards from subsequent
assignments. The following proposition shows that ignoring these
future rewards can lead to arbitrarily poor performance.

\begin{proposition}[Failure of greedy]
\label{prop:greedy-fail}
For every $\varepsilon>0$, there exists an instance with homogeneous
bonuses in which the expected reward of greedy is less than
$\varepsilon\OPT_{\mathrm{off}}$.
\end{proposition}

Note that this failure persists relative to the optimal online policy, rather than arising solely from comparison with the offline benchmark. 
To illustrate, consider two products with capacities one and two,
respectively. Product 1 has base reward one, while product 2 has
base reward zero. Both share the bonus function $f(1)=0$ and
$f(2)=M$. The first customer arrives with certainty, and a second
customer arrives with a fixed positive probability $p$.
Greedy assigns the first customer to product 1 and earns a total
reward of one. In contrast, assigning both customers to product 2,
if they arrive, yields an expected reward of $pM$.
Thus, the performance ratio approaches zero as $M\to\infty$.
We defer the formal proof, together with the remaining proofs in
this section, to Appendix~\ref{app:limit}.

An alternative to greedy is to account for future rewards through
an ex-ante LP and implement its solution using independent rounding.
Let $\mu_L:=\E[L]$ denote the expected number of arrivals.
For each product $i\in[m]$ and assignment level $k=1,\ldots,b_i$,
let $y_{i,k}$ represent the relaxed probability that product $i$
receives at least $k$ assignments. Consider the following mean-demand LP:
\begingroup
\setlength{\abovedisplayskip}{2pt}
\setlength{\belowdisplayskip}{2pt}
\begin{subequations}
\label{eq:lp}
\begin{align}
V^{\mathrm{LP}}=\max_{\boldsymbol y\in\mathbb R_{\ge0}^{B}}\quad
&\sum_{i\in[m]}\sum_{k=1}^{b_i}(r_i+f_i(k))y_{i,k}
\label{eq:lp-obj-inc}\\
\text{s.t.}\quad
&\sum_{i\in[m]}\sum_{k=1}^{b_i}y_{i,k}\le\mu_L,\label{eq:lp-cap-inc}\\
&y_{i,k}\le y_{i,k-1},\quad i\in[m],\ k=2,\ldots,b_i,\label{eq:lp-order-inc}\\
&0\le y_{i,k}\le1,\quad i\in[m],\ k=1,\ldots,b_i.
\label{eq:lp-bounds-inc}
\end{align}
\end{subequations}
\endgroup
The objective accounts for rewards at different assignment levels,
while the constraints impose an expected allocation budget and
precedence among assignment levels. 
Given an optimal solution $\boldsymbol y^*$, the algorithm independently
assigns each arriving customer to product $i$ with probability
\(
p_i=\frac{1}{\mu_L}\sum_{k=1}^{b_i}y_{i,k}^*.
\)
A customer is rejected with the remaining probability or if the
selected product has reached capacity. The following proposition
shows that this approach can also perform arbitrarily poorly.

\begin{proposition}[Failure of mean-demand LP rounding]
\label{prop:lp-fail}
For every $\varepsilon>0$, there exists an instance with homogeneous
bonuses in which the expected reward of mean-demand LP rounding is less than
$\varepsilon\OPT_{\mathrm{off}}$.
\end{proposition}
To illustrate, consider two products with capacities one and two,
respectively. Product 1 has base reward one, while product 2 has
base reward zero. Both share the bonus function $f(1)=0$ and $f(2)=3$.
The first customer arrives with certainty, and a second customer
arrives with probability $p$.
The optimal LP solution is
$y_{1,1}^*=0$ and $y_{2,1}^*=y_{2,2}^*=(1+p)/2$.
Thus, independent rounding assigns every arrival to product 2,
earning an expected reward of $3p$.
In contrast, assigning the first customer to product 1 guarantees
a reward of one. Hence, the performance ratio approaches zero as
$p\to0$.

A natural strengthening is to incorporate arrival reachability, for example through the constraints $y_{i,k}\le\P(L\ge k)$. We examine a more detailed time-indexed LP that explicitly accounts for arrival survival probabilities, assignment positions, and within-product precedence. Nevertheless, independent rounding of an optimal LP solution can still perform arbitrarily poorly, even relative to the optimal online policy and under homogeneous bonuses. The difficulty is that independent rounding may disperse assignments across products, preventing the coordinated buildup needed to realize increasing returns. We establish this failure in Appendix~\ref{app:dc-fail-2}.

Together, Propositions~\ref{prop:greedy-fail} and~\ref{prop:lp-fail} reveal two limitations of classical approaches under increasing returns. Greedy overlooks the value of buildup, whereas mean-demand LP rounding may place excessive weight on rewards that are unlikely to be realized. Moreover, independent rounding can disperse assignments across products, preventing valuable buildup even when arrival uncertainty is explicitly incorporated. These limitations highlight the need to account for both uncertainty in future arrivals and the coordination of assignments required to realize increasing returns.

\subsection{Homogeneous Bonuses: A Tight \texorpdfstring{$1/2$}{1/2} Guarantee}
\label{sec:limit-homo}
We now show that a simple algorithm achieves a tight competitive ratio of $1/2$ under homogeneous bonuses. Let all products share a common bonus function $f$, that is, $f_i=f$ for every $i\in[m]$. The cumulative reward from assigning $n$ customers to product $i$ is $R_i(n)=nr_i+\sum_{k=1}^n f(k)$. Thus, the offline objective consists of a product-specific base reward and a common bonus reward. These components favor different allocation priorities. Maximizing the base component favors products with larger $r_i$, whereas maximizing the bonus component favors concentrating assignments on products with larger capacities. Indeed, since $f$ is nondecreasing, $\sum_{k=1}^n f(k)$ is discrete convex and filling products in nonincreasing capacity order maximizes the bonus component for every arrival count. This motivates two complementary fixed-order allocation rules.
\begin{definition}[Fixed-order fill algorithm]
      Fix an ordering $\sigma=(\sigma_1,\dots,\sigma_m)$ of the products. The fixed-order fill algorithm assigns each arriving customer to the first product in the order $\sigma$ with remaining capacity. If all products have reached capacity, the customer is rejected.
\end{definition}
We consider two fixed-order fill algorithms corresponding to the priorities above:
\begin{itemize}
    \item \textbf{Base-ordered fill algorithm}: $\sigma^{\mathrm{base}}$ orders products in nonincreasing base reward $r_i$.
    \item \textbf{Capacity-ordered fill algorithm}: $\sigma^{\mathrm{cap}}$ orders products in nonincreasing capacity $b_i$.
\end{itemize} Each ordering optimizes one component of the offline objective for every realized arrival count. Selecting the ordering with the larger expected total reward yields the following $1/2$ competitive ratio. Moreover, we also show that the $1/2$ guarantee is tight. 


\begin{theorem}
\label{thm:homo-half-approx}
Under homogeneous bonuses, the better of the base-ordered and capacity-ordered fill algorithms achieves a $1/2$-competitive ratio. 
\end{theorem}
The proof establishes a stronger guarantee for every realized arrival count $\ell$. The base-ordered and capacity-ordered fill algorithms maximize the base reward and bonus reward of the offline objective, respectively. Since both components are nonnegative, the sum of the rewards earned by the two algorithms is at least the offline optimal reward for every $\ell$. Taking expectations over $L$ and selecting the algorithm with the larger expected reward yields the competitive ratio of $1/2$.

The result also illustrates how the two allocation priorities address the limitations identified in Section~\ref{sec:limit-fail}. The capacity-ordered algorithm preserves the opportunity to reach the valuable second assignment in the greedy counterexample, whereas the base-ordered algorithm secures the immediate reward when further arrivals are unlikely in the LP-based rounding counterexample. Therefore, rather than relying on a single rule, the proposed algorithm considers two complementary priorities and uses the known distribution $\mathcal D_L$ to select the better policy before arrivals begin. This indicates the algorithm design needs to carefully balance the trade-off under increasing returns arising from arrival uncertainty: investing in early buildup sacrifices immediate rewards but can unlock higher rewards if more customers arrive.

We next show that the $1/2$ competitive ratio is tight and formalize the corresponding result in Appendix~\ref{app:limit-homo}. The proof uses the same two-product, two-arrival structure as the preceding counterexamples, but balances the two possible first assignments to reveal a loss even for the optimal online policy. 
Consider two products and an arrival count satisfying $\P(L=1)=1-p$ and $\P(L=2)=p$, where $p\in(0,1)$. Product 1 has capacity one and base reward $1$, while product 2 has capacity two and base reward $0$. Both products share the bonus function $f(0)=f(1)=0$ and $f(k)=1/p$ for $k\ge2$. Assigning the first customer to product 1 secures reward $1$, while assigning it to product 2 yields expected reward $p(1/p)=1$. Thus, no online policy can earn more than $1$ in expectation. In contrast, the offline optimum earns $1$ when $L=1$ and $1/p$ when $L=2$, giving $\OPT_{\mathrm{off}}=2-p$. Consequently, the optimal online-to-offline ratio is $1/(2-p)$, which approaches $1/2$ as $p\to0$. 

The tightness construction reveals a fundamental loss arising from arrival uncertainty under increasing returns. Under fixed rewards and full compatibility, greedy is optimal for every arrival count, so uncertainty alone does not create a gap between online and offline performance. Under increasing returns, however, different arrival counts may favor conflicting early assignments, and even the optimal online policy cannot avoid this loss. Homogeneity nevertheless provides sufficient structure to limit the gap: the common bonus function admits a capacity-based ordering that maximizes bonus rewards for every arrival count, allowing the two complementary allocation priorities to achieve the tight $1/2$ guarantee. We next examine how the achievable guarantee changes when this common bonus structure is absent.


\subsection{Heterogeneous Bonuses: A Tight \texorpdfstring{$1/m$}{1/m} Competitive Ratio}
\label{sec:limit-hete}
With heterogeneous bonuses, products may exhibit different patterns of reward buildup, so different arrival counts can favor different products and their corresponding buildup targets. Unlike the homogeneous case, a single capacity-based ordering need not maximize bonus rewards across arrival counts. Consequently, an online algorithm must choose among potentially conflicting buildup opportunities without knowing how many customers will arrive. We show that this additional complexity has a sharp consequence: the optimal worst-case competitive ratio decreases to $1/m$, where $m$ is the number of products, ruling out any constant guarantee independent of $m$.

\begin{theorem}[Tight $1/m$ guarantee under arbitrary heterogeneous bonuses]
\label{thm:hete-decay}
For every integer $m\ge2$, there exists an online algorithm that achieves a $1/m$-competitive ratio for all instances with $m$ products and arbitrary heterogeneous bonuses. Moreover, no online algorithm can achieve a competitive ratio strictly greater than $1/m$ over this class.
\end{theorem}

The $1/m$ guarantee is achieved by a simple algorithm that commits all arrivals to the product with the largest standalone expected reward until its capacity is exhausted. For any realized arrival count $\ell$, the offline reward is at most the sum of rewards that each product could earn from up to $\ell$ assignments. Taking expectations gives
\(
\OPT_{\mathrm{off}}
\le \sum_{i=1}^m \E[R_i(\min\{L,b_i\})]
\le m\max_{i\in[m]}\E[R_i(\min\{L,b_i\})].
\)
Thus, committing to the product with the largest standalone expected reward achieves at least a $1/m$ fraction of the offline benchmark.

To establish tightness, we extend the preceding two-product
construction to $m$ products, each offering a lump-sum reward at a different assignment level. 
Fix $p\in(0,1)$. One customer arrives with certainty, and after each
arrival, another customer arrives with probability $p$, up to a total
of $m$ customers. Thus, $\P(L\ge i)=p^{i-1}$ for $i\in[m]$.
Product $i$ has capacity $b_i=i$ and base reward $r_i=0$. Its bonus function satisfies $f_i(k)=0$ for $k<i$ and
$f_i(k)=p^{-(i-1)}$ for $k\ge i$. 
Because its capacity is $i$, product $i$ earns no reward before its
final assignment, which yields $p^{-(i-1)}$. Assigning all arrivals
to any one product $i$, up to its capacity, therefore yields expected
reward $\P(L\ge i)p^{-(i-1)}=1$. The offline benchmark can choose which product to fill after observing
the arrival count. When $L=i$, filling product $i$ yields reward
$p^{-(i-1)}$. Since $\P(L=i)=(1-p)p^{i-1}$ for $i<m$ and
$\P(L=m)=p^{m-1}$, we obtain
\(
\OPT_{\mathrm{off}}
\ge \sum_{i=1}^m \P(L=i)p^{-(i-1)}
= m-(m-1)p.
\)

Now consider any online algorithm $\pi$ and condition on its first
decision. If the first customer is assigned to product $j$, the
expected reward from that product is at most
$p^{-(j-1)}\P(L\ge j)=1$. Filling any other product $i\neq j$
requires at least $i+1$ arrivals, since the first customer was
assigned elsewhere. Its expected reward is therefore at most
\(
p^{-(i-1)}\P(L\ge i+1)\le p.
\)
If the first customer is rejected, the same bound applies to every
product, giving total expected reward at most
$mp\le1+(m-1)p$. Averaging over the first decision therefore gives
$\Rev^\pi\le1+(m-1)p$, including when $\pi$ is randomized.
Consequently,
\(
\frac{\Rev^\pi}{\OPT_{\mathrm{off}}}
\le
\frac{1+(m-1)p}{m-(m-1)p}
\longrightarrow \frac1m
\) as $p\rightarrow 0$. 
This establishes tightness.

Taken together, the tightness constructions under homogeneous and arbitrary heterogeneous bonuses both rely on sharply backloaded rewards, where bonuses are concentrated at particular assignment levels and incomplete buildup generates no bonus. Under homogeneous bonuses, a common capacity-based ordering maximizes the bonus component across arrival counts, limiting the conflict between immediate rewards and future buildup. Under arbitrary heterogeneous bonuses, however, different products can have distinct reward thresholds, creating multiple conflicting buildup targets. This distinction explains the deterioration from the tight $1/2$ guarantee to the $1/m$ barrier: heterogeneity transforms a two-way trade-off into a multiway commitment problem, fundamentally limiting achievable online performance. 

\section{Algorithms under Increasing Discrete-Concave Bonuses}
\label{sec:dc} 

While arbitrary heterogeneous bonuses preclude a constant competitive ratio independent of the number of products,  the upper-bound construction relies on sharply backloaded rewards at different assignment thresholds. In many applications, however, rewards build up more gradually as adoption accumulates, suggesting that this worst-case limitation may be overly pessimistic. Motivated by this observation, we consider nondecreasing, discrete-concave bonus functions, which allow heterogeneous increasing returns while limiting reward backloading. We show that this structure restores constant competitive ratios and develop online algorithms that achieve them. 
Formally, we impose the following assumption throughout this section.
\begin{assumption}[Discrete concavity]
\label{assump:discrete-concavity}
For each $i\in[m]$, the bonus function $f_i$ is nonnegative, nondecreasing, and satisfies \(f_i(k+1)-f_i(k) \le f_i(k)-f_i(k-1)\) for $k=1,\ldots,b_i-1$.
\end{assumption}
Assumption~\ref{assump:discrete-concavity} is a mild regularity condition consistent with concave specifications commonly used in the literature. Such increasing and concave effects arise through several established economic mechanisms. Depending on the application, the bonus function $f_i$ may represent higher consumer valuations or willingness to pay, increased contribution margins, or cost reductions from previous assignments. We highlight three representative examples.
\begin{itemize}
    \item \textit{Network externality and installed-base effects:} $f_i(k)=\alpha_i k$, or a general increasing concave function such as $f_i(k)=\alpha_i\log(1+k)$, with $\alpha_i\ge0$. With network externalities, a larger installed base can raise a product's value to subsequent users. \citet{economides1996network} models willingness to pay as increasing with the installed base, considering both linear and concave network benefits. The linear form captures constant incremental benefits, while concave forms reflect saturation as adoption expands.

    \item \textit{Learning by doing and quantity economies:} $f_i(k)=\alpha_i\left[1-(1+k)^{-\beta_i}\right]$, with $\alpha_i\ge0$ and $\beta_i>0$. Learning-by-doing models describe how unit costs decline with cumulative production or experience \citep{wright1936factors,arrow1962economic}. Expressing these cost reductions as rewards yields an increasing concave bonus: productivity continues to improve, but gains from further experience diminish. Similar effects also arise from procurement economies, volume discounts, and fulfillment scale economies. 

    \item \textit{Social learning and Bayesian updating:} $f_i(k)=\alpha_ik/(k+\beta_i)$, for some $\alpha_i,\beta_i>0$. Social learning allows consumers to infer product quality from prior purchases. \citet{liu2026broadcast} study this mechanism in online sales, where accumulated purchases can raise perceived quality when true quality exceeds prior beliefs. Under Bayesian updating, the influence of additional favorable observations typically diminishes as information accumulates, motivating increasing concave specifications. Related mechanisms arise also through reviews, ratings, reputation, and observational learning.
\end{itemize}

These examples illustrate increasing discrete concavity as a common reduced-form structure for cumulative economic effects. Beyond this structure, our assumption imposes no specific functional form, allowing the magnitude and shape of buildup effects to vary across products.

Assumption~\ref{assump:discrete-concavity} applies to the per-assignment bonus $f_i$, rather than the cumulative reward $R_i$. Since $f_i$ is nondecreasing, the marginal reward $R_i(k)-R_i(k-1)=r_i+f_i(k)$ remains nondecreasing, so $R_i$ is discrete convex and repeated assignments to the same product remain complementary. Thus, the fundamental trade-off between immediate rewards and uncertain future buildup persists. 

Moreover, discrete concavity does not resolve the limitations of classical algorithms. As shown in Appendix~\ref{app:dc-fail}, both greedy and LP-based rounding still fail even with homogeneous linear bonuses, for reasons similar to those discussed above. These limitations call for new algorithms that exploit discrete concavity to pursue valuable buildup while preserving rewards under uncertain arrivals.

\subsection{The Re-solving Framework}
\label{sec:dc-re}

In this subsection, we introduce a common algorithmic  framework, followed by specific algorithms and performance guarantees under general and IFR arrival distributions in Sections~\ref{sec:dc-gen} and~\ref{sec:dc-ifr}, respectively.

In our setting, online allocation can be viewed as
an irrevocable sequence of product assignments, with the realized
arrival count determining how much of the sequence is executed.
This connects our setting to incremental maximization
\citep{bernstein2022general,disser2023incremental}, which seeks
a single sequence whose prefixes approximate offline optimal
solutions at different cardinalities. This connection motivates
our use of offline allocations as targets for online decisions.


Our re-solving framework divides the allocation process into
successive epochs. Each active epoch has two components.
First, \emph{allocation planning} selects an offline benchmark
to pursue, constructs a feasible set of additional assignments
from the current allocation, and determines when to re-solve.
These decisions balance rewards across possible arrival counts.
Second, \emph{assignment ordering} determines the sequence in
which the planned assignments are executed. Although the reward
from completing all planned assignments does not depend on their
order, the reward earned before completion does. The algorithm
therefore uses an ordering that protects against early stopping.
At the next re-solving point, it constructs a new plan based on
the assignments already made.

Although our re-solving framework shares the structure of incremental maximization, two distinctions require different design of allocation planning and assignment ordering. First, incremental maximization seeks uniform performance across all stopping counts, whereas our objective is expected reward under stochastic arrivals. Thus, allocation planning can exploit the arrival distribution, when available, to balance rewards across possible stopping counts. Second, incremental maximization relies on the accountability property to obtain a linear prefix-value guarantee for \emph{assignment ordering}: an appropriate ordering of $k$ elements ensures that the first $h$ retain at least an $h/k$ fraction of the total value. Such a guarantee need not hold under increasing returns, even with Assumption~\ref{assump:discrete-concavity}. For example, with a single product, zero base reward $r=0$, and linear bonuses $f(k)=k$, completing half of an allocation yields only about one-quarter of its total reward. Nevertheless, Assumption~\ref{assump:discrete-concavity} limits how rapidly marginal rewards can increase, allowing us to establish a quadratic prefix-value guarantee using backward-greedy ordering. This guarantee controls partial execution of a chosen plan from any current allocation, independently of the arrival distribution. We first establish this ordering guarantee and then develop
allocation planning strategies for general and IFR arrivals.

\paragraph{Backward-greedy ordering.}
Let $\boldsymbol s\in\mathbb Z_{\ge 0}^m$ denote the current allocation. An additional allocation plan is a vector
$\boldsymbol a\in\mathbb Z_{\ge0}^m$ satisfying
$\boldsymbol s+\boldsymbol a\le\boldsymbol b$. Let $h_a:=\|\boldsymbol a\|_1$ denote the number of additional assignments. An assignment order $\omega=(\omega_1,\ldots,\omega_{h_a})$ specifies the product receiving each successive assignment. For $h=0,\ldots,h_a$, let $\boldsymbol a^{(h)}$ denote the allocation induced by the first $h$ assignments, where \(a_i^{(h)}=\sum_{k=1}^{h}\mathbbm{1}\{\omega_k=i\}\) for $i\in[m]$. In particular, $\boldsymbol a^{(0)}=\boldsymbol0$ and $\boldsymbol a^{(h_a)}=\boldsymbol a$. If arrivals stop after $h$ additional assignments, the incremental reward is $V(\boldsymbol s+\boldsymbol a^{(h)})-V(\boldsymbol s)$.

We construct the assignment order backward, starting with
the assignment placed last. Suppose all planned assignments
have been completed. Removing one assignment to product $i$
reduces the reward by its terminal marginal value
$r_i+f_i(s_i+a_i)$. We therefore place last an assignment
with the smallest terminal marginal reward, minimizing the
reward lost if execution stops one assignment short of completion.
We then remove this assignment from the remaining plan and
apply the same rule recursively. Formally, initialize $\boldsymbol a^{(h_a)}=\boldsymbol a$.
For $h=h_a,\ldots,1$, choose
\(
\omega_h\in
\argmin_{i:a_i^{(h)}>0}
\left\{r_i+f_i(s_i+a_i^{(h)})\right\},
\)
and set
\(
\boldsymbol a^{(h-1)}
=\boldsymbol a^{(h)}-\boldsymbol e_{\omega_h}.
\)
Here, $\boldsymbol e_i$ is the $i$th unit vector. We call the resulting sequence the \emph{backward-greedy ordering}.
Unlike myopic greedy, which selects the assignment with the
largest current marginal reward, backward greedy evaluates
marginal rewards at the remaining planned allocation.
It therefore accounts for the buildup required to obtain
future rewards. This order enables us to establish a reward guarantee for every prefix under discrete concavity.

\begin{lemma}[Quadratic prefix-value guarantee]
\label{lemma:quad-prefix}
Under Assumption~\ref{assump:discrete-concavity}, for any current allocation $\boldsymbol s$ and feasible additional allocation $\boldsymbol a$ with $h_a=\|\boldsymbol a\|_1\ge1$, the backward-greedy ordering satisfies \(V(\boldsymbol s+\boldsymbol a^{(h)})-V(\boldsymbol s)\ge\frac{h(h+1)}{h_a(h_a+1)}\bigl[V(\boldsymbol s+\boldsymbol a)-V(\boldsymbol s)\bigr]\) for every $h=0,\ldots,h_a$.
\end{lemma}
Lemma~\ref{lemma:quad-prefix} guarantees that executing a fraction of the planned assignments retains at least the square of that fraction of the total incremental reward. Since the bound holds for any additional allocation from any current state, it can be applied repeatedly as new plans are constructed without revising previous assignments. The proof relies on discrete concavity. Among $k$ remaining
assignments, the smallest terminal marginal reward is at most
$2/(k+1)$ times their total incremental reward. Removing this
assignment therefore preserves at least a fraction
$(k-1)/(k+1)$ of the remaining reward. Applying this bound
successively yields the lemma. The guarantee is tight:
for a single product with zero base reward and linear bonuses
$f(k)=k$, starting from the empty allocation, the first
$h$ assignments earn $h(h+1)/2$, attaining the bound
with equality. 

The lemma also controls how the offline optimal value grows
with the arrival count. Fix integers $1\le h\le k\le B$.
Since rewards are nonnegative, an optimal allocation for
$k$ arrivals can be chosen to assign all $k$ customers.
Execute these assignments in backward-greedy order.
The first $h$ assignments form a feasible allocation for
$h$ arrivals, so their reward cannot exceed $U(h)$.
Applying Lemma~\ref{lemma:quad-prefix}, we obtain
\begin{equation}
\label{eq:offline-value-growth}
U(h)
\ge \frac{h(h+1)}{k(k+1)}U(k)
\ge \left(\frac hk\right)^2U(k).
\end{equation}
Thus, discrete concavity limits the growth of the offline
optimal value to at most quadratic in the arrival count.

The prefix-value guarantee and the offline comparison bound
serve complementary roles: the former protects rewards
under partial execution, while the latter allows comparisons
across planning horizons. The remaining challenge is to
decide which offline benchmarks to pursue and when to re-solve.
Since optimal allocations for different arrival counts need
not be nested, pursuing longer-horizon targets may sacrifice
immediate rewards, whereas shorter-horizon targets may delay
valuable buildup. We address this trade-off in two ways.
For general arrival distributions, we develop a distribution-oblivious
approach based on geometrically spaced arrival levels, with
guarantees for every stopping count. Under IFR arrivals, we
use the conditional arrival distribution to select planning
horizons and improve the competitive ratio.

\subsection{Algorithms under General Arrivals}
\label{sec:dc-gen}

We now consider general arrival distributions and aim to develop distribution-oblivious algorithms that require no knowledge of the arrival distribution. Specifically, we seek policies $\pi$ satisfying $\Rev^\pi(\ell)\ge c\,U(\ell)$ for every fixed arrival count $\ell\ge0$, where the expectation in $\Rev^\pi(\ell)$ is over the policy's internal randomization. Averaging this pointwise guarantee over any distribution of $L$ yields $\Rev^\pi\ge c\,\OPT_{\mathrm{off}}$. As in \citet{bernstein2022general}, we pursue offline benchmarks at geometrically increasing arrival levels. This spacing keeps previously attained benchmarks comparable to the offline optimum at later arrival counts.
We first develop a deterministic algorithm that achieves a $0.2537$-competitive ratio. We then improve the guarantee to $0.3864$ through a fractional construction that supports shorter re-solving epochs and a continuous random shift of the geometric grid, followed by joint rounding to obtain integral assignments.

\subsubsection{Allocation Planning and a Deterministic Algorithm.}
\label{sec:dc-gen-d}

We first describe the allocation-planning construction shared by
both algorithms. Given a current allocation $\boldsymbol s$ and
an integer benchmark level $n$, let $\boldsymbol x^*(n)$ be an
offline optimal allocation achieving $U(n)$ \eqref{eq:offline-opt}, chosen so that
$\|\boldsymbol x^*(n)\|_1=\min\{n,B\}$.
We construct the additional allocation plan
\(
\boldsymbol a(n)
=(\boldsymbol x^*(n)-\boldsymbol s)^+,
\)
where the positive part is taken coordinatewise. This plan fills the shortfalls relative to $\boldsymbol x^*(n)$
without undoing previous assignments.
The plan is feasible from $\boldsymbol s$, and its completion
guarantees
\(
V(\boldsymbol s+\boldsymbol a(n))
=V(\boldsymbol s\vee\boldsymbol x^*(n))
\ge U(n).
\) Here, $\vee$ denotes the coordinatewise maximum. The inequality holds because
$\boldsymbol s\vee\boldsymbol x^*(n)
\ge\boldsymbol x^*(n)$ coordinatewise
and each $R_i$ is nondecreasing.

If arrivals stop before the plan is completed, the target
benchmark may not be achieved. Nevertheless, the algorithm
keeps the reward already accumulated and earns additional
reward from partial execution of the plan.
The following lemma quantifies this combined contribution
under backward-greedy ordering.

\begin{lemma}
\label{lem:general-value-repair}
Under Assumption~\ref{assump:discrete-concavity},
for any current allocation $\boldsymbol s$ satisfying
$V(\boldsymbol s)\ge U(\tilde{n})$ and integers
$0\le\tilde{n}<n$, let
$\boldsymbol a(n)=(\boldsymbol x^*(n)-\boldsymbol s)^+$
be executed in backward-greedy order.
For every integer
$0\le h\le\|\boldsymbol a(n)\|_1$,
\begin{equation}
\label{eq:general-value-repair}
V(\boldsymbol s+\boldsymbol a(n)^{(h)})
\ge
V(\boldsymbol s)
+\frac{h^2}{n^2-\tilde{n}^2}
\bigl[U(n)-V(\boldsymbol s)\bigr].
\end{equation}
\end{lemma}

Lemma~\ref{lem:general-value-repair} combines the
quadratic prefix-value guarantee with the reward secured
before the current plan. It provides a lower bound on the
accumulated reward even when the plan is only partially executed.
The lemma will also be used in the randomized algorithm.
Its proof is provided in Appendix~\ref{app:dc-gen-plan}.

We now specify a deterministic schedule for pursuing these
benchmarks. Fix an integer geometric grid ratio $\gamma\ge2$,
which ensures that all benchmark levels and re-solving
checkpoints are integral. For each epoch $j=0,1,\ldots$, define
\(
n_j=(\gamma-1)\gamma^j,
T_j=\gamma^{j+1},
\)
with $T_{-1}=1$.
Here, $n_j$ is the target benchmark level and $T_j$ is the
re-solving checkpoint, measured by the cumulative number of
assignments. Since $T_j-T_{j-1}=n_j$, each epoch has enough
assignments to complete its plan if arrivals continue.

The algorithm assigns the first customer, if any, to a product
$i^\star\in\arg\max_{i:b_i>0}R_i(1)$,
thereby achieving $U(1)$.
Let $\boldsymbol s(\ell)$ denote the allocation after
$\ell$ assignments, for $\ell=0,\ldots,B$.
At the beginning of each active epoch $j$, with current
allocation $\boldsymbol s=\boldsymbol s(T_{j-1})$,
the algorithm constructs the additional plan
$\boldsymbol a(n_j)$ and executes it in backward-greedy order. Since
$\|\boldsymbol a(n_j)\|_1\le n_j=T_j-T_{j-1}$,
the plan can be completed within the epoch.
If it is completed before reaching $T_j$,
subsequent customers are assigned to products with remaining
capacity. Consequently, whenever the algorithm reaches
checkpoint $T_j$, it has secured
\(
V(\boldsymbol s(T_j))\ge U(n_j).
\)
The procedure is summarized in Algorithm~\ref{alg:general-det}. 
\begin{algorithm}[ht!]
\caption{Deterministic geometric allocation}
\label{alg:general-det}
\begin{algorithmic}[1]
\Require Capacities $\boldsymbol b$, reward functions $(R_i)_{i\in[m]}$, and integer grid ratio $\gamma\ge2$.
\State Set $n_j=(\gamma-1)\gamma^j$ and $T_j=\gamma^{j+1}$ for $j=0,1,\ldots$, with $T_{-1}=1$.
\State Choose $i^\star\in\argmax_{i:b_i>0}R_i(1)$.
\State Initialize the allocation path by assigning the first customer, if any, to $i^\star$; otherwise terminate.
\For{$j=0,1,\ldots$ until capacity is exhausted or arrivals stop}
    \State Set $\boldsymbol s\gets\boldsymbol s(T_{j-1})$.
    \State Compute an offline optimizer $\boldsymbol x^*(n_j)$ attaining $U(n_j)$.
    \State Set $\boldsymbol a(n_j)\gets(\boldsymbol x^*(n_j)-\boldsymbol s)^+$.
    \State Execute $\boldsymbol a(n_j)$ in backward-greedy order until completion or $\min\{T_j,B\}$.
    \State If completed early, assign remaining customers in the epoch to available products.
\EndFor
\end{algorithmic}
\end{algorithm}

\paragraph{Performance guarantee.} Let $\pi(\gamma)$ denote Algorithm~\ref{alg:general-det} with an integer grid ratio $\gamma\ge2$. Consider a stopping count $\ell$ within epoch $j\ge1$, where $T_{j-1}<\ell\le T_j$. To characterize the reward at $\ell$, we balance two sources: the benchmark $U(n_{j-1})$ secured before epoch $j$ and the progress toward $U(n_j)$ during the epoch. As arrivals continue, the preceding benchmark becomes a smaller fraction of the realized offline reward, while additional assignments increase the guaranteed progress toward the current benchmark. We compare both contributions with \(U(\ell)\) to obtain a performance certificate and minimize it over stopping positions.

Let $t=\ell/T_{j-1}\in(1,\gamma]$ denote the relative stopping
position in an epoch. For a given $\gamma$, define the following performance certificate on $[1,\gamma]$:
\begin{equation*}
\eta_D(\gamma,t)
:=\frac{(\gamma-1)^2}{\gamma^2t^2}+\min\left\{1,\frac{\gamma^2(t-1)^2}{(\gamma-1)^2(\gamma^2-1)}\right\}\left[\min\left\{1,\frac{(\gamma-1)^2}{t^2}\right\}-\frac{(\gamma-1)^2}{\gamma^2t^2}\right].
\end{equation*}
The two terms capture the contributions of the previous
and current benchmarks relative to $U(\ell)$. We show that minimizing the certificate $\eta_D(\gamma,t)$ over $t$ yields a constant competitive ratio.

\begin{theorem}[Deterministic competitive ratio under general arrivals]
\label{thm:general-d-cr}
Under Assumption~\ref{assump:discrete-concavity}, for any integer grid ratio $\gamma\ge2$, define
\(
\eta_D(\gamma):=\min_{1\le t\le\gamma}\eta_D(\gamma,t)\).
Then, for every integer stopping count $\ell\ge0$, $\Rev^{\pi(\gamma)}(\ell)\ge\eta_D(\gamma)U(\ell)$.
Consequently, $\Rev^{\pi(\gamma)}\ge\eta_D(\gamma)\OPT_{\mathrm{off}}$ for any arrival distribution. In particular, choosing $\gamma=3$ yields $\eta_D(3)>0.2537$, and hence $\pi(3)$ is a deterministic $0.2537$-competitive algorithm.
\end{theorem}

We now provide a sketch of the proof, with full proof in Appendix~\ref{app:dc-gen-d}. 
Since the policy is deterministic, $\Rev^{\pi(\gamma)}(\ell)=V(\boldsymbol s(\ell))$ for $\ell=0,\ldots,B$. Consider arrivals stopping during epoch $j\ge1$, with $T_{j-1}<\ell\le\min\{T_j,B\}$. We have
\begin{equation}
\label{eq:general-det-reward}
\begin{aligned}
V(\boldsymbol s(\ell))
&\ge V(\boldsymbol s(T_{j-1}))
+\min\left\{1,\frac{(\ell-T_{j-1})^2}{n_j^2-n_{j-1}^2}\right\}
[U(n_j)-V(\boldsymbol s(T_{j-1}))]\\
&\ge U(n_{j-1})
+\min\left\{1,\frac{(\ell-T_{j-1})^2}{n_j^2-n_{j-1}^2}\right\}
[U(n_j)-U(n_{j-1})],
\end{aligned}
\end{equation}
where the first inequality follows from Lemma~\ref{lem:general-value-repair} during partial execution, extended beyond plan completion by monotonicity of accumulated reward. The second inequality follows from the benchmark
$V(\boldsymbol s(T_{j-1}))\ge U(n_{j-1})$
secured in the preceding epoch.

Next, we compare both benchmarks with the offline reward
$U(\ell)$.
By \eqref{eq:offline-value-growth} and monotonicity,
\[
U(n_{j-1})
\ge
\frac{(\gamma-1)^2}{\gamma^2t^2}U(\ell),
\qquad
U(n_j)
\ge
\min\left\{1,\frac{(\gamma-1)^2}{t^2}\right\}
U(\ell),
\]
where we use
$n_{j-1}=(\gamma-1)T_{j-1}/\gamma$,
$n_j=(\gamma-1)T_{j-1}$,
and $\ell=tT_{j-1}$.
Substituting these bounds into
\eqref{eq:general-det-reward} gives
\begin{equation}
\label{eq:general-det-score}
\Rev^{\pi(\gamma)}(\ell)=V(\boldsymbol s(\ell))\ge\eta_D(\gamma,t)U(\ell).
\end{equation}
The initial epoch and boundary cases are handled separately.
Minimizing over $t\in[1,\gamma]$ establishes the pointwise
guarantee, and taking expectations over $L$ gives the
competitive ratio.

Theorem~\ref{thm:general-d-cr} establishes a constant guarantee
at every stopping count, without knowledge of the arrival
distribution. We next improve this guarantee through fractional
allocation planning and a randomly shifted geometric grid.


\subsubsection{A Randomized Algorithm.}
\label{sec:dc-general-random}
The deterministic schedule has two sources of inefficiency: it may delay re-solving unnecessarily, and its guarantee depends on where arrivals stop within an epoch.
First, it reserves $n_j$ assignments for epoch $j$, which is enough to
complete the allocation plan, but full completion may be unnecessary to achieve its target benchmark. In particular, previous assignments outside $\boldsymbol x^*(n_j)$ still contribute to the accumulated reward, potentially allowing $U(n_j)$ to be reached through partial execution of $\boldsymbol a(n_j)$. Indeed, Lemma~\ref{lem:general-value-repair} implies that, once the preceding epoch has secured $U(n_{j-1})$, at most $\lceil\sqrt{n_j^2-n_{j-1}^2}\rceil$ additional assignments suffice to attain $U(n_j)$, which is potentially fewer than the $n_j$ assignments reserved. Re-solving can therefore occur
earlier than the deterministic schedule allows.

Second, the geometric grid makes the performance guarantee sensitive to the stopping position. Theorem~\ref{thm:general-d-cr} minimizes the performance certificate $\eta_D(\gamma,t)$ over relative stopping positions $t\in[1,\gamma]$, which can be conservative because the guarantee varies within an epoch. At the beginning of an epoch, the preceding benchmark provides a relatively strong guarantee; at the end, the current benchmark is achieved. At intermediate positions, however, the preceding benchmark has weakened relative to the realized offline reward, while progress toward the current benchmark may be insufficient to compensate. For example, when $\gamma=3$, the endpoints satisfy $\eta_D(3,1)=\eta_D(3,3)\approx0.4444$, whereas the minimum $\eta_D(3)\approx0.2538$ occurs at $t\approx1.5178$. Thus, part of the performance loss arises from unfavorable alignment between the realized arrival count and the fixed geometric grid.

We address these limitations by shortening the re-solving
epochs and randomly shifting the geometric grid. The resulting
benchmark levels and checkpoints need not be integers.
We therefore construct a fractional allocation path that
follows this schedule and then jointly round it into a
feasible integral assignment sequence.



\paragraph{Randomly shifted geometric schedule.}
Fix a grid ratio $\gamma>1$ and draw a random shift
$Z\sim\mathrm{Unif}[0,1)$ independently of $L$.
For each $j\in\mathbb Z$, define the prospective benchmark
level
\(
n_j=\gamma^{j+Z}.
\)
The benchmark levels remain geometrically spaced, but their
positions vary with the random shift.


We next shorten the re-solving checkpoints. From the fractional extension of Lemma~\ref{lem:general-value-repair}, once the preceding epoch has secured $U(n_{j-1})$, at most $\sqrt{n_j^2-n_{j-1}^2}$ additional units of fractional service suffice to attain $U(n_j)$. Thus we choose
the epoch length
\(
T_j-T_{j-1}=\sqrt{n_j^2-n_{j-1}^2},
\)
rather than the $n_j$ assignments reserved in the deterministic
algorithm. 

Let $\lambda(\gamma):=\sqrt{(\gamma+1)/(\gamma-1)}$. Summing the epoch lengths gives
\(
    T_j
    =\sum_{k=-\infty}^j\sqrt{n_k^2-n_{k-1}^2}=\sqrt{\frac{\gamma+1}{\gamma-1}}\,n_j=\lambda(\gamma)n_j\).
Hence, epoch $j$ begins at $T_{j-1}=\lambda(\gamma)n_j/\gamma$, targets the benchmark at level $n_j$, and ends at
$T_j=\lambda(\gamma)n_j$. We require $\gamma>\lambda(\gamma)$, ensuring that
$T_{j-1}<n_j$. Hence, the allocation at the beginning
of each epoch is feasible for the offline benchmark
at level $n_j$.

\paragraph{Fractional allocation planning.} 
Because the randomly shifted benchmark levels and re-solving
checkpoints are generally noninteger, we first construct an
allocation path over continuous service counts. We then
convert this path into an integral assignment sequence.

To define fractional benchmark levels, we extend each cumulative
reward function $R_i$ by linear interpolation between
consecutive integers. For any real $n\ge0$, define
\begin{equation}
\label{eq:offline-opt-frac}
U(n)=\max_{\boldsymbol{x}\in\mathbb R_{\ge0}^m} V(\boldsymbol{x})
\quad\text{s.t.}\quad
0\le x_i\le b_i,\quad i\in[m],\quad
\|\boldsymbol{x}\|_1\le n,
\end{equation}
where $V(\boldsymbol x)=\sum_{i\in[m]}R_i(x_i)$ uses the
interpolated reward functions.
Since $V$ is convex and the feasible polytope has integral
extreme points whenever $n$ is an integer, this extension
is the same as the original offline benchmark at integer
arrival counts.

The backward-greedy ordering and
Lemma~\ref{lem:general-value-repair} also extend to fractional
allocations. In particular, from a current allocation
$\boldsymbol s$ satisfying $V(\boldsymbol s)\ge U(n_{j-1})$,
the fractional plan
\(
\boldsymbol a(n_j)
=(\boldsymbol x^*(n_j)-\boldsymbol s)^+
\)
can attain $U(n_j)$ within
$\sqrt{n_j^2-n_{j-1}^2}$ additional units of fractional
service. 
The technical details of the fractional extension are provided in
Appendix~\ref{app:dc-gen-r}. We now describe the resulting fractional allocation path.

For a fixed shift $Z$, let $\boldsymbol s(u)$ denote the
fractional allocation after $u$ units of service, with
$\|\boldsymbol s(u)\|_1=u$ for $0\le u\le B$. Let $j_0$ be the smallest integer satisfying $T_{j_0}>1$,
so that $T_{j_0-1}\le1<T_{j_0}$.
The algorithm initially allocates the first $T_{j_0-1}$
fractional units to a product
$i^\star\in\arg\max_{i:b_i>0}R_i(1)$.
Since $U(n)=nU(1)$ for $0\le n\le1$, this initialization
secures
\(
V(\boldsymbol s(T_{j_0-1}))
=T_{j_0-1}U(1)\ge U(n_{j_0-1}).
\)

At the beginning of each epoch $j\ge j_0$, the algorithm
computes an optimizer $\boldsymbol x^*(n_j)$ of
\eqref{eq:offline-opt-frac} and constructs
$\boldsymbol a(n_j)
=(\boldsymbol x^*(n_j)-\boldsymbol s(T_{j-1}))^+$.
It executes this plan in fractional backward-greedy order,
treating constant-marginal segments of the interpolated
reward functions as divisible allocation units.
If the plan finishes before $T_j$, the remaining service
is assigned to products with available capacity.
The procedure continues until all capacities are exhausted.



\paragraph{Joint rounding.}

Finally, we convert the fractional allocation sequence into a feasible integral assignment sequence using the bipartite dependent rounding procedure \citep{gandhi2006dependent}. This technique has also been applied to assortment optimization and assignment problems \citep{epstein2024selection, barre2026assortment}. For fixed shift $Z$, the fractional assignment of arrival position $\ell$ to product $i$ is $s_i(\ell)-s_i(\ell-1)$ for each integer 
$\ell=1,\ldots,B$. These amounts sum to one
across products, and their totals over all $B$
positions equal the product capacities. The procedure jointly rounds these assignments, ensuring one assignment at each position and respecting product capacities, while preserving each assignment in expectation. Let $\widehat{\boldsymbol s}(\ell)$ denote the resulting integral allocation after the first $\ell$ arrivals. Then, for every product $i$ and count $\ell$, we have $\E[\widehat s_i(\ell)\mid Z]=s_i(\ell)$. By convexity of $V$, we have $\E[V(\widehat{\boldsymbol s}(\ell))\mid Z]\ge V(\boldsymbol s(\ell))$.
Thus, joint rounding preserves the fractional path's reward guarantee in expectation and the integral assignment sequence can be computed before arrivals. The complete algorithm is summarized in Algorithm~\ref{alg:general-repair-rounding}, with the rounding construction and its proof provided in Appendix~\ref{app:dc-gen-r}.
\begin{algorithm}[ht!]
\caption{Randomly shifted geometric allocation}
\label{alg:general-repair-rounding}
\begin{algorithmic}[1]
\Require Capacities $\boldsymbol b$, reward functions $(R_i)_{i\in[m]}$, and real grid ratio $\gamma>1$ satisfying $\gamma>\lambda(\gamma)$.
\State Draw $Z\sim\mathrm{Unif}[0,1)$; set $n_j=\gamma^{j+Z}$ and $T_j=\lambda(\gamma)n_j$ for $j\in\mathbb Z$.
\State Let $j_0=\min\{j:T_j>1\}$ and choose $i^\star\in\argmax_{i:b_i>0}R_i(1)$.
\State Initialize the fractional allocation path by assigning the first $T_{j_0-1}$ units to $i^\star$.
\For{$j=j_0,j_0+1,\ldots$ until capacity is exhausted}
    \State Set $\boldsymbol s\gets\boldsymbol s(T_{j-1})$.
    \State Compute a fractional offline optimizer $\boldsymbol x^*(n_j)$ attaining $U(n_j)$.
    \State Set $\boldsymbol a(n_j)\gets(\boldsymbol x^*(n_j)-\boldsymbol s)^+$.
    \State Execute $\boldsymbol a(n_j)$ in fractional backward-greedy order until completion or $\min\{T_j,B\}$.
    \State If completed early, allocate the remaining service to available products.
\EndFor
\State Jointly round the fractional path to obtain an integral assignment sequence $\widehat{\boldsymbol s}(0),\ldots,\widehat{\boldsymbol s}(B)$.
\State Upon arrival $\ell\le B$, assign the customer according to $\widehat{\boldsymbol s}(\ell)-\widehat{\boldsymbol s}(\ell-1)$; reject subsequent arrivals.
\end{algorithmic}
\end{algorithm}

\paragraph{Performance guarantee.}
Let $\pi_R(\gamma)$ denote Algorithm~\ref{alg:general-repair-rounding}
with a grid ratio $\gamma>\lambda(\gamma)$.
For a fixed shift $Z$, consider a stopping count $\ell$
during epoch $j$, where $T_{j-1}<\ell\le T_j$,
and let $t=\ell/T_{j-1}\in(1,\gamma]$ denote its
relative position within the epoch.

Under the shortened schedule, the fraction of the epoch
completed is
\(
\frac{\ell-T_{j-1}}{T_j-T_{j-1}}
=\frac{t-1}{\gamma-1}.
\)
The corresponding progress coefficient is
\(
\left(\frac{t-1}{\gamma-1}\right)^2
\). As in the deterministic case, we define the performance certificate: 
\begin{equation*}
\eta_R(\gamma,t):=\frac{1}{\lambda(\gamma)^2t^2}+\frac{(t-1)^2}{(\gamma-1)^2}\left[\min\left\{1,\frac{\gamma^2}{\lambda(\gamma)^2t^2}\right\}-\frac{1}{\lambda(\gamma)^2t^2}\right].
\end{equation*}
As in the deterministic case, the two terms capture the preceding-benchmark guarantee and additional progress toward the current benchmark. The coefficients differ due to the shorter fractional epochs, while the random shift determines the relative stopping position $t$. 
For any fixed positive stopping count $\ell$, the uniform
random shift makes $\log_\gamma t$ uniform on $(0,1]$.
Thus, $t$ has density $1/(t\log\gamma)$ on $(1,\gamma]$.
Averaging the certificate with respect to this distribution gives
the following guarantee.
\begin{theorem}[Randomized competitive ratio under general arrivals]
\label{thm:general-r-cr}
Under Assumption~\ref{assump:discrete-concavity}, for any  $\gamma>1$ satisfying $\gamma>\lambda(\gamma)$, define
\(
\eta_R(\gamma):=\frac{1}{\log\gamma}\int_1^\gamma\eta_R(\gamma,t)\frac{dt}{t}\).
Then, for every integer stopping count $\ell\ge0$, $\Rev^{\pi_R(\gamma)}(\ell)\ge\eta_R(\gamma)U(\ell)$, where the expectation is over the random shift and joint rounding. Consequently, $\Rev^{\pi_R(\gamma)}\ge\eta_R(\gamma)\OPT_{\mathrm{off}}$ for any arrival distribution. In particular, choosing $\gamma=22/5$ yields $\eta_R(22/5)>0.3864$, and hence $\pi_R(22/5)$ is a randomized $0.3864$-competitive algorithm.
\end{theorem}
We now describe the high-level idea, with the full proof provided in Appendix~\ref{app:dc-gen-r}. The argument is similar to that for the deterministic algorithm. 

Fix the shift $Z$ and consider a stopping count $\ell$
during an epoch $j\ge j_0$, with
$T_{j-1}<\ell\le\min\{T_j,B\}$.
Let $t=\ell/T_{j-1}$. 
By the fractional extension of
Lemma~\ref{lem:general-value-repair},
the accumulated reward satisfies
\begin{equation*}
\begin{aligned}
V(\boldsymbol s(\ell))
&\ge V(\boldsymbol s(T_{j-1}))
+\frac{(t-1)^2}{(\gamma-1)^2}
[U(n_j)-V(\boldsymbol s(T_{j-1}))]\\
&\ge U(n_{j-1})
+\frac{(t-1)^2}{(\gamma-1)^2}
[U(n_j)-U(n_{j-1})],
\end{aligned}
\end{equation*}
where the first inequality follows from the fractional extension of Lemma~\ref{lem:general-value-repair}, extended beyond plan completion by monotonicity of accumulated reward. The second follows from $V(\boldsymbol s(T_{j-1}))\ge U(n_{j-1})$, guaranteed by the preceding epoch, and the fact that $\frac{(t-1)^2}{(\gamma-1)^2}$ lies in $[0,1]$.

To compare the two benchmarks with $U(\ell)$,
we apply the fractional extension of
\eqref{eq:offline-value-growth} and monotonicity.
Using
$n_{j-1}=T_{j-1}/\lambda(\gamma)$,
$n_j=\gamma T_{j-1}/\lambda(\gamma)$,
and $\ell=tT_{j-1}$ gives
\(
U(n_{j-1})
\ge\frac{1}{\lambda(\gamma)^2t^2}U(\ell)
\)
and
\(
U(n_j)
\ge
\min\left\{
1,\frac{\gamma^2}{\lambda(\gamma)^2t^2}
\right\}U(\ell).
\)
Substituting these bounds into the preceding inequality
yields
\begin{equation}
\label{eq:general-random-score}
V(\boldsymbol s(\ell))
\ge\eta_R(\gamma,t)U(\ell).
\end{equation}
We next average this guarantee over the random shift.
For any fixed integer stopping count $1\le\ell\le B$,
\(
\log_\gamma t
=\log_\gamma(\ell/\lambda(\gamma))-Z-(j-1),
\)
where the integer $j-1$ places $\log_\gamma t$
in $(0,1]$.
Since $Z$ is uniform on $[0,1)$,
$\log_\gamma t$ is uniform on $(0,1]$.
Consequently,
\[
\begin{aligned}
\E_Z[V(\boldsymbol s(\ell))]\ge
\E_Z[\eta_R(\gamma,t)]U(\ell)=
\frac{1}{\log\gamma}
\int_1^\gamma
\eta_R(\gamma,t)\frac{dt}{t}\,U(\ell)=\eta_R(\gamma)U(\ell).
\end{aligned}
\]

Finally, joint rounding preserves the fractional
allocation in expectation.
By convexity of $V$,
\(
\E[V(\widehat{\boldsymbol s}(\ell))\mid Z]
\ge V(\boldsymbol s(\ell)).
\)
Taking expectations over $Z$ therefore gives
\(
\Rev^{\pi_R(\gamma)}(\ell)
\ge\eta_R(\gamma)U(\ell).
\)
The initialization and boundary cases are handled in
Appendix~\ref{app:dc-gen-r}. Taking expectations over $L$ yields $\Rev^{\pi_R(\gamma)}\ge\eta_R(\gamma)\OPT_{\mathrm{off}}$. 

Note that 
Theorem~\ref{thm:general-r-cr} strengthens the distribution-oblivious guarantee by combining earlier transitions between benchmarks with reduced sensitivity to stopping positions. This improvement requires no knowledge of the arrival distribution: the guarantee holds for every fixed stopping count, in expectation over the algorithm's internal randomization.

To complement these guarantees, we establish an upper bound of $0.5703$ on the competitive ratio under general arrivals. The construction combines power-law arrivals with products of varying capacities and linear bonuses. The power-law tail decays at a similar relative rate across geometrically spaced arrival counts, sustaining uncertainty over multiple capacity scales. This makes it difficult for an online policy to prioritize a particular scale, as substantial future buildup opportunities remain relevant even after earlier arrival levels are reached. We choose the linear bonus parameters to balance the rewards at different capacity scales against their probabilities of realization, ensuring that no single allocation target dominates. The construction and proof are provided in Appendix~\ref{app:dc-gen-upper}; closing the gap between the upper and lower bounds remains open.

\subsection{A \texorpdfstring{$1/2$}{1/2}-Competitive Algorithm under Increasing Failure Rate (IFR) Arrivals}
\label{sec:dc-ifr}


The geometric schedule in Section~\ref{sec:dc-gen} protects against every stopping count without using the arrival distribution $\mathcal{D}_L$. We now exploit distributional information under IFR arrivals, for which the conditional probability of another
arrival is nonincreasing as more customers arrive. 


Our algorithm retains backward-greedy ordering but uses the
conditional arrival distribution to select each allocation plan.
The key idea is to maintain a potential equal to a fraction
of the expected residual offline reward. At each epoch, we
choose a plan whose expected reward compensates for the expected
decrease in this potential. These inequalities telescope across
epochs, yielding a competitive ratio of $1/2$. To prove the guarantee, we derive a new bound on residual offline
reward and combine it with
the quadratic prefix-value guarantee and the IFR property.

\subsubsection{Allocation Planning.}
\label{sec:dc-ifr-plan}

We specify how the algorithm selects an allocation plan at each re-solving epoch introduced in Section~\ref{sec:dc-re}. 

A new epoch begins with the first customer or when another
customer arrives after the preceding plan has been completed.
Consider an active epoch with current allocation
$\boldsymbol s$, and let
$h_s=\|\boldsymbol s\|_1$ and $N=B-h_s$ denote
the number of previous assignments and remaining capacity,
respectively. 
Since the epoch begins upon the arrival of customer $h_s+1$,
we condition on $L\ge h_s+1$. Define $K:=\min\{L,B\}-h_s$ and $p_h:=\P(K=h\mid L\ge h_s+1)$ for $h=1,\ldots,N$. Thus, $K$ denotes the number of customers available for assignment from the current epoch onward, including the arrival that initiates the epoch, and the probabilities $p_h$ characterize its conditional distribution.

For each epoch, we use offline optimization to evaluate the additional reward attainable from $\boldsymbol s$. Specifically, for each planning horizon $k=0,\ldots,N$, define the \emph{residual offline value} $U_{\boldsymbol s}(k)$ as the maximum additional reward obtainable by assigning at most $k$ customers starting from $\boldsymbol s$:
\begin{equation*}
\begin{aligned}
U_{\boldsymbol s}(k)&:=\max_{\boldsymbol a\in\mathbb Z_{\ge0}^m}\bigl[V(\boldsymbol s+\boldsymbol a)-V(\boldsymbol s)\bigr]\quad\text{s.t.}\quad
0\le a_i\le b_i-s_i,\quad i\in[m],\quad
\|\boldsymbol a\|_1\le k.
\end{aligned}
\end{equation*}
In particular, $U_{\boldsymbol s}(0)=0$ and
$U_{\boldsymbol 0}(k)=U(k)$.
Because marginal rewards are nonnegative, an optimal
allocation using exactly $k$ additional assignments
exists for every $k=1,\ldots,N$.
Let $\boldsymbol a^*(k)$ denote such an allocation,
which we call the candidate plan for horizon $k$. Once a horizon is selected, the algorithm executes the corresponding candidate plan in backward-greedy order and initiates a new epoch if another customer arrives after its completion.

To select among candidate horizons, fix a target competitive
ratio $\rho\in(0,1]$.
At the beginning of the current epoch, define the potential
\begin{equation*}
\Phi(\boldsymbol s)
:=\rho\,\E[U_{\boldsymbol s}(K)\mid L\ge h_s+1]
=\rho\sum_{h=1}^{N}p_hU_{\boldsymbol s}(h).
\end{equation*}
The potential represents the fraction of expected residual
offline reward that the algorithm seeks to secure. Our planning rule requires the reward earned during the
current epoch to cover the expected decrease in this
potential. If another customer arrives after the current
plan is completed, the algorithm re-solves from the
resulting allocation, and the same requirement applies
to the next epoch.


Consider a candidate horizon $k\in\{1,\ldots,N\}$.
The algorithm executes $\boldsymbol a^*(k)$ in
backward-greedy order and re-solves if another customer
arrives after the plan is completed. Let $\boldsymbol a^*(k)^{(j)}$ denote the first $j$
assignments in this order, and let
$\boldsymbol s_k=\boldsymbol s+\boldsymbol a^*(k)$
be the allocation upon plan completion.

If $K=h$, the algorithm executes the first
$\min\{h,k\}$ assignments and earns
\(
G_{\boldsymbol s}(k,h)
:=
V\bigl(\boldsymbol s+
\boldsymbol a^*(k)^{(\min\{h,k\})}\bigr)
-V(\boldsymbol s).
\)
Thus, the expected reward earned during the current epoch is
\(
\sum_{h=1}^{N}p_hG_{\boldsymbol s}(k,h).
\) A subsequent epoch begins from $\boldsymbol s_k$
only if $K>k$. Its expected potential, accounting for
the probability of reaching that epoch, is
\[
\P(K>k\mid L\ge h_s+1)\Phi(\boldsymbol s_k)
=
\rho\sum_{h=k+1}^{N}p_hU_{\boldsymbol s_k}(h-k).
\]

Therefore, requiring the expected reward earned during the current epoch to cover the expected decrease in potential is equivalent to selecting a horizon $k$ satisfying
\begin{equation}
\label{eq:ifr-potential-balance}
\sum_{h=1}^{N}p_hG_{\boldsymbol s}(k,h)+\rho\sum_{h=k+1}^{N}p_hU_{\boldsymbol s_k}(h-k)
\ge\Phi(\boldsymbol s)=\rho\sum_{h=1}^{N}p_hU_{\boldsymbol s}(h).
\end{equation}
For each realization $K=h$, the plan earns $G_{\boldsymbol s}(k,h)$ and leaves an optimal continuation worth $U_{\boldsymbol s_k}(h-k)$ if $h>k$. To identify a horizon satisfying
\eqref{eq:ifr-potential-balance}, we derive lower bounds
on the reward earned during the current plan and the
residual offline value after its completion.



The reward earned during the current epoch is controlled by backward-greedy ordering. By Lemma~\ref{lemma:quad-prefix}, for every $h=1,\ldots,N$, we have
\begin{equation}
\label{eq:ifr-epoch-reward-bound}
G_{\boldsymbol s}(k,h)
\ge\frac{\min\{h,k\}(\min\{h,k\}+1)}{k(k+1)}
U_{\boldsymbol s}(k).
\end{equation}
This bound accounts for the possibility that arrivals stop before the selected plan is completed.
The residual offline value is more challenging to bound, since optimal allocations for different horizons need not be nested. Consequently, completing $\boldsymbol a^*(k)$ may prevent the full realization of the additional offline value $U_{\boldsymbol s}(h)-U_{\boldsymbol s}(k)$ for $h>k$. Nevertheless, the following lemma shows that a guaranteed fraction of this additional value remains attainable from $\boldsymbol s_k$.
\begin{lemma}[Residual-value preservation]
\label{lem:ifr-residual-value}
Under Assumption~\ref{assump:discrete-concavity}, for any current allocation $\boldsymbol s$ and integers $1\le k<h\le N$, the allocation $\boldsymbol s_k$ obtained by completing candidate plan $\boldsymbol{a}^*(k)$ satisfies
\[
U_{\boldsymbol s_k}(h-k)\ge\frac{h-k+1}{h+k+1}\bigl[U_{\boldsymbol s}(h)-U_{\boldsymbol s}(k)\bigr].
\]
\end{lemma}

The proof compares the optimal candidate plans $\boldsymbol a^*(k)$ and $\boldsymbol a^*(h)$. Since these plans need not be nested, completing $\boldsymbol a^*(k)$ may limit the reward attainable from $\boldsymbol a^*(h)$ with the remaining $h-k$ assignments. However, this loss is partially offset by the optimality of $\boldsymbol a^*(k)$, which limits the additional value that $\boldsymbol a^*(h)$ can offer. Combining this optimality comparison with Lemma~\ref{lemma:quad-prefix}, we show that a guaranteed fraction of the offline value gap remains attainable.

Using the bounds in \eqref{eq:ifr-epoch-reward-bound} and Lemma~\ref{lem:ifr-residual-value}, define the planning score
\begin{equation*}
\begin{aligned}
\Psi_{\boldsymbol s}^{\rho}(k)
:={}&
U_{\boldsymbol s}(k)
\sum_{h=1}^{N}p_h
\frac{\min\{h,k\}(\min\{h,k\}+1)}{k(k+1)}+\rho\sum_{h=k+1}^{N}p_h
\frac{h-k+1}{h+k+1}
\bigl[U_{\boldsymbol s}(h)-U_{\boldsymbol s}(k)\bigr].
\end{aligned}
\end{equation*}
The first term lower bounds the expected reward earned
during the current epoch, and the second lower bounds
the expected potential carried into the next epoch. Consequently, for every candidate horizon $k$,
\begin{equation}
\label{eq:ifr-score-certificate}
\sum_{h=1}^{N}p_hG_{\boldsymbol s}(k,h)
+
\rho\sum_{h=k+1}^{N}p_hU_{\boldsymbol s_k}(h-k)
\ge
\Psi_{\boldsymbol s}^{\rho}(k).
\end{equation}
Thus, any horizon satisfying
$\Psi_{\boldsymbol s}^{\rho}(k)\ge\Phi(\boldsymbol s)$
also satisfies the planning condition
\eqref{eq:ifr-potential-balance}.

The algorithm selects
\(
k^*\in\arg\max_{1\le k\le N}
\Psi_{\boldsymbol s}^{\rho}(k),
\)
using a fixed deterministic tie-breaking rule.
It executes $\boldsymbol a^*(k^*)$ in
backward-greedy order and re-solves if another
customer arrives after plan completion.
We show next that IFR arrivals guarantee
$\Psi_{\boldsymbol s}^{\rho}(k^*)\ge\Phi(\boldsymbol s)$
for every active epoch when $\rho\le1/2$.

Algorithm~\ref{alg:ifr} summarizes the procedure.

\begin{algorithm}[ht!]
\caption{Allocation planning under IFR arrivals}
\label{alg:ifr}
\begin{algorithmic}[1]
\Require Capacities $\boldsymbol b$, reward functions $(R_i)_{i\in[m]}$, arrival distribution $\mathcal D_L$, and $\rho\in(0,1/2]$.
\State Initialize $\boldsymbol s\gets\boldsymbol 0$.
\While{capacity remains and a new customer arrives}
    \State Set $h_s\gets\|\boldsymbol s\|_1$ and $N\gets B-h_s$.
    \State Compute $U_{\boldsymbol s}(k)$ and $\boldsymbol a^*(k)$ for $k=1,\ldots,N$.
    \State Select $k^*\in\argmax_{1\le k\le N}\Psi_{\boldsymbol s}^{\rho}(k)$.
    \State Execute $\boldsymbol a^*(k^*)$ in backward-greedy order, updating $\boldsymbol s$ after each assignment.
\EndWhile
\end{algorithmic}
\end{algorithm}

\subsubsection{Performance Guarantee.}
\label{sec:dc-ifr-cr}

We now show that, under IFR arrivals, at least one planning horizon has a score at least as large as the current potential. This is the key step where the IFR property is used to establish the competitive ratio of $1/2$.

\begin{lemma}[Planning-score guarantee under IFR arrivals]
\label{lem:ifr-scalar}
Under Assumption~\ref{assump:discrete-concavity} and IFR arrivals, at every active epoch with current allocation $\boldsymbol s$ and for every $\rho\in(0,1/2]$,
\[
\max_{1\le k\le N}\Psi_{\boldsymbol s}^{\rho}(k)
\ge\rho\,\E[U_{\boldsymbol s}(K)\mid L\ge h_s+1]
=\Phi(\boldsymbol s).
\]
\end{lemma}

The proof lower bounds the maximum planning score by its average over a randomly selected horizon. Let $\widetilde K$ be an independent copy of $K$ under the conditional law given $L\ge h_s+1$. Then the maximum score is at least $\E[\Psi_{\boldsymbol s}^{\rho}(\widetilde K)\mid L\ge h_s+1]$. Applying \eqref{eq:offline-value-growth} to the residual offline problem shows that $U_{\boldsymbol s}(h)/[h(h+1)]$ is nonincreasing in $h$, allowing the residual offline values to be expressed as a nonnegative combination of quadratic cutoff sequences. Since the difference between the average planning score and the potential is linear in these values, it suffices to show that the average score covers the reward target for each cutoff sequence. For the full quadratic sequence $u_h=h(h+1)$, let $\underline K=\min\{K,\widetilde K\}$. The difference between the average score and $\rho\E[u_K]$ simplifies to $\E[\underline K^2+(1-\rho)\underline K]-\rho(\E[K])^2$. The IFR property ensures $(\E[K])^2\le2\E[\underline K^2]$, making this difference nonnegative for $\rho\le1/2$. For other cutoffs, we apply the IFR argument when both counts are below the cutoff and use symmetry to balance the cases where only one exceeds it. These comparisons establish the score guarantee for $\rho\le1/2$, with details provided in Appendix~\ref{app:dc-ifr-cr}.

We now establish the performance guarantee by telescoping the potential inequalities.
\begin{theorem}[Competitive ratio under IFR arrivals]
\label{thm:ifr-half}
Under Assumption~\ref{assump:discrete-concavity} and IFR arrivals, Algorithm~\ref{alg:ifr} with any $\rho\in(0,1/2]$ achieves a $\rho$-competitive ratio. In particular, choosing $\rho=1/2$ yields a $1/2$-competitive algorithm.
\end{theorem}

\proof{Proof of Theorem~\ref{thm:ifr-half}:}
Fix $\rho\in(0,1/2]$, and let $\pi$ denote the policy in Algorithm~\ref{alg:ifr}. At every active epoch, the selected horizon $k^*$ satisfies
\begin{equation}
\label{eq:ifr-nonnegative-drift}
\sum_{h=1}^{N}p_hG_{\boldsymbol s}(k^*,h)+\rho\sum_{h=k^*+1}^{N}p_hU_{\boldsymbol s_{k^*}}(h-k^*)
\ge\Psi_{\boldsymbol s}^{\rho}(k^*)
\ge\Phi(\boldsymbol s),
\end{equation}
where the first inequality follows from \eqref{eq:ifr-score-certificate} and the second follows from Lemma~\ref{lem:ifr-scalar}.

If $\P(L\ge1)=0$, both the algorithm and offline benchmark earn zero reward. Otherwise, condition on $L\ge1$ and index successive active epochs by $j=1,2,\ldots$. Let $G_j$ denote the reward earned during epoch $j$ and $\Phi_j$ the potential at its beginning, with both set to zero if that epoch does not occur. Since each active epoch assigns at least one customer, at most $B$ epochs can occur, and $\Phi_{B+1}=0$.

By \eqref{eq:ifr-nonnegative-drift}, the expected reward earned during each active epoch covers the decrease from its initial potential to the expected potential of the next epoch. Conditioning on the allocation at the beginning of epoch $j$ and then taking expectations yields $\E[G_j+\Phi_{j+1}\mid L\ge1]\ge\E[\Phi_j\mid L\ge1]$. Summing over $j=1,\ldots,B$, the intermediate potentials telescope:
\begin{equation*}
    \begin{aligned}
    \E\left[\sum_{j=1}^{B}G_j\Bigm|L\ge1\right]
    &\ge\sum_{j=1}^{B}\E[\Phi_j-\Phi_{j+1}\mid L\ge1]
    =\Phi(\boldsymbol0)
    =\rho\,\E[U(L)\mid L\ge1],
    \end{aligned}
\end{equation*}
where the last equality follows from $U_{\boldsymbol0}(\min\{L,B\})=U(L)$. Multiplying by $\P(L\ge1)$ and considering the zero reward when $L=0$ yield $\Rev^\pi\ge\rho\,\OPT_{\mathrm{off}}$.
Taking $\rho=1/2$ completes the proof.\Halmos
\endproof
Theorem~\ref{thm:ifr-half} shows that known IFR arrivals support a higher competitive ratio by allowing plans to account for the likelihood of future buildup opportunities. This reveals that the difficulty of online allocation with increasing returns depends not only on the reward structure but also on the nature of arrival uncertainty: knowing how continuation probabilities evolve enables more effective decisions about which buildup opportunities to pursue and when to re-solve.

We also establish an upper bound of $0.6608$ on the best
competitive ratio under IFR arrivals. This bound is higher than the $0.5703$ upper bound for general arrivals, where the power-law construction does not satisfy IFR. Our IFR construction instead combines geometric arrivals with products of varying capacities and linear bonuses. Before the terminal cap, the geometric distribution has a constant continuation probability, so observing additional arrivals provides no new information about the likelihood of further arrivals. Consequently, an online policy cannot infer whether continued buildup in a larger-capacity product will pay off. We choose the linear bonus parameters so that filling each product yields an expected reward equal to its expected number of assignments, keeping competing capacity scales relevant. Thus, although IFR regularity enables stronger guarantees, arrival uncertainty still imposes a fundamental limit on achievable performance. The construction and proof are provided in Appendix~\ref{app:dc-ifr-upper}; closing the gap between the upper and lower bounds remains open.

\begin{remark}[General arrivals]
\label{rem:general-quarter}
Under Assumption~\ref{assump:discrete-concavity}, Algorithm~\ref{alg:ifr} achieves a deterministic $1/4$-competitive ratio for any known arrival distribution. While the $1/2$ guarantee in Lemma~\ref{lem:ifr-scalar} relies on an IFR moment inequality, the planning-score guarantee holds for general distributions when $\rho\le1/4$. The other parts of the analysis do not use IFR. Consequently, \eqref{eq:ifr-score-certificate} and the potential argument in Theorem~\ref{thm:ifr-half} remain valid with $\rho=1/4$. We note that this distribution-aware guarantee is weaker than the distribution-oblivious guarantees of $0.2537$ and $0.3864$ established in Section~\ref{sec:dc-gen}. We provide the formal result and proof in Appendix~\ref{app:dc-ifr-general}.
\end{remark}

\section{Conclusion}
\label{sec:conclusion}

In this work, we demonstrate that increasing returns fundamentally reshape online resource allocation by introducing a trade-off between immediate rewards and uncertain future buildup, even under full compatibility. Our results reveal that the difficulty of this trade-off depends critically on the structure of bonus functions. While homogeneous bonuses admit a tight constant guarantee, arbitrary heterogeneous bonuses can preclude guarantees independent of the number of products. Imposing discrete concavity on the bonus functions restores constant guarantees for heterogeneous bonuses by limiting reward backloading. We develop a re-solving framework integrating allocation planning and assignment ordering, achieving distribution-oblivious guarantees for general arrivals and stronger performance under IFR arrivals. These findings provide a systematic understanding of structural limits and algorithmic opportunities in online allocation with increasing returns.

Several questions remain open. First, the optimal competitive ratios under discrete concavity remain unresolved for both general and IFR arrivals. Our upper bounds show that a loss relative to hindsight is unavoidable, but do not match the algorithmic guarantees. Closing these gaps would sharpen our understanding of the limits imposed by uncertain arrivals.
Second, discrete concavity is sufficient to control losses from partial buildup but may not be necessary. Identifying broader reward structures could provide a more comprehensive characterization of when constant competitive ratios are achievable.
Finally, our full-compatibility setting isolates the new trade-off between immediate rewards and future buildup from the classical trade-off of preserving capacity for future compatible customers. Introducing compatibility constraints and heterogeneous customer types would bring these trade-offs together. Understanding their interaction could yield new structural insights and guide the design of effective algorithms for broader online allocation settings.


\newpage

\bibliographystyle{informs2014}
\bibliography{references}

@article{udwani2026optimality,
  title={Optimality of Nonadaptive Algorithms in Online Submodular Welfare Maximization with Stochastic Outcomes},
  author={Udwani, Rajan},
  journal={Operations Research},
  year={2026},
  publisher={INFORMS}
}

@article{gandhi2006dependent,
  title   = {Dependent Rounding and Its Applications
             to Approximation Algorithms},
  author  = {Gandhi, Rajiv and Khuller, Samir and
             Parthasarathy, Srinivasan and Srinivasan, Aravind},
  journal = {Journal of the ACM},
  volume  = {53},
  number  = {3},
  pages   = {324--360},
  year    = {2006},
  doi     = {10.1145/1147954.1147956}
}

@article{barre2026assortment,
  title   = {Assortment Optimization with Visibility Constraints},
  author  = {Barr{\'e}, Th{\'e}o and El Housni, Omar and
             Ibn Brahim, Marouane and Lodi, Andrea and
             Segev, Danny},
  journal = {Mathematical Programming},
  volume  = {216},
  number  = {1-2},
  pages   = {177--220},
  year    = {2026},
  doi     = {10.1007/s10107-025-02238-9}
}

@article{epstein2024selection,
  title   = {Selection and Ordering Policies for Hiring
             Pipelines via Linear Programming},
  author  = {Epstein, Boris and Ma, Will},
  journal = {Operations Research},
  volume  = {72},
  number  = {5},
  pages   = {2000--2013},
  year    = {2024},
  doi     = {10.1287/opre.2023.0061}
}

@article{bernstein2022general,
  title={General bounds for incremental maximization},
  author={Bernstein, Aaron and Disser, Yann and Gro{\ss}, Martin and Himburg, Sandra},
  journal={Mathematical Programming},
  volume={191},
  number={2},
  pages={953--979},
  year={2022},
  publisher={Springer}
}

@article{disser2023incremental,
  title={Incremental maximization via continuization},
  author={Disser, Yann and Klimm, Max and Schewior, Kevin and Weckbecker, David},
  journal={arXiv preprint arXiv:2305.01310},
  year={2023}
}

@book{shaked2007stochastic,
  title={Stochastic orders},
  author={Shaked, Moshe and Shanthikumar, J George},
  year={2007},
  publisher={Springer}
}

@inproceedings{KVV1990,
  author    = {Richard M. Karp and Umesh V. Vazirani and Vijay V. Vazirani},
  title     = {An Optimal Algorithm for On-line Bipartite Matching},
  booktitle = {Proceedings of the 22nd Annual ACM Symposium on Theory of Computing (STOC)},
  pages     = {352--358},
  year      = {1990},
  publisher = {ACM},
  doi       = {10.1145/100216.100262}
}

@article{MSVV2007,
  author  = {Aranyak Mehta and Amin Saberi and Umesh V. Vazirani and Vijay V. Vazirani},
  title   = {{AdWords} and Generalized Online Matching},
  journal = {Journal of the ACM},
  volume  = {54},
  number  = {5},
  pages   = {22:1--22:19},
  year    = {2007},
  doi     = {10.1145/1284320.1284321}
}

@article{Mehta2013,
  author    = {Aranyak Mehta},
  title     = {Online Matching and Ad Allocation},
  journal   = {Foundations and Trends in Theoretical Computer Science},
  volume    = {8},
  number    = {4},
  pages     = {265--368},
  year      = {2013},
  publisher = {Now Publishers},
  doi       = {10.1561/0400000057}
}

@article{vera2021online,
  title={Online allocation and pricing: Constant regret via bellman inequalities},
  author={Vera, Alberto and Banerjee, Siddhartha and Gurvich, Itai},
  journal={Operations Research},
  volume={69},
  number={3},
  pages={821--840},
  year={2021},
  publisher={INFORMS}
}

@inproceedings{FMMM2009,
  author    = {Jon Feldman and Aranyak Mehta and Vahab S. Mirrokni and S. Muthukrishnan},
  title     = {Online Stochastic Matching: Beating $1-1/e$},
  booktitle = {Proceedings of the 50th Annual IEEE Symposium on Foundations of Computer Science (FOCS)},
  pages     = {117--126},
  year      = {2009},
  publisher = {IEEE Computer Society},
  doi       = {10.1109/FOCS.2009.72}
}

@inproceedings{DevanurHayes2009,
  author    = {Nikhil R. Devanur and Thomas P. Hayes},
  title     = {The {AdWords} Problem: Online Keyword Matching with Budgeted Bidders under Random Permutations},
  booktitle = {Proceedings of the 10th ACM Conference on Electronic Commerce (EC)},
  pages     = {71--78},
  year      = {2009},
  publisher = {ACM},
  doi       = {10.1145/1566374.1566384}
}

@article{katz1985network,
  author    = {Michael L. Katz and Carl Shapiro},
  title     = {Network Externalities, Competition, and Compatibility},
  journal   = {The American Economic Review},
  volume    = {75},
  number    = {3},
  pages     = {424--440},
  year      = {1985},
  publisher = {American Economic Association}
}

@article{arrow1962economic,
  author  = {Kenneth J. Arrow},
  title   = {The Economic Implications of Learning by Doing},
  journal = {The Review of Economic Studies},
  volume  = {29},
  number  = {3},
  pages   = {155--173},
  year    = {1962},
  doi     = {10.2307/2295952}
}

@article{lin2024content,
  title={Content promotion for online content platforms with the diffusion effect},
  author={Lin, Yunduan and Wang, Mengxin and Zhang, Heng and Zhang, Renyu and Shen, Zuo-Jun Max},
  journal={Manufacturing \& Service Operations Management},
  volume={26},
  number={3},
  pages={1062--1081},
  year={2024},
  publisher={INFORMS}
}

@article{wang2017consumer,
  title={Consumer choice models with endogenous network effects},
  author={Wang, Ruxian and Wang, Zizhuo},
  journal={Management Science},
  volume={63},
  number={11},
  pages={3944--3960},
  year={2017},
  publisher={INFORMS}
}

@article{chen2021duopoly,
  title={Duopoly competition with network effects in discrete choice models},
  author={Chen, Ningyuan and Chen, Ying-Ju},
  journal={Operations Research},
  volume={69},
  number={2},
  pages={545--559},
  year={2021},
  publisher={INFORMS}
}

@article{hu2016liking,
  title={Liking and following and the newsvendor: Operations and marketing policies under social influence},
  author={Hu, Ming and Milner, Joseph and Wu, Jiahua},
  journal={Management Science},
  volume={62},
  number={3},
  pages={867--879},
  year={2016},
  publisher={INFORMS}
}

@article{candogan2012optimal,
  title={Optimal pricing in networks with externalities},
  author={Candogan, Ozan and Bimpikis, Kostas and Ozdaglar, Asuman},
  journal={Operations Research},
  volume={60},
  number={4},
  pages={883--905},
  year={2012},
  publisher={INFORMS}
}

@article{pallais2014inefficient,
  title={Inefficient hiring in entry-level labor markets},
  author={Pallais, Amanda},
  journal={American Economic Review},
  volume={104},
  number={11},
  pages={3565--3599},
  year={2014},
  publisher={American Economic Association, 2014 Broadway, Suite 305, Nashville, TN 37203}
}

@article{gray2009outsourcing,
  title={Outsourcing to a powerful contract manufacturer: The effect of learning-by-doing},
  author={Gray, John V and Tomlin, Brian and Roth, Aleda V},
  journal={Production and Operations Management},
  volume={18},
  number={5},
  pages={487--505},
  year={2009},
  publisher={Wiley Online Library}
}

@article{cachon2002competition,
  title={Competition and outsourcing with scale economies},
  author={Cachon, G{\'e}rard P and Harker, Patrick T},
  journal={Management Science},
  volume={48},
  number={10},
  pages={1314--1333},
  year={2002},
  publisher={INFORMS}
}

@misc{reuters2025comiccon,
  author       = {Danielle Broadway},
  title        = {Labubu Fans Dote over Ugly-Cute Doll Trending at {Comic-Con}},
  howpublished = {Reuters, \url{https://www.reuters.com/lifestyle/labubu-fans-dote-over-ugly-cute-doll-trending-comic-con-2025-07-27/}},
  year         = {2025},
  month        = jul,
  note         = {Accessed: 2026-06-30}
}

@inproceedings{kapralov2013online,
  title={Online submodular welfare maximization: Greedy is optimal},
  author={Kapralov, Michael and Post, Ian and Vondr{\'a}k, Jan},
  booktitle={Proceedings of the twenty-fourth annual ACM-SIAM symposium on Discrete algorithms},
  pages={1216--1225},
  year={2013},
  organization={SIAM}
}

@inproceedings{devanur2012online,
  title={Online matching with concave returns},
  author={Devanur, Nikhil R and Jain, Kamal},
  booktitle={Proceedings of the forty-fourth annual ACM symposium on Theory of computing},
  pages={137--144},
  year={2012}
}

@article{JasinKumar2012,
  author  = {Jasin, Stefanus and Kumar, Sunil},
  title   = {A Re-Solving Heuristic with Bounded Revenue Loss for Network Revenue Management with Customer Choice},
  journal = {Mathematics of Operations Research},
  volume  = {37},
  number  = {2},
  pages   = {313--345},
  year    = {2012}
}

@article{BumpensantiWang2020,
  author  = {Bumpensanti, Pornpawee and Wang, He},
  title   = {A Re-Solving Heuristic with Uniformly Bounded Loss for Network Revenue Management},
  journal = {Management Science},
  volume  = {66},
  number  = {7},
  pages   = {2993--3009},
  year    = {2020}
}

@inproceedings{AggarwalGoelKarandeMehta2011,
  author    = {Gagan Aggarwal and
               Gagan Goel and
               Chinmay Karande and
               Aranyak Mehta},
  title     = {Online Vertex-Weighted Bipartite Matching and Single-Bid Budgeted Allocations},
  booktitle = {Proceedings of the Twenty-Second Annual ACM-SIAM Symposium on Discrete Algorithms (SODA)},
  pages     = {1253--1264},
  year      = {2011},
  publisher = {SIAM}
}

@article{huang2024online,
  title={Online matching: A brief survey},
  author={Huang, Zhiyi and Tang, Zhihao Gavin and Wajc, David},
  journal={ACM SIGecom Exchanges},
  volume={22},
  number={1},
  pages={135--158},
  year={2024},
  publisher={ACM New York, NY, USA}
}

@misc{modernretail2025labubunomics,
  author       = {Julia Waldow and Allison Smith},
  title        = {The Rise of {Labubunomics}: How Merchants and Marketplaces Are Cashing in on the Viral {Pop Mart} Toy},
  howpublished = {Modern Retail, {\def\UrlBreaks{\do\/\do-\do.}\url{https://www.modernretail.co/marketing/the-rise-of-labubunomics-how-merchants-and-marketplaces-are-cashing-in-on-the-viral-pop-mart-toy/}}},
  year         = {2025},
  month        = jun,
  note         = {Accessed: 2026-06-30}
}

@inproceedings{patton2026online,
  title={Online Resource Allocation with Concave, Diminishing-Returns Objectives},
  author={Patton, Kalen},
  booktitle={Proceedings of the 2026 Annual ACM-SIAM Symposium on Discrete Algorithms (SODA)},
  pages={6727--6741},
  year={2026},
  organization={SIAM}
}

@inproceedings{blum2015online,
  title={Online allocation and pricing with economies of scale},
  author={Blum, Avrim and Mansour, Yishay and Yang, Liu},
  booktitle={International Conference on Web and Internet Economics},
  pages={159--172},
  year={2015},
  organization={Springer}
}

@inproceedings{zhang2024nonstationary,
author = {Zhang, Xiaoyue and Qin, Hanzhang and Chou, Mabel C.},
title = {Online resource allocation with non-stationary customers},
year = {2024},
publisher = {JMLR.org},
booktitle = {Proceedings of the 41st International Conference on Machine Learning},
articleno = {2467},
numpages = {31},
location = {Vienna, Austria},
series = {ICML'24}
}

@article{zhalechian2022online,
  title={Online resource allocation with personalized learning},
  author={Zhalechian, Mohammad and Keyvanshokooh, Esmaeil and Shi, Cong and Van Oyen, Mark P},
  journal={Operations Research},
  volume={70},
  number={4},
  pages={2138--2161},
  year={2022},
  publisher={Informs}
}

@article{klimashevskaia2024survey,
  title={A survey on popularity bias in recommender systems: A. Klimashevskaia et al.},
  author={Klimashevskaia, Anastasiia and Jannach, Dietmar and Elahi, Mehdi and Trattner, Christoph},
  journal={User Modeling and User-Adapted Interaction},
  volume={34},
  number={5},
  pages={1777--1834},
  year={2024},
  publisher={Springer}
}

@article{economides1996network,
  title={Network externalities, complementarities, and invitations to enter},
  author={Economides, Nicholas},
  journal={European Journal of Political Economy},
  volume={12},
  number={2},
  pages={211--233},
  year={1996},
  publisher={Elsevier}
}

@article{wright1936factors,
  title={Factors affecting the cost of airplanes},
  author={Wright, Theodore P},
  journal={Journal of the aeronautical sciences},
  volume={3},
  number={4},
  pages={122--128},
  year={1936}
}

@article{liu2026broadcast,
  title={When to broadcast? Inventory disclosure policies for online sales of limited inventory},
  author={Liu, Zibo and Chen, Shi and Moinzadeh, Kamran and Tan, Yong},
  journal={Information Systems Research},
  volume={37},
  number={1},
  pages={117--137},
  year={2026},
  publisher={INFORMS}
}
\clearpage

%
%
%







\newpage
\setcounter{page}{1}

\begin{APPENDICES}

\makeatletter
\renewcommand{\theHsection}{appendix.\Alph{section}}
\renewcommand{\theHsubsection}{appendix.\Alph{section}.\arabic{subsection}}
\renewcommand{\theHsubsubsection}{appendix.\Alph{section}.\arabic{subsection}.\arabic{subsubsection}}
\makeatother

\section{Supporting Material for Section \ref{sec:model}}
\label{app:benchmark}

\subsection{Dynamic Programming for the Offline Benchmark}
\label{app:benchmark-offline}

We give an exact dynamic program for computing the offline optimal value $U(\ell)$ in \eqref{eq:offline-opt}, together with a procedure for recovering an optimal allocation.

For each $i=0,\dots,m$ and $\ell=0,\dots,B$, let $F_i(\ell)$ denote the maximum total reward obtainable by assigning exactly $\ell$ customers among the first $i$ products. Formally,
\[
F_i(\ell):=\max_{\boldsymbol{x}\in\mathbb Z_{\ge0}^i}\sum_{j=1}^i R_j(x_j)
\quad\text{s.t.}\quad
0\le x_j\le b_j,\quad j\in[i],\quad
\sum_{j=1}^i x_j=\ell,
\]
where an infeasible state has value $-\infty$. We initialize $F_0(0)=0$ and $F_0(\ell)=-\infty$ for $\ell\ge1$. The following proposition characterizes the offline optimal value through a dynamic programming recursion.

\begin{proposition}[Exact offline dynamic program]
\label{prop:offline-dp}
For each $i=1,\dots,m$ and $\ell=0,\dots,B$,
\begin{equation}
    \label{eq:offline-dp}
    F_i(\ell)=\max_{s\in\{0,\ldots,\min\{b_i,\ell\}\}}
    \left\{F_{i-1}(\ell-s)+R_i(s)\right\}.
\end{equation}
For any realized arrival count $\ell$, the offline optimal value is given by $U(\ell)=F_m(\min\{\ell,B\})$. An optimal offline allocation can be recovered by storing the maximizing choices in the dynamic program.
\end{proposition}

\proof{Proof of Proposition~\ref{prop:offline-dp}:}
Consider the offline problem restricted to the first $i$ products and exactly $\ell$ assignments. If $s$ customers are assigned to product $i$, then the remaining $\ell-s$ customers must be assigned among the first $i-1$ products. The reward from product $i$ is $R_i(s)$, while the maximum reward obtainable from the remaining products is $F_{i-1}(\ell-s)$. Maximizing over all feasible choices of $s$ yields \eqref{eq:offline-dp}.

Since all marginal rewards are nonnegative, an optimal offline allocation can be chosen to assign exactly $\min\{\ell,B\}$ customers. Therefore, $U(\ell)=F_m(\min\{\ell,B\})$. An optimal allocation can be recovered by tracing the maximizing choices in \eqref{eq:offline-dp} backward from product $m$ to product $1$, determining the number of customers assigned to each product.

This completes the proof.
\Halmos
\endproof

We next analyze the computational complexity of the dynamic program \eqref{eq:offline-dp}. There are $O(m(B+1))$ states, and each state requires evaluating at most $B+1$ possible assignments to the current product. Thus, the dynamic program runs in $O(m(B+1)^2)$ time. By storing the maximizing choice in each state, an optimal allocation for any given arrival count can be recovered in $O(m)$ time. Recovering optimal allocations for all allocation levels $\ell=0,\dots,B$ requires an additional $O(m(B+1))$ time. Hence, the total time complexity remains $O(m(B+1)^2)$, with space complexity $O(m(B+1))$.

\section{Supporting Material for Section~\ref{sec:limit}}
\label{app:limit}

\subsection{Supporting Material for Section~\ref{sec:limit-fail}}
\label{app:limit-fail}
We first establish the counterexamples showing greedy and LP-based rounding could be arbitrarily bad. 
\proof{Proof of Proposition~\ref{prop:greedy-fail}:}
Consider two products with capacities $b_1=1$ and $b_2=2$,
base rewards $r_1=1$ and $r_2=0$, and a common bonus function
$f(1)=0$ and $f(2)=M$.
The first customer arrives with certainty, and a second customer
arrives with probability $p=1/2$. 
Greedy assigns the first customer to product 1, which offers a
higher immediate reward. If a second customer arrives, it is
assigned to product 2 and earns zero reward. Thus, greedy earns
a total reward of one. Consider instead the feasible online policy that assigns both
customers to product 2 whenever they arrive. Its expected reward
is $pM=M/2$. Consequently, the expected reward of greedy is at most
$2/M$ times the optimal online expected reward. Since $M$ can be
arbitrarily large, the result follows. \Halmos
\endproof

\proof{Proof of Proposition~\ref{prop:lp-fail}:}
Consider two products with capacities $b_1=1$ and $b_2=2$,
base rewards $r_1=1$ and $r_2=0$, and a common bonus function
$f(1)=0$ and $f(2)=3$.
The first customer arrives with certainty, and a second customer
arrives with probability $p\in(0,1)$.
Thus, $\mu_L=1+p$.

For this instance, the mean-demand LP~\eqref{eq:lp} becomes
\[
\begin{aligned}
\max\quad & y_{1,1}+3y_{2,2}\\
\text{s.t.}\quad
&y_{1,1}+y_{2,1}+y_{2,2}\le 1+p,\\
&0\le y_{2,2}\le y_{2,1}\le1,\qquad
0\le y_{1,1}\le1.
\end{aligned}
\]
For any feasible solution,
\(
y_{1,1}+3y_{2,2}
\le \frac32(y_{1,1}+y_{2,1}+y_{2,2})
\le \frac32(1+p).
\)
This bound is achieved by
\(
y_{1,1}^*=0,
y_{2,1}^*=y_{2,2}^*=\frac{1+p}{2}.
\)
Moreover, equality requires $y_{1,1}=0$,
$y_{2,1}=y_{2,2}$, and a binding allocation constraint,
so this solution is uniquely optimal.

The independent-rounding algorithm therefore selects product 2
with probability
\(
p_2=\frac{y_{2,1}^*+y_{2,2}^*}{\mu_L}=1.
\)
It assigns every arriving customer to product 2 and earns an
expected reward of $3p$.

In contrast, the feasible online policy that assigns the first
customer to product 1 earns a reward of one with certainty.
Hence, the expected reward of LP rounding is at most $3p$ times
the optimal online expected reward. Letting $p\to0$ completes
the proof. \Halmos
\endproof

\subsection{Supporting Material for Section \ref{sec:limit-homo}}
\label{app:limit-homo}

We first establish how to evaluate a fixed-order fill algorithm, then establish the competitive ratio and its tightness.

\begin{lemma}[Polynomial-time evaluation of a fixed-order fill algorithm]
\label{lemma:fixed-order-eval}
For fixed ordering $\sigma$, the expected reward $\Rev^{\sigma}$ of the corresponding fixed-order fill algorithm can be computed in $O(B+L_{\max})$ time.
\end{lemma}

\proof{Proof of Lemma \ref{lemma:fixed-order-eval}:}
Fix an ordering $\sigma=(\sigma_1,\dots,\sigma_m)$. Under the fixed-order fill algorithm, product $\sigma_i$ receives assignments only after all preceding products $\sigma_1,\dots,\sigma_{i-1}$ have reached their capacities. Thus, its $k$-th assignment occurs if and only if $L\ge \sum_{j=1}^{i-1}b_{\sigma_j}+k$. The expected reward of the algorithm is therefore
\[
    \Rev^{\sigma}=\sum_{i=1}^m\sum_{k=1}^{b_{\sigma_i}}\left(r_{\sigma_i}+f(k)\right)\P\left(L\ge \sum_{j=1}^{i-1}b_{\sigma_j}+k\right).
\]

We next evaluate the computational complexity of this expression. Since $\mathcal D_L$ is known and supported on $\{0,1,\dots,L_{\max}\}$, the tail probabilities can be computed as 
\[
    \P(L\ge \ell)
    = \sum_{k=\ell}^{L_{\max}}\P(L=k)
    = 1-\sum_{k=0}^{\ell-1}\P(L=k),
\]
with $\P(L\ge \ell)=0$ for $\ell>L_{\max}$. All tail probabilities can be obtained in $O(L_{\max})$ time through a single backward cumulative summation. Given these probabilities, the expected reward expression contains $\sum_{i=1}^m b_i=B$ terms, each of which can be evaluated in constant time. Thus, $\Rev^{\sigma}$ can be computed in $O(B+L_{\max})$ time.

This completes the proof. \Halmos
\endproof

\proof{Proof of Theorem \ref{thm:homo-half-approx}:}
We establish the performance guarantee by separately bounding the base-reward and bonus components of the offline optimum for each realized arrival count. Specifically, an exchange argument shows that base-ordered fill maximizes the base reward, while a majorization argument shows that capacity-ordered fill maximizes the bonus reward. These two bounds yield a stronger realization-wise inequality, from which the competitive ratio of $1/2$ follows by taking expectations. We then analyze the running time.

Since the total capacity is $B$, the offline value and both fixed-order fill algorithms remain unchanged for arrival counts exceeding $B$. It therefore suffices to consider a fixed $\ell\in\{0,1,\dots,B\}$. Throughout this proof, maximizations over $\boldsymbol{x}$ are subject to the constraints in \eqref{eq:offline-opt}. Since all marginal rewards are nonnegative, an optimal offline allocation can be chosen to assign exactly $\ell$ customers. Under homogeneous bonuses, we have
\begin{equation}
\label{eq:pathwise-split}
\begin{aligned}
    U(\ell)
    &=\max_{\boldsymbol{x}}
    \left\{\sum_{i=1}^m x_i r_i+
    \sum_{i=1}^m\sum_{k=1}^{x_i}f(k)\right\}
    \le
    \max_{\boldsymbol{x}} \sum_{i=1}^m x_i r_i
    +
    \max_{\boldsymbol{x}} \sum_{i=1}^m\sum_{k=1}^{x_i}f(k).
\end{aligned}
\end{equation}
Let $\boldsymbol{x}^{\mathrm{base}}(\ell)$ and $\boldsymbol{x}^{\mathrm{cap}}(\ell)$ denote the allocations induced by the base-ordered and capacity-ordered fill algorithms, respectively, when $L=\ell$. We establish that these allocations attain the two maxima in \eqref{eq:pathwise-split}.

First, consider the base-reward component. Take any feasible allocation $\boldsymbol{x}$ with $\|\boldsymbol{x}\|_1=\ell$. If there exist positions $j<k$ such that $x_{\sigma_j^{\mathrm{base}}}<b_{\sigma_j^{\mathrm{base}}}$ and $x_{\sigma_k^{\mathrm{base}}}>0$, moving one assignment from product $\sigma_k^{\mathrm{base}}$ to product $\sigma_j^{\mathrm{base}}$ yields a feasible allocation $\boldsymbol{x}'=\boldsymbol{x}+\boldsymbol{e}_{\sigma_j^{\mathrm{base}}}-\boldsymbol{e}_{\sigma_k^{\mathrm{base}}}$ satisfying
\[
    \sum_{i=1}^m x_i'r_i-\sum_{i=1}^m x_i r_i
    =r_{\sigma_j^{\mathrm{base}}}-r_{\sigma_k^{\mathrm{base}}}\ge0.
\]
Repeating these exchanges yields $\boldsymbol{x}^{\mathrm{base}}(\ell)$ without decreasing the base reward. Therefore,
\[
    \max_{\boldsymbol{x}}
    \sum_{i=1}^m x_i r_i
    =\sum_{i=1}^m x_i^{\mathrm{base}}(\ell)r_i.
\]

Next, consider the bonus component. Take any feasible allocation $\boldsymbol{x}$ with $\|\boldsymbol{x}\|_1=\ell$, and let $\sigma^{\mathrm{fea}}=(\sigma_1^{\mathrm{fea}},\dots,\sigma_m^{\mathrm{fea}})$ be an ordering such that $x_{\sigma_1^{\mathrm{fea}}}\ge\cdots\ge x_{\sigma_m^{\mathrm{fea}}}$. Since $\sigma^{\mathrm{cap}}$ orders products by nonincreasing capacities, for every $i=1,\dots,m$,
\[
    \sum_{j=1}^i x_{\sigma_j^{\mathrm{fea}}}
    \le\min\left\{\ell,\sum_{j=1}^i b_{\sigma_j^{\mathrm{cap}}}\right\}
    =\sum_{j=1}^i x_{\sigma_j^{\mathrm{cap}}}^{\mathrm{cap}}(\ell).
\]
Here, the inequality follows because the total allocation to any $i$ products cannot exceed either $\ell$ or the combined capacity of the $i$ largest-capacity products; the equality follows from the definition of capacity-ordered fill.

Recall that a vector $\boldsymbol{u}$ majorizes another vector $\boldsymbol{v}$ if, after sorting their components in nonincreasing order, the sum of the largest $i$ components of $\boldsymbol{u}$ is at least that of $\boldsymbol{v}$ for every $i=1,\dots,m-1$, with equality for $i=m$. Since the components of $\boldsymbol{x}^{\mathrm{cap}}(\ell)$ are nonincreasing in the order $\sigma^{\mathrm{cap}}$, and both allocations contain exactly $\ell$ assignments, the inequalities above establish that $\boldsymbol{x}^{\mathrm{cap}}(\ell)$ majorizes $\boldsymbol{x}$.

Since $f$ is nondecreasing, the common cumulative bonus $\sum_{k=1}^n f(k)$ is discrete convex in $n$, and its piecewise linear interpolation is convex. By Karamata's inequality, if $\boldsymbol{u}$ majorizes $\boldsymbol{v}$, then the sum of any convex function evaluated at the components of $\boldsymbol{u}$ is at least that of $\boldsymbol{v}$. Applying this inequality to the common cumulative bonus gives
\[
    \sum_{i=1}^m\sum_{k=1}^{x_i}f(k)
    \le\sum_{i=1}^m\sum_{k=1}^{x_i^{\mathrm{cap}}(\ell)}f(k).
\]
Since all bonuses are nonnegative, an optimal allocation for the bonus component can be chosen to assign exactly $\ell$ customers. Hence,
\[
    \max_{\boldsymbol{x}}
    \sum_{i=1}^m\sum_{k=1}^{x_i}f(k)
    =\sum_{i=1}^m\sum_{k=1}^{x_i^{\mathrm{cap}}(\ell)}f(k).
\]

Combining the two componentwise optimality results with \eqref{eq:pathwise-split}, we obtain
\[
\begin{aligned}
    U(\ell)
    &\le \sum_{i=1}^m x_i^{\mathrm{base}}(\ell)r_i +\sum_{i=1}^m\sum_{k=1}^{x_i^{\mathrm{cap}}(\ell)}f(k)
    \le \Rev^{\sigma^{\mathrm{base}}}(\ell) +\Rev^{\sigma^{\mathrm{cap}}}(\ell),
\end{aligned}
\]
where the last inequality follows because both base rewards and bonuses are nonnegative. The same inequality holds for $\ell>B$, since all products have reached capacity. Taking expectations over $L\sim\mathcal D_L$ yields
\[
\begin{aligned}
    \max\left\{\Rev^{\sigma^{\mathrm{base}}}, \Rev^{\sigma^{\mathrm{cap}}}\right\}
    \ge\frac12\left(\Rev^{\sigma^{\mathrm{base}}}+\Rev^{\sigma^{\mathrm{cap}}}\right)
    \ge\frac12\,\OPT_{\mathrm{off}},
\end{aligned}
\]
establishing the competitive ratio of $1/2$.

It remains to establish the running time. Constructing the two orderings requires $O(m\log m)$ time. By Lemma~\ref{lemma:fixed-order-eval}, each fixed-order fill algorithm can be evaluated in $O(B+L_{\max})$ time. Thus, evaluating both algorithms and selecting the one with the larger expected reward takes $O(m\log m+B+L_{\max})$ time.

This completes the proof.
\Halmos
\endproof
We next show that no online algorithm achieves a competitive ratio larger than $1/2$. Define \(\OPT_{\mathrm{on}}:=\sup_\pi\Rev^\pi\) as the optimal online value. 
\begin{proposition}[Tightness under homogeneous bonuses]
\label{prop:homo-tight}
For every $\varepsilon>0$, there exists an instance with homogeneous bonuses such that ${\OPT_{\mathrm{on}}}/{\OPT_{\mathrm{off}}}<1/2+\varepsilon$.
Consequently, no online algorithm can achieve a competitive ratio strictly greater than $1/2$ under homogeneous bonuses.
\end{proposition}

\proof{Proof of Proposition \ref{prop:homo-tight}:}
We prove the result by constructing a family of instances with a common bonus function. For each $p\in(0,1)$, consider an arrival count $L$ distributed as: $\P(L=1)=1-p$ and $\P(L=2)=p$.
Thus, one customer arrives with certainty, while a second customer arrives with probability $p$. Product 1 has capacity $b_1=1$ and base reward $r_1=1$, whereas product 2 has capacity $b_2=2$ and base reward $r_2=0$. Both products share the same bonus function $f_1=f_2=f$, defined by
\[
    f(0)=f(1)=0,\qquad f(k)=\frac{1}{p}\ \text{for all } k\ge 2.
\]
Hence, $f$ is nonnegative and nondecreasing.

We first evaluate the optimal online value. Upon the first arrival, an online algorithm has three possible decisions. Assigning the customer to product 1 secures reward $1$, while any subsequent assignment to product 2 yields zero reward. Alternatively, assigning the first customer to product 2 yields zero immediate reward but allows the algorithm to obtain $1/p$ if the second customer arrives, giving an expected reward of $p\cdot(1/p)=1$. Finally, rejecting the first customer yields an expected reward of at most $p<1$. Thus, no decision yields an expected reward greater than $1$, and randomization cannot improve this bound. Since assigning the first customer to product 1 achieves reward $1$, we have $\OPT_{\mathrm{on}}=1$.

In contrast, the offline benchmark can choose its allocation based on the realized arrival count. When $L=1$, it assigns the customer to product 1 and earns reward $1$. When $L=2$, it assigns both customers to product 2 and earns reward $1/p$. Therefore,
\[
    \OPT_{\mathrm{off}}=(1-p)\cdot 1+p\cdot\frac{1}{p}=2-p.
\]

Consequently, for any online algorithm $\pi$,
\[
    \frac{\Rev^\pi}{\OPT_{\mathrm{off}}}
    \le\frac{\OPT_{\mathrm{on}}}{\OPT_{\mathrm{off}}}
    =\frac{1}{2-p}.
\]
As $p\to0$, the right-hand side converges to $1/2$. Therefore, for any prescribed constant $\varepsilon>0$, choosing $0<p<\min\{1,4\varepsilon/(1+2\varepsilon)\}$ gives $\Rev^\pi<(1/2+\varepsilon)\OPT_{\mathrm{off}}$ for every online algorithm $\pi$. Hence, no online algorithm admits a universal competitive ratio strictly larger than $1/2$ under homogeneous bonuses.

This concludes the proof. \Halmos
\endproof

\subsection{Supporting Material for Section \ref{sec:limit-hete}}
\label{app:limit-hete}

\proof{Proof of Theorem \ref{thm:hete-decay}:}

We first establish a universal $1/m$ guarantee for every instance with $m$ products and then construct a family of instances showing that this guarantee is tight.

\textbf{(a) Universal $1/m$ guarantee.} 
Consider the online algorithm $\pi$ that commits all arrivals to the product with the largest standalone expected reward. Its expected reward is
\[
    \Rev^\pi=\max_{i\in[m]}\E\left[R_i(\min\{L,b_i\})\right].
\]
For any realized arrival count $\ell$, each product $i$ can receive at most $\min\{\ell,b_i\}$ assignments. Since $R_i$ is nondecreasing, the offline optimal value satisfies
\[
    U(\ell)\le\sum_{i=1}^m R_i(\min\{\ell,b_i\}).
\]
Taking expectations over $L$, we obtain
\[
\begin{aligned}
    \OPT_{\mathrm{off}}
    \le\sum_{i=1}^m\E\left[R_i(\min\{L,b_i\})\right]
    \le m\max_{i\in[m]}\E\left[R_i(\min\{L,b_i\})\right]
    =m\Rev^\pi.
\end{aligned}
\]
Thus, $\pi$ achieves the claimed competitive ratio of $1/m$.

\textbf{(b) A matching family of instances.} 
We next construct a family of heterogeneous bonus instances for which the ratio approaches $1/m$. Fix $m\ge2$. For each $p\in(0,1)$, construct an instance as follows. One customer arrives with certainty, and after each arrival, another customer arrives with probability $p$, up to a maximum of $m$ customers. Thus, the arrival count $L$ follows a truncated geometric distribution satisfying $\P(L\ge i)=p^{i-1}$ for all $i=1,\dots,m$.

There are $m$ products. Product $i$ has capacity $b_i=i$, base reward $r_i=0$, and bonus function
\[
    f_i(k)=
    \begin{cases}
        0, & k<i,\\
        p^{-(i-1)}, & k\ge i.
    \end{cases}
\]
Hence, each $f_i$ is nonnegative and nondecreasing. Product $i$ thus earns a lump-sum reward of $p^{-(i-1)}$ only upon completing its $i$-th assignment. By construction, committing all arrivals to any single product $i$ yields expected reward $\P(L\ge i)p^{-(i-1)}=1$.

We first upper bound the optimal online value. Consider any admissible online algorithm $\pi$. Since rejection is weakly dominated whenever some product has remaining capacity, we may assume that the first customer is assigned to some product $j\in[m]$. Completing product $j$ requires at least $j$ arrivals, whereas completing any other product $i\neq j$ requires at least $i+1$ arrivals, since it did not receive the first customer. Conditioning on this possibly randomized first-arrival decision, we obtain
\[
\begin{aligned}
    \OPT_{\mathrm{on}}
    \le\max_{j\in[m]}\left\{p^{-(j-1)}\P(L\ge j)+\sum_{i\in[m]\setminus\{j\}}p^{-(i-1)}\P(L\ge i+1)\right\}
    \le 1+(m-1)p,
\end{aligned}
\]
where $\P(L\ge m+1)=0$.

We next lower bound the offline value by considering a feasible allocation for each arrival-count realization. Conditional on $L=i$, the offline benchmark assigns all $i$ customers to product $i$ and collects reward $p^{-(i-1)}$. Since $\P(L=i)=(1-p)p^{i-1}$ for $i<m$ and $\P(L=m)=p^{m-1}$, we have
\[
\begin{aligned}
    \OPT_{\mathrm{off}}
    \ge\sum_{i\in[m]}\P(L=i)p^{-(i-1)}
    =(m-1)(1-p)+1
    =m-(m-1)p.
\end{aligned}
\]

Consequently, for any online algorithm $\pi$,
\[
    \frac{\Rev^\pi}{\OPT_{\mathrm{off}}}
    \le \frac{\OPT_{\mathrm{on}}}{\OPT_{\mathrm{off}}}
    \le\frac{1+(m-1)p}{m-(m-1)p}.
\]
As $p\to 0$, the right-hand side converges to $1/m$. Hence, for any constant $\varepsilon>1/m$, choosing $0<p<\min\{1,(m\varepsilon-1)/((m-1)(1+\varepsilon))\}$ gives $\Rev^\pi<\varepsilon\,\OPT_{\mathrm{off}}$ for every online algorithm $\pi$. Hence, no online algorithm can achieve a universal competitive ratio strictly larger than $1/m$ over this class of instances.

Combining parts~\textbf{(a)} and \textbf{(b)}, we conclude that the optimal competitive ratio for $m$ products with arbitrary heterogeneous bonuses is exactly $1/m$ for every $m\ge2$.

This completes the proof.
\Halmos
\endproof

\section{Supporting Material for Section~\ref{sec:dc}}
\label{app:dc}

\subsection{Limitations of Classical Algorithms under Discrete Concavity}
\label{app:dc-fail}

\subsubsection{Limitations of Greedy and LP-Based Rounding Algorithms.}
\label{app:dc-fail-1}

\begin{proposition}[Failure of greedy under Assumption~\ref{assump:discrete-concavity}]
\label{prop:greedy-fail-concave}
For every $\varepsilon>0$, there exists an instance with homogeneous linear bonuses satisfying Assumption~\ref{assump:discrete-concavity} such that $\Rev^{\mathrm{Greedy}}/\OPT_{\mathrm{off}}< \varepsilon$.
\end{proposition}

\proof{Proof of Proposition \ref{prop:greedy-fail-concave}:}
We prove the result by constructing a family of instances with a common linear bonus function. For each integer $N\ge2$, construct an instance with an arrival count $L$ distributed as $\P(L=1)=\P(L=N)=1/2$.
Thus, one customer arrives with certainty, while the total number of arrivals reaches $N$ with probability $1/2$. There are $N+1$ products. Product 1 has capacity $b_1=N$ and base reward $r_1=0$, whereas products $2,\dots,N+1$ each have capacity $b_i=1$ and base reward $r_i=1$. All products share the same bonus function $f_i=f$, defined by $f(k)=k$, for all $ k\ge0$.
Hence, $f$ is nonnegative and nondecreasing. Moreover, $f(k+1)-f(k)=1$ for every $k\ge0$, so $f$ satisfies Assumption~\ref{assump:discrete-concavity}.

Under greedy, assigning the first customer to product 1 yields marginal reward $r_1+f(1)=1$, whereas assigning a customer to any product $i\in\{2,\dots,N+1\}$ yields reward $r_i+f(1)=2$. Greedy therefore assigns each arriving customer to an unfilled product among $2,\dots,N+1$, obtaining
\[
    \Rev^{\mathrm{Greedy}}
    =2\mathbb E[L] =2\left(\frac12+\frac N2\right) =N+1.
\]

Now consider an online policy that assigns every arriving customer to product 1. We have
\[
    \OPT_{\mathrm{on}}
    \ge \mathbb E[R_1(L)]
    =\mathbb E\left[\sum_{k=1}^L k\right]
    =\frac12\cdot1+\frac12\cdot\frac{N(N+1)}{2}
    =\frac{N(N+1)+2}{4}.
\] Consequently,
\[
    \frac{\Rev^{\mathrm{Greedy}}}{\OPT_{\mathrm{off}}}
    \le\frac{\Rev^{\mathrm{Greedy}}}{\OPT_{\mathrm{on}}}
    \le\frac{4(N+1)}{N(N+1)+2}
    \le\frac4N.
\]
For any prescribed constant $\varepsilon>0$, choosing an integer $N>\max\{1,4/\varepsilon\}$ gives $\Rev^{\mathrm{Greedy}}<\varepsilon\,\OPT_{\mathrm{off}}$. Therefore, greedy admits no universal constant guarantee over this class.

This concludes our proof. 
\Halmos
\endproof

\begin{proposition}[Failure of mean-demand LP rounding under Assumption~\ref{assump:discrete-concavity}]
\label{prop:lp-fail-concave}
For every $\varepsilon>0$, there exists an instance with homogeneous linear bonuses such that $\Rev^{\mathrm{LP}}/\OPT_{\mathrm{off}}<\varepsilon$.
\end{proposition}

\proof{Proof of Proposition~\ref{prop:lp-fail-concave}:}
We prove the result by constructing a family of instances with a common linear bonus function. For each integer $N\ge3$, construct an instance with an arrival count $L$ distributed as
\[
    \P(L=1)=1-\frac{1}{N^2},
    \qquad
    \P(L=N)=\frac{1}{N^2}.
\]
Thus, one customer arrives with certainty, while the total number of arrivals reaches $N$ with probability $1/N^2$. The expected arrival count is
\[
    \mu_L:=\E[L]=1+\frac{N-1}{N^2}.
\]

There are two products. Product 1 has capacity $b_1=1$ and base reward $r_1=1$, whereas product 2 has capacity $b_2=N$ and base reward $r_2=0$. Both products share the same bonus function $f_1=f_2=f$, defined by $f(k)={4k}/{N}$ for all $k\ge0$.
Hence, $f$ is nonnegative and nondecreasing. Moreover, $f(k+1)-f(k)=4/N$ for every $k\ge0$, so $f$ satisfies Assumption~\ref{assump:discrete-concavity}.

We first characterize the optimal solution of LP~\eqref{eq:lp}:
\[
\begin{aligned}
    \max_{\boldsymbol y\in\mathbb R_{\ge0}^{N+1}}\quad
    &\left(1+\frac4N\right)y_{1,1}
    +\sum_{k=1}^N\frac{4k}{N}y_{2,k}\\
    \text{s.t.}\quad
    &y_{1,1}+\sum_{k=1}^N y_{2,k}\le\mu_L,
    \quad 0\le y_{2,N}\le\cdots\le y_{2,1}\le1,
    \quad 0\le y_{1,1}\le1.
\end{aligned}
\]

We claim that the unique optimal solution is
\[
    y_{1,1}^*=0,
    \qquad
    y_{2,k}^*=\frac{\mu_L}{N},
    \quad k=1,\dots,N.
\]

First, the proposed solution is feasible. Indeed, $\sum_{k=1}^N y_{2,k}^*=\mu_L$, so the LP budget constraint holds with equality, and the monotonicity constraints hold because all $y_{2,k}^*$ are equal. Moreover, since $N\ge3$, we have $\mu_L=1+(N-1)/N^2<2<N$, implying $0<y_{2,k}^*=\mu_L/N<1$.

Next, we establish optimality and uniqueness. For any feasible solution $\boldsymbol y$, we have
\begin{equation}
\label{eq:lp-obj-ub}
\begin{aligned}
    \left(1+\frac4N\right)y_{1,1} +\sum_{k=1}^N\frac{4k}{N}y_{2,k}
    &\le \left(1+\frac4N\right)y_{1,1} +\frac4{N^2}\left(\sum_{k=1}^N k\right)
    \left(\sum_{k=1}^N y_{2,k}\right)\\
    &=\left(1+\frac4N\right)y_{1,1} +\left(2+\frac2N\right)\sum_{k=1}^N y_{2,k}\\
    &\le\left(2+\frac2N\right) \left(y_{1,1}+\sum_{k=1}^N y_{2,k}\right)
    \le\left(2+\frac2N\right)\mu_L,
\end{aligned}
\end{equation}
where the first inequality follows from the reverse Chebyshev sum inequality, since $k$ is strictly increasing and $y_{2,k}$ is nonincreasing in $k$; the second follows from $1+4/N<2+2/N$ for $N\ge3$; and the last inequality follows from the LP budget constraint \eqref{eq:lp-cap-inc}. 

The proposed solution $\boldsymbol y^*$ attains this upper bound, since
\[
    \left(1+\frac4N\right)y_{1,1}^* +\sum_{k=1}^N\frac{4k}{N}y_{2,k}^*
    =\frac{4\mu_L}{N^2}\sum_{k=1}^N k
    =\left(2+\frac2N\right)\mu_L.
\]
Hence, $\boldsymbol y^*$ is optimal. To establish uniqueness, observe that any optimal solution must attain equality in every inequality in \eqref{eq:lp-obj-ub}. Since $1+4/N<2+2/N$, equality in the second inequality requires $y_{1,1}=0$. Equality in the last inequality then requires $\sum_{k=1}^N y_{2,k}=\mu_L$. Finally, equality in the reverse Chebyshev sum inequality requires $y_{2,1}=\cdots=y_{2,N}$, because $k$ is strictly increasing. Combining these conditions yields $y_{2,k}=\mu_L/N$ for every $k=1,\dots,N$. Therefore, $\boldsymbol y^*$ is the unique optimal solution.

Since $y_{1,1}^*=0$ and $\sum_{k=1}^N y_{2,k}^*=\mu_L$, the LP-based independent-rounding policy assigns every arriving customer to product 2. Its expected reward is therefore
\[
\begin{aligned}
    \Rev^{\mathrm{LP}}
    &=\E[R_2(L)]
    =\E\left[\sum_{k=1}^L\frac{4k}{N}\right]
    =\left(1-\frac1{N^2}\right)\frac4N+\frac1{N^2}\cdot2(N+1)
    =\frac6N+\frac2{N^2}-\frac4{N^3}.
\end{aligned}
\]

In contrast, consider an online policy that assigns the first arriving customer to product 1. Since the first customer arrives with certainty, we have
\[
    \OPT_{\mathrm{on}}\ge r_1+f(1)=1+\frac4N.
\]
Consequently,
\[
    \frac{\Rev^{\mathrm{LP}}}{\OPT_{\mathrm{off}}}
    \le\frac{\Rev^{\mathrm{LP}}}{\OPT_{\mathrm{on}}}
    \le\frac{6/N+2/N^2-4/N^3}{1+4/N}\le\frac6N.
\]
For any constant $\varepsilon>0$, choosing an integer $N>\max\{2,6/\varepsilon\}$ gives $\Rev^{\mathrm{LP}}<\varepsilon\,\OPT_{\mathrm{off}}$. Therefore, LP-based independent rounding admits no universal constant guarantee over this class.

This completes the proof.
\Halmos
\endproof

\subsubsection{Limitation of Time-Indexed LP and Independent Rounding.}
\label{app:dc-fail-2}

We next consider a hypothetical full assignment sequence of length $B$ and an LP that explicitly tracks the arrival position of each assignment and its level within the assigned product. Let $z_{\ell,i,k}$ denote the probability that position $\ell$ in this sequence contains the $k$th assignment to product $i$. 
Consider the following time-indexed LP:
\begin{subequations}
\label{eq:time-indexed-lp}
\begin{align}
    V^{\mathrm{TI}}=\max_{\boldsymbol z\in\mathbb R_{\ge0}^{B^2}}\quad
    &\sum_{\ell=1}^{B}\P(L\ge\ell)\sum_{i\in[m]}\sum_{k=1}^{b_i}(r_i+f_i(k))z_{\ell,i,k}
    \label{eq:ti-obj}\\
    \text{s.t.}\quad
    &\sum_{i\in[m]}\sum_{k=1}^{b_i}z_{\ell,i,k}=1,
    &&\ell=1,\ldots,B,
    \label{eq:ti-position}\\
    &\sum_{\ell=1}^{B}z_{\ell,i,k}=1,
    &&i\in[m],\ k=1,\ldots,b_i,
    \label{eq:ti-level}\\
    &\sum_{\tau=1}^{\ell}z_{\tau,i,k}
    \le\sum_{\tau=1}^{\ell-1}z_{\tau,i,k-1},
    &&i\in[m],\ k=2,\ldots,b_i,\ \ell=1,\ldots,B,
    \label{eq:ti-precedence}\\
    &z_{\ell,i,k}\ge0,
    &&i\in[m],\ k=1,\ldots,b_i,\ \ell=1,\ldots,B.
    \label{eq:ti-nonneg}
\end{align}
\end{subequations}
Constraints~\eqref{eq:ti-position} assign one product level to each position, while constraints~\eqref{eq:ti-level} place every level exactly once in the hypothetical sequence. Constraints~\eqref{eq:ti-precedence} impose marginal within-product precedence: the probability that product $i$ reaches level $k$ by position $\ell$ cannot exceed the probability that it reached level $k-1$ by position $\ell-1$. 
For any feasible solution $\boldsymbol z$, independent rounding selects product $i$ at arrival position $\ell\le B$ with probability $\sum_{k=1}^{b_i}z_{\ell,i,k}$, independently across positions; it rejects the customer if the selected product is full, and rejects all arrivals after position $B$. Let $\Rev^{\mathrm{TI\text{-}ind}}(\boldsymbol z)$ denote the expected reward of this policy.

\begin{proposition}[Failure of time-indexed LP rounding]
\label{prop:ti-round-fail}
For every $\varepsilon>0$, there exists an instance with homogeneous
linear bonuses in which the expected reward of independent rounding
of an optimal solution of the time-indexed LP is less than
$\varepsilon\OPT_{\mathrm{off}}$.
\end{proposition}

\proof{Proof of Proposition~\ref{prop:ti-round-fail}:}
Fix an integer $N\ge2$. There are $N$ products, each with capacity
$N$, and $L=N$ deterministically. All products have zero base rewards
and share the linear bonus function $f_i(k)=k$. Thus,
$R_i(k)=k(k+1)/2$. Assigning all customers to one product is optimal,
so
\(
    \OPT_{\mathrm{on}}=\OPT_{\mathrm{off}}
    =\frac{N(N+1)}2.
\)

To construct a feasible LP solution $\boldsymbol z^*$, consider a
uniformly random permutation of the products and assign $N$
consecutive positions to each product in this order. This gives a
random feasible assignment sequence of length $B=N^2$, whose
marginals satisfy all constraints of the time-indexed LP.
In particular, for the first $N$ positions,
\[
    z^*_{\ell,i,k}
    =
    \begin{cases}
        1/N, & k=\ell,\\
        0,   & k\ne\ell,
    \end{cases}
    \qquad \ell,i\in[N].
\]
The objective value of this solution is
$\sum_{\ell=1}^N\ell=N(N+1)/2$.
Moreover, the precedence constraints imply that
$z_{\ell,i,k}=0$ whenever $k>\ell$, so no feasible LP solution
can have a larger objective value. Hence, $\boldsymbol z^*$ is
optimal and
\(
    V^{\mathrm{TI}}
    =\OPT_{\mathrm{on}}
    =\OPT_{\mathrm{off}}.
\)

Independent rounding assigns each customer independently
and uniformly to one of the $N$ products. No customer is rejected,
since each product has capacity $N$. Let $X_i$ denote the final
number of assignments to product $i$. Then
$X_i\sim\mathrm{Binomial}(N,1/N)$, and
\[
\begin{aligned}
    \Rev^{\mathrm{TI\text{-}ind}}(\boldsymbol z^*)=\sum_{i=1}^N
      \E\left[X_i+\binom{X_i}{2}\right]=N+\frac1N\binom N2
      =\frac{3N-1}{2}.
\end{aligned}
\]
Therefore,
\[
    \frac{\Rev^{\mathrm{TI\text{-}ind}}(\boldsymbol z^*)}
    {\OPT_{\mathrm{off}}}
    =\frac{3N-1}{N(N+1)}
    <\frac3N.
\]
Choosing $N>3/\varepsilon$ proves the result.\Halmos
\endproof

\subsection{Supporting Material for Section~\ref{sec:dc-re}}
\label{app:dc-re}

\proof{Proof of Lemma~\ref{lemma:quad-prefix}:}
Fix any current allocation $\boldsymbol s$ and feasible additional allocation $\boldsymbol a$ with $h_a:=\|\boldsymbol a\|_1\ge1$. We first use discrete concavity to bound the endpoint marginal reward of each product, then apply this bound to control the reward lost at each backward-greedy step. Combining these one-step bounds through a telescoping product yields the desired prefix guarantee.

First, we bound the endpoint marginal reward of each product by the incremental value of its planned assignments. Fix a product $i$ with $a_i\ge1$. For every $k=1,\ldots,a_i$, we have
\[
    f_i(s_i+k)
    \ge\left(1-\frac{k}{a_i}\right)f_i(s_i)+\frac{k}{a_i}f_i(s_i+a_i)
    \ge\frac{k}{a_i}f_i(s_i+a_i),
\]
where the first inequality follows from discrete concavity under Assumption~\ref{assump:discrete-concavity}, and the second follows from $f_i(s_i)\ge0$. Summing over the planned assignments gives
\begin{equation}
\label{eq:one-prod-bound}
\begin{aligned}
    R_i(s_i+a_i)-R_i(s_i)
    &=\sum_{k=1}^{a_i}\bigl[r_i+f_i(s_i+k)\bigr]
    \ge\sum_{k=1}^{a_i}\left[r_i+\frac{k}{a_i}f_i(s_i+a_i)\right]\\
    &\ge\sum_{k=1}^{a_i}\frac{k}{a_i}\bigl[r_i+f_i(s_i+a_i)\bigr]
    =\frac{a_i+1}{2}\bigl[r_i+f_i(s_i+a_i)\bigr],
\end{aligned}
\end{equation}
where the second inequality follows from $r_i\ge0$ and $k/a_i\le1$. Thus, the endpoint marginal reward is at most $2/(a_i+1)$ times the incremental reward of the planned assignments. Since this bound holds for any current allocation $s_i$ and positive planned count $a_i$, it also applies to the remaining assignments at each step of backward greedy.

Next, we bound the reward lost at each step of backward greedy. Consider position $k$, with $1\le k\le h_a$. At this step, $\boldsymbol a^{(k)}$ contains the $k$ assignments remaining to be placed in positions $1,\ldots,k$. We have
\begin{equation}
\label{eq:delta-bound}
\begin{aligned}
    V(\boldsymbol s+\boldsymbol a^{(k)})-V(\boldsymbol s)
    &=\sum_{i\in[m]}\bigl[R_i(s_i+a_i^{(k)})-R_i(s_i)\bigr]
    \ge\frac12\sum_{i:\,a_i^{(k)}>0}(a_i^{(k)}+1)\bigl[r_i+f_i(s_i+a_i^{(k)})\bigr]\\
    &\ge\frac12\left(k+\sum_{i:\,a_i^{(k)}>0}1\right)\bigl[r_{\omega_k}+f_{\omega_k}(s_{\omega_k}+a_{\omega_k}^{(k)})\bigr]
    \ge\frac{k+1}{2}\bigl[r_{\omega_k}+f_{\omega_k}(s_{\omega_k}+a_{\omega_k}^{(k)})\bigr],
\end{aligned}
\end{equation}
where the first inequality follows from \eqref{eq:one-prod-bound}, the second follows from the backward-greedy selection rule and $\sum_i a_i^{(k)}=k$, and the last follows since at least one product has a positive remaining count.

As a result, we have
\[
\begin{aligned}
    V(\boldsymbol s+\boldsymbol a^{(k-1)})-V(\boldsymbol s)
    &=V(\boldsymbol s+\boldsymbol a^{(k)})-V(\boldsymbol s)-\bigl[r_{\omega_k}+f_{\omega_k}(s_{\omega_k}+a_{\omega_k}^{(k)})\bigr]
    \ge\frac{k-1}{k+1}\bigl[V(\boldsymbol s+\boldsymbol a^{(k)})-V(\boldsymbol s)\bigr],
\end{aligned}
\]
where the equality follows from $\boldsymbol a^{(k-1)}=\boldsymbol a^{(k)}-\boldsymbol e_{\omega_k}$, and the inequality follows from \eqref{eq:delta-bound}. Hence, each backward step preserves at least a fraction $(k-1)/(k+1)$ of the remaining plan's incremental reward.

Finally, for every $1\le h<h_a$, we have 
\[
\begin{aligned}
    V(\boldsymbol s+\boldsymbol a^{(h)})-V(\boldsymbol s)
    \ge\left(\prod_{k=h+1}^{h_a}\frac{k-1}{k+1}\right)\bigl[V(\boldsymbol s+\boldsymbol a)-V(\boldsymbol s)\bigr]
    =\frac{h(h+1)}{h_a(h_a+1)}\bigl[V(\boldsymbol s+\boldsymbol a)-V(\boldsymbol s)\bigr],
\end{aligned}
\]
where the inequality follows by applying the preceding one-step bound successively for $k=h_a,\ldots,h+1$, and the equality follows from the telescoping product. The claimed inequality also holds for $h=0$, since both sides are zero, and for $h=h_a$, with equality.

Moreover, since each backward step removes one assignment from a product with a positive remaining count, we have $\boldsymbol a^{(0)}=\boldsymbol0$ and $\boldsymbol a^{(h)}=\sum_{k=1}^h\boldsymbol e_{\omega_k}$ for every $h=0,\ldots,h_a$. Thus, $\omega$ implements exactly $\boldsymbol a$, and each prefix is feasible from $\boldsymbol s$.

This completes the proof.
\Halmos
\endproof

\subsection{Supporting Material for Section~\ref{sec:dc-gen}}
\label{app:dc-gen}

We first establish benchmark attainment, then prove the deterministic and randomized performance guarantees, and finally give an upper-bound construction for general arrivals.

\subsubsection{Supporting Material for Section~\ref{sec:dc-gen-d}.}
\label{app:dc-gen-plan}

\proof{Proof of Lemma~\ref{lem:general-value-repair}:}
We use the overlap between the current allocation and the offline optimizer to apply the prefix-value guarantee to the missing assignments. Combining this bound with the previously secured benchmark quantifies progress toward the next benchmark.

If $V(\boldsymbol s)\ge U(n)$, the result follows from the monotonicity of $V$. Hence, in the following, we focus on $V(\boldsymbol s)<U(n)$, which implies that the plan is nonempty. We first consider $n\le B$ and let
\(
\underline{\boldsymbol x}=\boldsymbol s\wedge\boldsymbol x^*(n),
h_{\underline x}=\|\underline{\boldsymbol x}\|_1,
h_a=\|\boldsymbol a(n)\|_1=n-h_{\underline x}>0.
\) Here, $\wedge$ denotes the coordinatewise minimum.

Since $\boldsymbol a(n)=\boldsymbol x^*(n)-\underline{\boldsymbol x}$ and $\boldsymbol s$ agrees with $\underline{\boldsymbol x}$ on every coordinate where $\boldsymbol a(n)$ is positive, the separability of $V$ implies that executing the plan from either allocation yields the same incremental reward. For every integer $0\le h\le h_a$, the discrete prefix fraction $h(h+1)/[h_a(h_a+1)]$ is at least $(h/h_a)^2$. Thus Lemma~\ref{lemma:quad-prefix} gives
\begin{equation}
\label{eq:general-repair-prefix}
\begin{aligned}
V(\boldsymbol s+\boldsymbol a(n)^{(h)})-V(\boldsymbol s)
&=V(\underline{\boldsymbol x}+\boldsymbol a(n)^{(h)})-V(\underline{\boldsymbol x})
\ge\frac{h^2}{h_a^2}\bigl[V(\underline{\boldsymbol x}+\boldsymbol a(n))-V(\underline{\boldsymbol x})\bigr]
=\frac{h^2}{h_a^2}\bigl[U(n)-V(\underline{\boldsymbol x})\bigr].
\end{aligned}
\end{equation}

We distinguish two cases based on whether $h_{\underline x}\le\tilde{n}$.

\emph{Case 1: $h_{\underline x}\le\tilde{n}$.}
Starting from $\underline{\boldsymbol x}$, the first $\tilde{n}-h_{\underline x}$ assignments produce a feasible allocation of size $\tilde{n}$. Since $V(\boldsymbol s)\ge U(\tilde{n})$, applying Lemma~\ref{lemma:quad-prefix} yields
\[
\begin{aligned}
V(\boldsymbol s)&\ge U(\tilde{n})\ge V(\underline{\boldsymbol x}+\boldsymbol a(n)^{(\tilde{n}-h_{\underline x})})
\ge V(\underline{\boldsymbol x})+\frac{(\tilde{n}-h_{\underline x})^2}{h_a^2}\bigl[U(n)-V(\underline{\boldsymbol x})\bigr].
\end{aligned}
\]
Here the second inequality holds since $U(\tilde{n})$ is the offline optimum for $\tilde{n}$ customers, and the last inequality follows from the prefix reward guarantee. Consequently,
\[
\begin{aligned}
U(n)-V(\boldsymbol s)
&\le U(n)-V(\underline{\boldsymbol x})-\frac{(\tilde{n}-h_{\underline x})^2}{h_a^2}\bigl[U(n)-V(\underline{\boldsymbol x})\bigr]
=\frac{h_a^2-(\tilde{n}-h_{\underline x})^2}{h_a^2}\bigl[U(n)-V(\underline{\boldsymbol x})\bigr].
\end{aligned}
\]
Combining this inequality with \eqref{eq:general-repair-prefix}, we obtain
\[
\begin{aligned}
V(\boldsymbol s+\boldsymbol a(n)^{(h)})-V(\boldsymbol s)
&\ge\frac{h^2}{h_a^2-(\tilde{n}-h_{\underline x})^2}\bigl[U(n)-V(\boldsymbol s)\bigr]
\ge\frac{h^2}{n^2-\tilde{n}^2}\bigl[U(n)-V(\boldsymbol s)\bigr],
\end{aligned}
\]
where the last inequality follows from
\(
h_a^2-(\tilde{n}-h_{\underline x})^2
=(n-h_{\underline x})^2-(\tilde{n}-h_{\underline x})^2
\le n^2-\tilde{n}^2.
\)

\emph{Case 2: $h_{\underline x}>\tilde{n}$.}
We have
\[
\begin{aligned}
V(\boldsymbol s+\boldsymbol a(n)^{(h)})-V(\boldsymbol s)
&\ge\frac{h^2}{h_a^2}\bigl[U(n)-V(\underline{\boldsymbol x})\bigr]\ge\frac{h^2}{h_a^2}\bigl[U(n)-V(\boldsymbol s)\bigr]
\ge\frac{h^2}{n^2-\tilde{n}^2}\bigl[U(n)-V(\boldsymbol s)\bigr],
\end{aligned}
\]
where the second uses $\underline{\boldsymbol x}\le\boldsymbol s$ and monotonicity of $V$, which give $V(\underline{\boldsymbol x})\le V(\boldsymbol s)$. For the last inequality, the value gap $U(n)-V(\boldsymbol s)$ is positive, and
\(
h_a^2=(n-h_{\underline x})^2\le(n-\tilde{n})^2\le(n-\tilde{n})(n+\tilde{n})=n^2-\tilde{n}^2,
\)
so enlarging the denominator weakens the lower bound.
Thus, adding $V(\boldsymbol s)$ in either case proves \eqref{eq:general-value-repair} when $n\le B$.

For $n>B$, the offline optimizer is $\boldsymbol b$, so $U(n)=U(B)$ and $\boldsymbol a(n)=\boldsymbol a(B)$. Under the remaining assumption $V(\boldsymbol s)<U(B)$, we necessarily have $\tilde{n}<B$. Apply the argument just proved with the prospective arrival count equal to $B$. Since $B^2-\tilde{n}^2\le n^2-\tilde{n}^2$ and the value gap is positive, replacing its denominator by $n^2-\tilde{n}^2$ only weakens the lower bound. This proves the stated guarantee for $n>B$ as well.

This completes the proof.\Halmos
\endproof

\subsubsection{Supporting Material for Section~\ref{sec:dc-gen-d}.}
\label{app:dc-gen-d}

\proof{Proof of Theorem~\ref{thm:general-d-cr}:}
Fix an integer $\gamma\ge2$. In particular, $\Rev^{\pi(\gamma)}(\ell)=V(\boldsymbol s(\ell))$ for $\ell=0,\ldots,B$. If $B=0$, both values are zero. Assume $B\ge1$. We first establish the pointwise guarantee for epochs $j\ge1$ by combining the value secured before the epoch with the progress made during it. We then verify initialization and the boundary cases before evaluating the competitive ratio.

Fix an epoch $j\ge1$ and an integer $T_{j-1}\le\ell\le\min\{T_j,B\}$. Every epoch allows $n_j$ additional assignments, whereas its plan has size $\|\boldsymbol a(n_j)\|_1\le\min\{n_j,B\}\le n_j$. Thus, whenever $T_j\le B$, the plan is completed by the checkpoint and $V(\boldsymbol s(T_j))\ge U(n_j)$. In particular, the preceding epoch secures $V(\boldsymbol s(T_{j-1}))\ge U(n_{j-1})$. The starting value also satisfies $V(\boldsymbol s(T_{j-1}))\le U(T_{j-1})\le U(n_j)$, because $n_j=(\gamma-1)T_{j-1}\ge T_{j-1}$.

During partial execution of a nonempty plan, Lemma~\ref{lem:general-value-repair} with $h=\ell-T_{j-1}$ gives the first inequality in \eqref{eq:general-det-reward}: capping the progress coefficient at one only weakens the bound because $U(n_j)-V(\boldsymbol s(T_{j-1}))\ge0$. If the plan is empty or completed, $V(\boldsymbol s(\ell))\ge U(n_j)$, which also bounds the convex combination on the right-hand side from above. The second inequality follows by replacing the starting value with $U(n_{j-1})$, since its weight is nonnegative. For $n_j\ge B$, the same argument applies with $\boldsymbol x^*(n_j)=\boldsymbol b$ and $U(n_j)=U(B)$.

We next compare the two benchmarks with $U(\ell)$. Since $n_{j-1}\le T_{j-1}\le\ell$, \eqref{eq:offline-value-growth} and monotonicity give
\[
U(n_{j-1})\ge\frac{n_{j-1}^2}{\ell^2}U(\ell),
\qquad
U(n_j)\ge\min\left\{1,\frac{n_j^2}{\ell^2}\right\}U(\ell).
\]
Both benchmarks have nonnegative weights in the second inequality of \eqref{eq:general-det-reward}. Substituting these comparisons therefore yields
\[
\begin{aligned}
V(\boldsymbol s(\ell))
&\ge\biggl[\frac{n_{j-1}^2}{\ell^2}
+\min\left\{1,\frac{(\ell-T_{j-1})^2}{n_j^2-n_{j-1}^2}\right\}
\left(\min\left\{1,\frac{n_j^2}{\ell^2}\right\}-\frac{n_{j-1}^2}{\ell^2}\right)\biggr]U(\ell)
=\eta_D(\gamma,t)U(\ell).
\end{aligned}
\]
The equality follows by substituting $n_j=(\gamma-1)T_{j-1}$, $n_{j-1}=(\gamma-1)T_{j-1}/\gamma$, and $\ell=tT_{j-1}$; the same calculation holds at $t=1$. Since $\Rev^{\pi(\gamma)}(\ell)=V(\boldsymbol s(\ell))$, minimizing over $t\in[1,\gamma]$ proves the pointwise guarantee for epochs $j\ge1$.

It remains to verify the first epoch. Initialization gives $V(\boldsymbol s(1))=U(1)$. When $\gamma\ge3$, we have $1<n_0=\gamma-1$, so Lemma~\ref{lem:general-value-repair} applies with preceding level $1$ and next level $n_0$. Since $(\gamma-1)/\gamma<1$, replacing the denominator $n_0^2-1$ by $n_0^2-((\gamma-1)/\gamma)^2$ and capping the progress coefficient at one only weakens the lower bound. Empty or completed plans are again covered by monotonicity. For $\gamma=2$, initialization already attains $U(n_0)=U(1)$, so the corresponding value gap is zero. Thus the first inequality in \eqref{eq:general-det-reward} also holds for integer $1<\ell\le\min\{T_0,B\}$ with $T_{-1}=1$ and the preceding-level parameter set to $(\gamma-1)/\gamma$. Applying \eqref{eq:offline-value-growth} with $h=1$ and $k=\ell$ supplies the preceding-benchmark comparison for this parameter, since $(\gamma-1)/\gamma<1$. The same convex-combination calculation therefore establishes the pointwise guarantee in the first epoch as well, without evaluating the integer offline benchmark at a noninteger arrival count. The certificate is a convex combination of two numbers in $[0,1]$, so $0\le\eta_D(\gamma)\le1$. The same guarantee therefore holds at $\ell=0$ and $\ell=1$ by initialization, and for $\ell>B$ because all capacities have been filled and the policy attains $U(B)=U(\ell)$.

Averaging the pointwise guarantee over the arrival distribution gives, for every integer $\gamma\ge2$,
\[
\Rev^{\pi(\gamma)}
\ge\eta_D(\gamma)\mathbb E[U(L)]
=\eta_D(\gamma)\OPT_{\mathrm{off}}.
\]

Finally, we evaluate the minimum defining $\eta_D(\gamma)$ in Theorem~\ref{thm:general-d-cr}. For $\gamma=3$, the objective on $t\in[1,2]$ is
\[
\frac{81t^4-162t^3+45t^2+72t+92}{288t^2},
\]
which is minimized at $t^*\approx1.5177774358$, the unique root of $81t^4-81t^3-36t-92=0$ in $[1,2]$. On $[2,1+4\sqrt2/3]$, the objective equals $1-2/t+13/(9t^2)\ge13/36$; on $[1+4\sqrt2/3,3]$, it equals $4/t^2\ge4/9$. Thus, the global minimum is attained at $t^*$, yielding $\eta_D(3)=0.2537847484\ldots>0.2537$. Consequently,
\[
\Rev^{\pi(3)}
\ge\eta_D(3)\OPT_{\mathrm{off}}
\ge0.2537\,\OPT_{\mathrm{off}},
\]
completing the proof. \Halmos
\endproof

\subsubsection{Supporting Material for Section~\ref{sec:dc-general-random}.}
\label{app:dc-gen-r}

Under Assumption~\ref{assump:discrete-concavity}, retain the notation of Section~\ref{sec:dc-general-random}. We give the fractional construction, extend the bound on benchmark attainment with accumulated reward, and describe integral implementation before proving the performance guarantee.

\paragraph{Fractional benchmark and backward-greedy ordering.}

Use the fractional extensions of $R_i$, $V$, and $U$ defined in Section~\ref{sec:dc-general-random}. In particular, the linear interpolation of each cumulative reward function is
\[
R_i(k+\delta)=(1-\delta)R_i(k)+\delta R_i(k+1),
\qquad k=0,\ldots,b_i-1,\quad \delta\in[0,1].
\]
The interpolated rewards $R_i$ are nondecreasing and convex. 
We first verify that the fractional benchmark agrees with the original offline benchmark at integer arrival counts. For any integer $n\le B$, the feasible polytope in~\eqref{eq:offline-opt-frac} has only integral extreme points. Indeed, if the total-size constraint is slack, any fractional coordinate can be perturbed in both directions, contradicting extremality. If the constraint is tight, two fractional coordinates can be perturbed in opposite directions while preserving feasibility. A single fractional coordinate is also impossible because the total size and all other coordinates are integers. Thus, every extreme point is integral. Since $V$ is convex, an optimal solution exists at an extreme point, so the fractional and original benchmarks coincide. For $n\ge B$, both equal $V(\boldsymbol b)$. Finally, monotonicity of $V$ ensures that an optimizer of total size $\min\{n,B\}$ exists for every real $n\ge0$.

Since the interpolated $R_i$ is linear between consecutive integers, fractional backward greedy partitions the planned allocation of each product into constant-marginal blocks. Each product has at most $\lfloor a_i\rfloor+2$ such blocks: the first and last blocks may have fractional lengths, while every intermediate block has unit length.

\begin{lemma}[Quadratic prefix-value guarantee for fractional allocations]
\label{lem:general-prefix}
Under Assumption~\ref{assump:discrete-concavity}, fix a feasible fractional allocation $\boldsymbol s$ and a feasible additional allocation plan $\boldsymbol a$, with $0\le a_i\le b_i-s_i$ for every $i\in[m]$ and $h_a=\|\boldsymbol a\|_1>0$. Construct a backward-greedy execution order backward, starting from $\boldsymbol a^{(h_a)}=\boldsymbol a$. At each remaining size $h>0$, select
\[
\omega_h\in\argmin_{i:\,a_i^{(h)}>0}
\left[r_i+f_i\!\left(\left\lceil s_i+a_i^{(h)}\right\rceil\right)\right]
\]
and decrease the selected product's allocation until $s_{\omega_h}+a_{\omega_h}^{(h)}$ reaches the next lower integer level or $s_{\omega_h}$, whichever comes first. Reversing these deletions gives a fractional backward-greedy execution. Let $\boldsymbol a^{(h)}$ denote the additional product counts in its first $h$ units, with $\|\boldsymbol a^{(h)}\|_1=h$. Then, for every real $h\in[0,h_a]$,
\begin{equation}\label{eq:general-prefix}
V(\boldsymbol s+\boldsymbol a^{(h)})-V(\boldsymbol s)
\ge\left(\frac{h}{h_a}\right)^2
\bigl[V(\boldsymbol s+\boldsymbol a)-V(\boldsymbol s)\bigr].
\end{equation}
\end{lemma}

\proof{Proof of Lemma~\ref{lem:general-prefix}:}
When $h=0$, both sides are zero. For $h>0$, we extend the single-product bound from Lemma~\ref{lemma:quad-prefix} to fractional counts and use it to control the endpoint marginal selected by backward greedy. We then compare the incremental reward at two sizes within each constant-marginal block and join these comparisons across the backward path by continuity.

Fix a product $i$, a fractional base allocation $s_i$, and a feasible additional amount $a_i>0$. Let $k=\lceil s_i+a_i\rceil-1$, so that $s_i+a_i\in(k,k+1]$. On this interval, $R_i$ is linear with slope $r_i+f_i(k+1)$. Combining this interpolation with convexity and the integer bound \eqref{eq:one-prod-bound} gives
\begin{equation}
\label{eq:one-prod-bound-frac}
\begin{aligned}
R_i(s_i+a_i)-R_i(s_i)
&\ge\frac{a_i}{s_i+a_i}R_i(s_i+a_i)
=\frac{a_i}{s_i+a_i}\left[R_i(k+1)-(k+1-s_i-a_i)\bigl[r_i+f_i(k+1)\bigr]\right]\\
&\ge\frac{a_i}{s_i+a_i}\left[\frac{k+2}{2}-(k+1-s_i-a_i)\right]\bigl[r_i+f_i(k+1)\bigr]\\
&=\frac{a_i}{s_i+a_i}\left(s_i+a_i-\frac{k}{2}\right)\bigl[r_i+f_i(k+1)\bigr]\\
&\ge\frac{a_i}{2}\bigl[r_i+f_i(k+1)\bigr]
=\frac{a_i}{2}\left[r_i+f_i\!\left(\left\lceil s_i+a_i\right\rceil\right)\right].
\end{aligned}
\end{equation}
The first inequality uses convexity and $R_i(0)=0$, which imply $R_i(s_i)\le[s_i/(s_i+a_i)]R_i(s_i+a_i)$. The next inequality applies \eqref{eq:one-prod-bound} with base allocation zero and additional count $k+1$, and the last uses $s_i+a_i\ge k$.

Apply \eqref{eq:one-prod-bound-frac} to the remaining plan $\boldsymbol a^{(h)}$ at any size $h>0$. As in the integer proof, backward greedy selects the smallest endpoint marginal among the active products. Since $\sum_i a_i^{(h)}=h$, we obtain
\begin{equation}
\label{eq:delta-bound-frac}
\begin{aligned}
V(\boldsymbol s+\boldsymbol a^{(h)})-V(\boldsymbol s)
&=\sum_{i\in[m]}\bigl[R_i(s_i+a_i^{(h)})-R_i(s_i)\bigr]
\ge\frac12\sum_{i:\,a_i^{(h)}>0}a_i^{(h)}
\left[r_i+f_i\!\left(\left\lceil s_i+a_i^{(h)}\right\rceil\right)\right]\\
&\ge\frac12\sum_{i:\,a_i^{(h)}>0}a_i^{(h)}
\left[r_{\omega_h}+f_{\omega_h}\!\left(\left\lceil s_{\omega_h}+a_{\omega_h}^{(h)}\right\rceil\right)\right]
=\frac h2 \left[r_{\omega_h}+f_{\omega_h}\!\left(\left\lceil s_{\omega_h}+a_{\omega_h}^{(h)}\right\rceil\right)\right].
\end{aligned}
\end{equation}
Thus, the selected endpoint marginal is at most $2/h$ times the remaining plan's incremental reward.

Within each block of the backward construction, the selected product and its endpoint marginal remain constant. Take $0<h_1<h_2\le h_a$ in the interior of the same block. The change in incremental reward between these sizes is therefore the additional service $h_2-h_1$ times this marginal. Applying \eqref{eq:delta-bound-frac} at size $h_1$ gives
\[
\begin{aligned}
V(\boldsymbol s+\boldsymbol a^{(h_2)})-V(\boldsymbol s)
&=\bigl[V(\boldsymbol s+\boldsymbol a^{(h_1)})-V(\boldsymbol s)\bigr]
+(h_2-h_1)\left[r_{\omega_{h_1}}+f_{\omega_{h_1}}\!\left(\left\lceil s_{\omega_{h_1}}+a_{\omega_{h_1}}^{(h_1)}\right\rceil\right)\right]\\
&\le\left[1+\frac{2(h_2-h_1)}{h_1}\right]\bigl[V(\boldsymbol s+\boldsymbol a^{(h_1)})-V(\boldsymbol s)\bigr]\\
&=\left(2\frac{h_2}{h_1}-1\right)\bigl[V(\boldsymbol s+\boldsymbol a^{(h_1)})-V(\boldsymbol s)\bigr]
\le\left(\frac{h_2}{h_1}\right)^2\bigl[V(\boldsymbol s+\boldsymbol a^{(h_1)})-V(\boldsymbol s)\bigr],
\end{aligned}
\]
where the last inequality follows from $(h_2/h_1-1)^2\ge0$ and the nonnegativity of the incremental reward. Dividing by $h_2^2$ yields
\[
\frac{V(\boldsymbol s+\boldsymbol a^{(h_1)})-V(\boldsymbol s)}{h_1^2}
\ge\frac{V(\boldsymbol s+\boldsymbol a^{(h_2)})-V(\boldsymbol s)}{h_2^2}.
\]
Hence the incremental reward divided by the square of the remaining size is nonincreasing within each block. The allocation path and interpolated reward functions are continuous, so this ratio is continuous at every positive block boundary. The comparison therefore extends to the block endpoints and can be chained across the entire backward path. For every $0<h\le h_a$, comparison with the full plan $\boldsymbol a^{(h_a)}=\boldsymbol a$ gives
\[
\frac{V(\boldsymbol s+\boldsymbol a^{(h)})-V(\boldsymbol s)}{h^2}
\ge\frac{V(\boldsymbol s+\boldsymbol a)-V(\boldsymbol s)}{h_a^2}.
\]
Multiplying by $h^2$ proves \eqref{eq:general-prefix}, completing the proof. \Halmos
\endproof

The bound in Lemma~\ref{lem:general-value-repair} extends to feasible fractional allocations $\boldsymbol s$ satisfying $V(\boldsymbol s)\ge U(\tilde{n})$, real prospective arrival counts $0\le \tilde{n}<n$, the additional allocation plan $\boldsymbol a(n)=(\boldsymbol x^*(n)-\boldsymbol s)^+$, and every real prefix length $0\le h\le\|\boldsymbol a(n)\|_1$. Use Lemma~\ref{lem:general-prefix} in place of the discrete quadratic prefix bound in its proof. The overlap allocation $\underline{\boldsymbol x}=\boldsymbol s\wedge\boldsymbol x^*(n)$ and its size $h_{\underline x}$ are now real. When $h_{\underline x}\le \tilde{n}$, the backward-greedy prefix of length $\tilde{n}-h_{\underline x}$ is feasible for the fractional benchmark at $\tilde{n}$, so the same comparison with $U(\tilde{n})$ applies. When $h_{\underline x}>\tilde{n}$, the same inequalities apply. Subsequent algebra, including the case $n>B$, is unchanged. Consequently, \eqref{eq:general-value-repair} holds for the fractional construction as well. This extension uses no property of the random shift.

\paragraph{Bipartite dependent rounding.}
We implement the joint rounding procedure described in Section~\ref{sec:dc-general-random} by randomized cycle updates. Fix the shift $Z$. For $\ell=1,\ldots,B$ and $i\in[m]$, let $p_{\ell i}:=s_i(\ell)-s_i(\ell-1)$ be the fractional assignment at position $\ell$ to product $i$, and form the matrix $P=(p_{\ell i})$. Since the fractional path allocates one unit between successive integer service counts and fills all capacities at service level $B$, this matrix has unit row sums and column sums $\boldsymbol b$: $\sum_{i\in[m]}p_{\ell i}=1$ and $\sum_{\ell=1}^B p_{\ell i}=b_i$.

While fractional entries remain, construct the bipartite graph with arrival positions and products as vertices and an edge for each entry strictly between zero and one. Every vertex incident to a fractional entry has at least two such entries, because its incident entries sum to an integer. Hence this graph contains an even cycle. Along the cycle, alternately increase and decrease entries by a common amount, preserving every row and column sum. In each of the two possible directions, continue until an entry first reaches zero or one, obtaining two endpoint matrices. Replace the current matrix by one of these endpoints at random: the probability of each endpoint is the distance from the current matrix to the opposite endpoint divided by the sum of the two distances, where each distance is the common amount by which a cycle entry changes. These probabilities make the weighted average of the endpoint matrices equal to the current matrix, so the update preserves the conditional expectation of every entry.

Each update makes at least one fractional entry integral and leaves all previously integral entries unchanged. Thus the procedure terminates after at most $Bm$ updates with a binary matrix $\widehat P=(\widehat p_{\ell i})$ having the same row and column sums. The rounded prefix counts satisfy $\widehat s_i(\ell)=\sum_{k=1}^{\ell}\widehat p_{ki}$, with $\widehat s_i(0)=0$. Unit row sums ensure that the resulting path adds exactly one assignment at each step, and column sums $\boldsymbol b$ ensure that every prefix respects capacity. Successive conditional expectation preservation gives
\[
\mathbb E[\widehat p_{\ell i}\mid Z]=p_{\ell i},
\qquad \ell=1,\ldots,B,\quad i\in[m].
\]
Summing over the first $\ell$ assignments yields $\mathbb E[\widehat s_i(\ell)\mid Z]=s_i(\ell)$ for every product $i$ and integer $\ell\in\{0,\ldots,B\}$.

\paragraph{Computational complexity.}
We finally verify that Algorithm~\ref{alg:general-repair-rounding} admits a polynomial-time implementation. For $B\ge1$, the randomly shifted geometric grid requires $O(1+\log_\gamma B)$ computations of offline optimizers before capacity is exhausted. Each optimizer of \eqref{eq:offline-opt-frac} can be chosen at an extreme point with at most one fractional coordinate. By enumerating the product and integer part of this coordinate, the optimization reduces to the capacity-indexed dynamic program described in Appendix~\ref{app:benchmark-offline}. For backward-greedy execution, the piecewise linear rewards yield at most $B+2m$ constant-marginal blocks per epoch, which can be ordered by backward-greedy selection. Finally, the joint rounding procedure operates on a $B\times m$ matrix and terminates after at most $Bm$ iterations, each requiring polynomial time. Therefore, the overall algorithm runs in polynomial time under the explicit-capacity representation and ideal real arithmetic.

\paragraph{Performance guarantee.}

\proof{Proof of Theorem~\ref{thm:general-r-cr}:}
Fix a real $\gamma>1$ satisfying $\gamma>\lambda(\gamma)$. If $B=0$, both the policy value and the offline benchmark are zero. Assume $B\ge1$ and fix the shift $Z$. Following the deterministic proof, we first establish the fractional pointwise guarantee for epochs $j>j_0$ beginning with a secured preceding benchmark. We then verify initialization and the boundary cases, average over the shift, and transfer the guarantee to the integral policy.

Fix an epoch $j>j_0$ starting with $V(\boldsymbol s(T_{j-1}))\ge U(n_{j-1})$. As in the deterministic proof, the starting value is at most $U(n_j)$: feasibility gives $V(\boldsymbol s(T_{j-1}))\le U(T_{j-1})$, and $T_{j-1}<n_j$ follows from $\gamma>\lambda(\gamma)$. For every real amount of additional service $0\le h\le\min\{T_j,B\}-T_{j-1}$, the fractional extension of \eqref{eq:general-value-repair}, with $\tilde{n}=n_{j-1}$ and $n=n_j$, therefore gives the first inequality in \eqref{eq:general-det-reward} for the fractional path at service level $T_{j-1}+h$. Under the shortened fractional schedule,
\[
(T_j-T_{j-1})^2=n_j^2-n_{j-1}^2,
\qquad
\frac{h^2}{n_j^2-n_{j-1}^2}=\frac{h^2}{(T_j-T_{j-1})^2}\le1,
\]
so the cap on the progress coefficient is inactive. Empty or completed plans and epochs truncated by capacity are covered exactly as in the deterministic proof. Taking $h=T_j-T_{j-1}$ whenever $T_j\le B$ makes the coefficient one and secures $U(n_j)$. Thus, the benchmark-attainment bound preserves the same benchmark invariant under the shorter schedule.

To compare with the realized offline benchmark, fix an integer $1\le\ell\le B$ in the epoch. Lemma~\ref{lem:general-prefix} extends the square bound in \eqref{eq:offline-value-growth} to fractional target levels. Using this extension and monotonicity in the first inequality of \eqref{eq:general-det-reward} gives the same convex-combination calculation as in the deterministic proof. Using the relative stopping position $t$ from the main text, the substitutions $n_{j-1}=T_{j-1}/\lambda(\gamma)$, $n_j=\gamma T_{j-1}/\lambda(\gamma)$, and $T_j=\gamma T_{j-1}$ yield $V(\boldsymbol s(\ell))\ge \eta_R(\gamma,t)U(\ell)$.

It remains to verify the first active epoch. By the definition of $j_0$, we have $T_{j_0-1}\le1<T_{j_0}$ and $n_{j_0-1}<T_{j_0-1}\le1$. Since the fractional benchmark is linear below one, initialization gives
\[
V(\boldsymbol s(T_{j_0-1}))=T_{j_0-1}U(1)\ge n_{j_0-1}U(1)=U(n_{j_0-1}).
\]
This supplies the initial secured benchmark, so the preceding argument applies to epoch $j_0$ and inductively to every subsequent epoch. If $T_{j_0-1}<1$, the stopping count $\ell=1$ lies inside the first active epoch and is covered by that argument. If $\ell=T_{j_0-1}=1$ is the initialization endpoint, then $V(\boldsymbol s(1))=U(1)$ and $\eta_R(\gamma,1)=n_{j_0-1}^2<1$, so the pointwise bound holds there as well. The performance certificate $\eta_R(\gamma,t)$ is a convex combination of two numbers in $[0,1]$, and hence its shift average satisfies $0\le\eta_R(\gamma)\le1$. Both the fractional path and its integral implementation have value zero at $\ell=0$ and attain $U(B)=U(\ell)$ for $\ell>B$, so these boundary cases also satisfy the claimed guarantee.

We now evaluate the shift average for a fixed integer stopping count $1\le\ell\le B$. Since $T_j=\lambda(\gamma)\gamma^{j+Z}$, the relative stopping position satisfies $\log_\gamma t=\log_\gamma(\ell/\lambda(\gamma))-Z-(j-1)$.
Translating a uniformly shifted unit interval modulo integers preserves lengths, so $\log_\gamma t$ is uniform on $(0,1]$ and $t$ has density $1/(t\log\gamma)$ on $(1,\gamma]$, independently of $\ell$. Shifts placing $\ell$ exactly at a checkpoint have probability zero and do not affect this calculation; the initialization endpoint was covered above. Averaging the performance certificate with this density gives $\eta_R(\gamma)$ by definition.
Since $U(\ell)$ is independent of $Z$, taking expectations in the fractional pointwise bound gives $
\mathbb E_{Z}[V(\boldsymbol s(\ell))]
\ge \mathbb E_{Z}[\eta_R(\gamma,t)]U(\ell)
=\eta_R(\gamma)U(\ell)$.

To transfer this bound to the integral policy, condition on $Z$. The joint rounding construction above produces a feasible integral allocation path $\widehat{\boldsymbol s}(\ell)$ satisfying $\E[\widehat s_i(\ell)\mid Z]=s_i(\ell)$ for every product $i$ and integer $\ell=0,\ldots,B$. By the convexity of $R_i$, we apply Jensen's inequality to obtain, for every integer $1\le\ell\le B$,
\begin{equation*}
\begin{aligned}
\Rev^{\pi_R(\gamma)}(\ell)
&=\E_{Z}\left[\E[V(\widehat{\boldsymbol s}(\ell))\mid Z]\right]
=\E_{Z}\left[\sum_{i\in[m]}\E[R_i(\widehat s_i(\ell))\mid Z]\right]\\
&\ge\E_{Z}\left[\sum_{i\in[m]}R_i\bigl(\E[\widehat s_i(\ell)\mid Z]\bigr)\right]
=\E_{Z}[V(\boldsymbol s(\ell))]\ge\eta_R(\gamma)U(\ell).
\end{aligned}
\end{equation*}
Together with the boundary cases above, this proves the pointwise guarantee for every integer $\ell\ge0$. Averaging the pointwise guarantee over $L\sim\mathcal D_L$ yields, for every admissible $\gamma$,
\[
\begin{aligned}
\Rev^{\pi_R(\gamma)}
&=\E_{L\sim\mathcal D_L}[\Rev^{\pi_R(\gamma)}(L)]
=\sum_{\ell\ge0}\P(L=\ell)\Rev^{\pi_R(\gamma)}(\ell)
\ge\eta_R(\gamma)\E_{L\sim\mathcal D_L}[U(L)]
=\eta_R(\gamma)\OPT_{\mathrm{off}}.
\end{aligned}
\]

Finally, we evaluate the integral defining $\eta_R(\gamma)$ in Theorem~\ref{thm:general-r-cr}. Since $1<\gamma/\lambda(\gamma)<\gamma$, splitting the integral at $t=\gamma/\lambda(\gamma)$ yields
\begin{equation*}
\begin{aligned}
\mathbb E_{Z}\bigl[\eta_R(\gamma,t)\bigr]
&=\frac1{\log\gamma}\Biggl\{
\int_1^{\gamma/\lambda(\gamma)}
\left[\frac1{\lambda(\gamma)^2t^2}
+\frac{(t-1)^2}{(\gamma-1)^2}\left(1-\frac1{\lambda(\gamma)^2t^2}\right)\right]\frac{dt}{t}\\
&\hspace{1.2cm}+\int_{\gamma/\lambda(\gamma)}^\gamma
\left[\frac1{\lambda(\gamma)^2t^2}
+\frac{(t-1)^2}{(\gamma-1)^2}\frac{\gamma^2-1}{\lambda(\gamma)^2t^2}\right]\frac{dt}{t}\Biggr\}\\
&=\frac1{\log \gamma}\biggl[\frac{(\gamma^2+1)^2-4\gamma^2(\gamma-1)\lambda(\gamma)}{\gamma(\gamma-1)^2(\gamma+1)}+\frac{2\log \gamma}{(\gamma-1)^2(\gamma+1)}+\frac{\gamma^3-\gamma^2-\gamma-1}{(\gamma-1)^2(\gamma+1)}\log\lambda(\gamma)\biggr]\\
&=\eta_R(\gamma).
\end{aligned}
\end{equation*}
Choosing $\gamma=22/5$ and substituting $\lambda(\gamma)=\sqrt{27/17}$ into the expression for $\eta_R(\gamma)$ gives $\eta_R(22/5)=0.3864058297\ldots>0.3864$, so $\pi_R(22/5)$ is $0.3864$-competitive. This completes the proof. \Halmos
\endproof

\subsubsection{An Upper Bound under General Arrivals.}
\label{app:dc-gen-upper}

\begin{proposition}[Upper bound under general arrivals]\label{prop:dc-upper-general}
Under Assumption~\ref{assump:discrete-concavity}, no online algorithm can achieve a competitive ratio greater than $0.5703$ under general arrival distributions, even with randomization and knowledge of the arrival distribution.
\end{proposition}

\proof{Proof of Proposition~\ref{prop:dc-upper-general}:}
We construct a family of instances with truncated power-law arrivals and products of varying capacities. The power-law distribution maintains uncertainty across multiple arrival scales, while the bonus functions balance the rewards from filling products of different capacities against their probabilities of being filled.

Specifically, for $m\ge3$, fix $1<\alpha<2$ to be the exponent of the power-law distribution. Consider $m$ products with capacities $b_i=i$, zero base rewards, and linear bonus functions \(f_i(k)=i^{\alpha-2}k\) for $i\in[m]$. The arrival count $L$ follows a truncated power-law distribution with \(\P(L\ge\ell)=\ell^{-\alpha}\) for $1\le\ell\le m$,
and $\P(L>m)=0$. All bonus functions are linear and hence satisfy Assumption~\ref{assump:discrete-concavity}.

We first lower bound the offline optimum.
For any realization $L=\ell$, assigning all arrivals
to product $i=\ell$ is feasible and yields
\(
U(\ell)
\ge \frac{\ell(\ell+1)}{2\ell^{2-\alpha}}
\ge \frac{\ell^\alpha}{2}.
\)
The tail-sum formula gives $\E[L^\alpha]=\sum_{\ell=1}^m[1-(1-1/\ell)^\alpha]$. Using $1-(1-x)^\alpha\ge\alpha x-(\alpha-1)x^2$ for $x\in[0,1]$ yields
\[
\begin{aligned}
\E[L^\alpha]
&\ge
\alpha\sum_{\ell=1}^m\frac1\ell
-(\alpha-1)\sum_{\ell=1}^m\frac1{\ell^2} \ge
\alpha\log(m+1)-2(\alpha-1).
\end{aligned}
\]
Consequently,
\begin{equation}
\label{eq:upper-general-offline}
\OPT_{\mathrm{off}}
\ge
\frac{\alpha}{2}\log(m+1)-(\alpha-1).
\end{equation}


We next upper bound the expected reward of an online policy.
It suffices to consider deterministic policies, since the expected
reward of a randomized policy is an average of the expected rewards
of deterministic policies. Consider any deterministic online policy. On the arrival history
in which all $m$ customers appear, let $A_i\subseteq[m]$ be the
set of arrival positions assigned to product $i$.
Because the policy does not know the realized value of $L$,
its decisions when $L=\ell$ are the first $\ell$ decisions along
this history. For $\ell\in A_i$, let
\(
k_i(\ell)=|A_i\cap[\ell]|
\)
denote the number of customers assigned to product $i$
up to and including arrival $\ell$. 
The policy's expected reward is therefore
\begin{equation}
\label{eq:upper-general-online-sum}
\Rev^\pi
=\sum_{i=1}^{m}i^{\alpha-2}
\sum_{\ell\in A_i}k_i(\ell)\ell^{-\alpha}.
\end{equation}



To bound the online reward, define
\(
\Lambda_\alpha:=\sup_{v>0}
\frac{\int_0^1\frac{u}{(v+u)^\alpha}\,du}
{\log(1+1/v)}.
\)
Applying Lemma~\ref{lem:upper-general-price} to each product with $b=i$ and
$A=A_i$, and noting that the reward from the first arrival is at most $1$, we obtain
\[
\begin{aligned}
\Rev^\pi
&\le 1+\Lambda_\alpha\sum_{i=1}^m\sum_{\ell\in A_i,\ \ell\ge2}\log\frac{\ell}{\ell-1}+2\sum_{i=1}^m i^{\alpha-2}\sum_{\ell\in A_i,\ \ell\ge2}\int_{\ell-1}^{\ell}t^{-\alpha}\,dt\\
&\le 1+\Lambda_\alpha\sum_{\ell=2}^m\log\frac{\ell}{\ell-1}+2\sum_{\ell=2}^m\int_{\ell-1}^{\ell}t^{-\alpha}\,dt
\le 1+\Lambda_\alpha\log m+\frac{2}{\alpha-1}.
\end{aligned}
\]
The second inequality follows because each arrival is
assigned to at most one product and
$i^{\alpha-2}\le1$.

Combining this upper bound on the online value with the offline lower bound in \eqref{eq:upper-general-offline}, we obtain
\[
\frac{\OPT_{\mathrm{on}}}{\OPT_{\mathrm{off}}}
\le
\frac{\Lambda_\alpha\log m+1+2/(\alpha-1)}
{\frac{\alpha}{2}\log(m+1)-(\alpha-1)}
\longrightarrow
\frac{2\Lambda_\alpha}{\alpha}
\quad\text{as }m\to\infty.
\]

To evaluate this bound, direct integration gives
\[
\int_0^1\frac{u}{(v+u)^\alpha}\,du
=
\frac{(v+\alpha-1)(1+v)^{1-\alpha}-v^{2-\alpha}}
{(\alpha-1)(2-\alpha)}.
\]
Choosing $\alpha=25/16$ and evaluating the
one-dimensional supremum defining $\Lambda_\alpha$ yields
\(
\Lambda_{25/16}\approx0.4455152532, 
\frac{2\Lambda_{25/16}}{25/16}
\approx0.5702595242<0.57026.
\)
Thus, for every $\varepsilon>0$, a sufficiently large finite
instance satisfies
\(
\frac{\OPT_{\mathrm{on}}}{\OPT_{\mathrm{off}}}
<0.57026+\varepsilon,
\)
which establishes the claimed upper bound.
\Halmos
\endproof

We complete the proof by establishing the single-product
inequality used in the online bound.


\begin{lemma}
\label{lem:upper-general-price}
Fix $1<\alpha<2$, and let $\Lambda_\alpha$ be the supremum defined above.
Consider a product with capacity $b$ and let
$A\subseteq[m]$, $|A|\le b$, be the set of arrival positions
assigned to this product.
For each $\ell\in A$, define
$k(\ell)=|A\cap[\ell]|$, the number of assignments made to
the product up to and including arrival $\ell$.
Then
\(b^{\alpha-2}
\sum_{\ell\in A,\ \ell\ge2}
k(\ell)\ell^{-\alpha}\quad\le
\Lambda_\alpha
\sum_{\ell\in A,\ \ell\ge2}
\log\frac{\ell}{\ell-1}
+
2b^{\alpha-2}
\sum_{\ell\in A,\ \ell\ge2}
\int_{\ell-1}^{\ell}t^{-\alpha}\,dt.
\)
\end{lemma}

\proof{Proof of Lemma~\ref{lem:upper-general-price}:}
We first bound the discrete reward by a continuous integral.
For each assignment at position $\ell\ge2$, consider the
interval $[\ell-1,\ell]$, and let
\(
S=\bigcup_{\ell\in A,\ \ell\ge2}[\ell-1,\ell].
\)
Define
\(
q(t)=\int_1^t\ind_{\{u\in S\}}\,du,
\)
which measures the cumulative length of these intervals up
to $t$.

For every $\ell\in A$ with $\ell\ge2$ and
$t\in[\ell-1,\ell]$, we have
\(
k(\ell)\le q(t)+2,
\ell^{-\alpha}\le t^{-\alpha}.
\)
The additional $2$ accounts for a possible assignment at
position $1$ and the current assignment.
Consequently,
\begin{equation}
\label{eq:single-product-discrete-continuous}
\begin{aligned}
&b^{\alpha-2}
\sum_{\ell\in A,\ \ell\ge2}
k(\ell)\ell^{-\alpha}\le
b^{\alpha-2}\int_Sq(t)t^{-\alpha}\,dt
+
2b^{\alpha-2}\int_St^{-\alpha}\,dt.
\end{aligned}
\end{equation}

It remains to show that
\begin{equation}
\label{eq:single-product-continuous}
b^{\alpha-2}\int_Sq(t)t^{-\alpha}\,dt
\le\Lambda_\alpha\int_S\frac{dt}{t}.
\end{equation}

First, suppose $S=[a,a+h]$ for some $a\ge1$ and
$0<h\le b$. In this case, $q(t)=t-a$ for $t\in S$.
Setting $v=a/h$ and substituting $t=h(v+u)$ gives
\(
b^{\alpha-2}\int_Sq(t)t^{-\alpha}\,dt
=\left(\frac{h}{b}\right)^{2-\alpha}
\int_0^1\frac{u}{(v+u)^\alpha}\,du \le\Lambda_\alpha\log(1+1/v)
=\Lambda_\alpha\int_S\frac{dt}{t},
\)
where the inequality follows from $h\le b$ and the
definition of $\Lambda_\alpha$.

We next extend the inequality to any finite union of
intervals. Define
\(
F(S)=
b^{\alpha-2}\int_Sq(t)t^{-\alpha}\,dt
-\Lambda_\alpha\int_S\frac{dt}{t}.
\)
We show that any gaps between the intervals in $S$
can be eliminated without decreasing $F(S)$.
The resulting set is a single interval, for which
we have already established $F(S)\le0$.

Consider two consecutive intervals $[a,x]$ and
$[y,c]$ in $S$, where $a<x<y<c$.
Hold $a$, $c$, and the gap length $g=y-x$ fixed,
and vary $x$ with $y=x+g$.
This shifts the gap without changing the total length
of $S$ or the value of $q(t)$ for $t\ge c$.

Let $Q=q(x)=q(y)$. Differentiating $F(S)$
with respect to $x$ gives
\(
\frac{dF}{dx}
=b^{\alpha-2}Q(x^{-\alpha}-y^{-\alpha})
-\Lambda_\alpha(x^{-1}-y^{-1})=(x^{-1}-y^{-1})
\left[
b^{\alpha-2}Qx^{1-\alpha}
\frac{1-z^\alpha}{1-z}
-\Lambda_\alpha
\right],
\)
where $z=x/y\in(0,1)$. The expression in brackets is strictly increasing in $x$.
Indeed, $dQ/dx=1$ and $Q\le x$ imply
\(
\frac{d}{dx}(Qx^{1-\alpha})
=x^{-\alpha}[x-(\alpha-1)Q]
\ge(2-\alpha)x^{1-\alpha}>0.
\)
Furthermore, $z=x/(x+g)$ increases with $x$, and
$(1-z^\alpha)/(1-z)$ is increasing in $z$
by convexity of $z^\alpha$.

Since $x^{-1}-y^{-1}>0$, the derivative $dF/dx$
can change sign only from negative to positive.
Hence, $F(S)$ is maximized at an endpoint of the
feasible range of $x$, where one of the two intervals
disappears. Repeating this argument eliminates all
gaps without decreasing $F(S)$. The resulting set is a single interval, for which
$F(S)\le0$ was established above.
This proves \eqref{eq:single-product-continuous}.
\Halmos
\endproof

\subsection{Supporting Material for Section~\ref{sec:dc-ifr}}
\label{app:dc-ifr}

We first prove the residual-value comparison and the IFR planning-score guarantee, then give an IFR upper-bound construction and extend the planning-score analysis to general arrivals.

\subsubsection{Supporting Material for Section~\ref{sec:dc-ifr-plan}.}
\label{app:dc-ifr-plan}

\proof{Proof of Lemma~\ref{lem:ifr-residual-value}:}
We compare the optimal allocation plans $\boldsymbol a^*(k)$ and $\boldsymbol a^*(h)$ for $k$ and $h$ additional assignments, respectively. Since these plans need not be nested, we focus on the assignments unique to the longer plan. We use these assignments both to compare the two offline optima and to construct a feasible continuation after completing $\boldsymbol a^*(k)$.

Let $\underline{\boldsymbol a}$ denote the coordinatewise overlap between the two plans, and let $\boldsymbol u$ denote the assignments unique to the longer plan:
\[
    \underline{\boldsymbol a}:=\boldsymbol a^*(k)\wedge\boldsymbol a^*(h),\qquad h_{\underline a}:=\|\underline{\boldsymbol a}\|_1,\qquad \boldsymbol u:=\boldsymbol a^*(h)-\underline{\boldsymbol a},
\]
where $\wedge$ denotes the coordinatewise minimum. Then $0\le h_{\underline a}\le k$ and $\|\boldsymbol u\|_1=h-h_{\underline a}$. Let $\boldsymbol u^{(j)}$ denote the first $j$ assignments of $\boldsymbol u$ in backward-greedy order from $\boldsymbol s+\underline{\boldsymbol a}$, with $\boldsymbol u^{(0)}=\boldsymbol 0$.

We first compare the residual offline values for the two horizons. We have
\begin{equation*}
\begin{aligned}
    U_{\boldsymbol s}(h)-U_{\boldsymbol s}(k)
    &=V(\boldsymbol s+\boldsymbol a^*(h))-V(\boldsymbol s_k)
    \le V(\boldsymbol s+\boldsymbol a^*(h))-V(\boldsymbol s+\underline{\boldsymbol a}+\boldsymbol u^{(k-h_{\underline a})})\\
    &=[V(\boldsymbol s+\underline{\boldsymbol a}+\boldsymbol u)-V(\boldsymbol s+\underline{\boldsymbol a})]-[V(\boldsymbol s+\underline{\boldsymbol a}+\boldsymbol u^{(k-h_{\underline a})})-V(\boldsymbol s+\underline{\boldsymbol a})],
\end{aligned}
\end{equation*}
where the inequality follows from the optimality of $\boldsymbol a^*(k)$, since $\underline{\boldsymbol a}+\boldsymbol u^{(k-h_{\underline a})}$ is feasible for the shorter horizon; the second equality follows from the definition of $\boldsymbol{u}$.

To relate this comparison to the residual offline value after completing $\boldsymbol a^*(k)$, consider assigning any prefix of $\boldsymbol u$ after either $\boldsymbol s+\underline{\boldsymbol a}$ or $\boldsymbol s_k$. For every product $i$ with $u_i>0$, we have $\underline a_i=a_i^*(k)$. Thus, both starting allocations have the same number of assignments on product $i$, and assigning any additional units from $\boldsymbol u$ to this product earns the same marginal rewards in both cases. Since these assignments are also feasible from either allocation, every prefix $\boldsymbol u^{(j)}$ generates the same incremental reward:
\begin{equation*}
    \Delta V_j:=V(\boldsymbol s+\underline{\boldsymbol a}+\boldsymbol u^{(j)})-V(\boldsymbol s+\underline{\boldsymbol a})
    =V(\boldsymbol s_k+\boldsymbol u^{(j)})-V(\boldsymbol s_k).
\end{equation*}
Since $\boldsymbol u$ is ordered backward-greedily, Lemma~\ref{lemma:quad-prefix} guarantees that every prefix retains a controlled fraction of its total incremental reward:
\begin{equation}\label{eq:ifr-overlap-prefix}
    \Delta V_j\ge\frac{j(j+1)}{(h-h_{\underline a})(h-h_{\underline a}+1)}\Delta V_{h-h_{\underline a}},\qquad 0\le j\le h-h_{\underline a}.
\end{equation}

Since $\boldsymbol u^{(h-k)}$ provides a feasible continuation after completing $\boldsymbol a^*(k)$, we can combine the preceding comparison with the prefix-value guarantee in \eqref{eq:ifr-overlap-prefix} to obtain
\begin{equation*}
\begin{aligned}
    U_{\boldsymbol s}(h)-U_{\boldsymbol s}(k)
    &\le\Delta V_{h-h_{\underline a}}-\Delta V_{k-h_{\underline a}}
    \le\left[1-\frac{(k-h_{\underline a})(k-h_{\underline a}+1)}{(h-h_{\underline a})(h-h_{\underline a}+1)}\right]\Delta V_{h-h_{\underline a}}\\
    &=\frac{(h-k)(h+k-2h_{\underline a}+1)}{(h-h_{\underline a})(h-h_{\underline a}+1)}\Delta V_{h-h_{\underline a}}
    \le\frac{h+k-2h_{\underline a}+1}{h-k+1}\Delta V_{h-k}\\
    &\le\frac{h+k-2h_{\underline a}+1}{h-k+1}U_{\boldsymbol s_k}(h-k)
    \le\frac{h+k+1}{h-k+1}U_{\boldsymbol s_k}(h-k),
\end{aligned}
\end{equation*}
where the first inequality follows from the definition of $\Delta V_j$, the next two inequalities follow from \eqref{eq:ifr-overlap-prefix}, applied with $j=k-h_{\underline a}$ and $j=h-k$, respectively, and the remaining inequalities follow from the feasibility of the continuation prefix and $h_{\underline a}\ge0$. Rearranging gives
\[
    U_{\boldsymbol s_k}(h-k)
    \ge\frac{h-k+1}{h+k+1}
    [U_{\boldsymbol s}(h)-U_{\boldsymbol s}(k)].
\]

This completes the proof.\Halmos
\endproof

\subsubsection{Supporting Material for Section~\ref{sec:dc-ifr-cr}.}
\label{app:dc-ifr-cr}

\proof{Proof of Lemma~\ref{lem:ifr-scalar}:}
We compare the maximum planning score with the average score obtained by drawing the planning horizon independently from the residual arrival distribution. The main difficulty is that this comparison involves residual offline values at every possible horizon. The quadratic prefix-value guarantee allows us to express these values as a nonnegative combination of quadratic cutoff sequences. We first establish the comparison for the full quadratic sequence using an IFR moment inequality, then extend it to each cutoff sequence by conditioning below the cutoff and pairing outcomes on opposite sides. Linearity then yields the desired comparison for the residual offline values.

Fix an active epoch with current allocation $\boldsymbol s$ and $\rho\in(0,1/2]$. Throughout the proof, all probabilities and expectations are conditional on $L\ge h_s+1$. Let $\widetilde K$ be an independent copy of $K$ under this conditional law, and set $\underline{K}:=\min\{K,\widetilde K\}$. The random horizon $\widetilde K$ is used only for analysis; the algorithm selects a horizon maximizing the score.
For any sequence $u=(u_h)_{h=1}^N$, define the linear function
\begin{equation}\label{eq:ifr-score-gap}
\begin{aligned}
\Delta^\rho(u):=\E\Biggl[
    &\frac{\underline{K}(\underline{K}+1)}{\widetilde K(\widetilde K+1)}u_{\widetilde K}
    +\rho\frac{K-\underline{K}+1}{K+\underline{K}+1}(u_K-u_{\underline{K}})
    -\rho u_K
\Biggr].
\end{aligned}
\end{equation}
Applied to the residual offline values $U_{\boldsymbol s}(h)$, this function equals the average planning score minus the target potential. Indeed, when $K>\widetilde K$, we have $\underline{K}=\widetilde K$; otherwise, $\underline{K}=K$ and the continuation term vanishes. Since the maximum score is at least its average over $\widetilde K$, it suffices to prove
\begin{equation*}
    \E[\Psi_{\boldsymbol s}^{\rho}(\widetilde K)]-\rho\E[U_{\boldsymbol s}(K)]
    =\Delta^\rho(U_{\boldsymbol s})\ge0.
\end{equation*}

\emph{Reduction to Quadratic Cutoff Sequences.}
Applying \eqref{eq:offline-value-growth} to the residual offline problem shows that $v_h:=U_{\boldsymbol s}(h)/[h(h+1)]$ is nonincreasing. Set $v_{N+1}:=0$ and define the quadratic cutoff sequence $u_h^{(j)}:=h(h+1)\ind_{\{h\le j\}}$. Telescoping yields
\begin{equation*}
    U_{\boldsymbol s}(h)
    =h(h+1)\sum_{j=h}^{N}(v_j-v_{j+1})
    =\sum_{j=1}^{N}(v_j-v_{j+1})u_h^{(j)},
    \qquad 1\le h\le N.
\end{equation*}
Since $v_h$ is nonincreasing and $v_{N+1}=0$, all coefficients $v_j-v_{j+1}$ are nonnegative. By linearity, we obtain
$
    \Delta^\rho(U_{\boldsymbol s})
    =\sum_{j=1}^{N}(v_j-v_{j+1})\Delta^\rho(u^{(j)}).
$
Thus, it suffices to prove $\Delta^\rho(u^{(j)})\ge0$ for every cutoff $j$, even though these sequences need not correspond to feasible offline value functions.

\emph{Quadratic Comparison Under IFR.}
We first consider the full quadratic sequence $u_h=h(h+1)$, with $u_0=0$. Since $u_K-u_{\underline{K}}=(K-\underline{K})(K+\underline{K}+1)$, we can simplify \eqref{eq:ifr-score-gap} as
\begin{equation}\label{eq:ifr-quadratic-functional}
\begin{aligned}
    \Delta^\rho(u)
    &=\E\bigl[\underline{K}(\underline{K}+1)+\rho(K-\underline{K})(K-\underline{K}+1)-\rho K(K+1)\bigr]\\
    &=\E\bigl[(1+\rho)\underline{K}^2-2\rho K\underline{K}+(1-\rho)\underline{K}\bigr].
\end{aligned}
\end{equation}
To simplify the mixed moment, recall that $K-\underline{K}=(K-\widetilde K)_+$, where $K$ and $\widetilde K$ are independent and identically distributed. Symmetry gives
\begin{equation*}
    2\E[K\underline{K}]
    =\E[K^2+\underline{K}^2]-\E[(K-\widetilde K)_+^2]
    =\E[K^2+\underline{K}^2]-\tfrac12\E[(K-\widetilde K)^2]
    =\E[\underline{K}^2]+(\E[K])^2.
\end{equation*}
Substituting this into \eqref{eq:ifr-quadratic-functional}, we obtain
\begin{equation}\label{eq:ifr-quadratic-moments}
    \Delta^\rho(u)
    =\E[\underline{K}^2]+(1-\rho)\E[\underline{K}]-\rho(\E[K])^2.
\end{equation}

We now use the IFR property to establish $(\E[K])^2\le2\E[\underline{K}^2]$. Let $P_h:=\P(K\ge h)$, with $P_h=0$ for $h>N$, and write $p_h=P_h-P_{h+1}$. For $h<N$ with $P_h>0$, the continuation ratio $P_{h+1}/P_h$ equals the corresponding continuation probability for $L$ after $h_s+h$ arrivals. These ratios are nonincreasing by IFR, and the final ratio $P_{N+1}/P_N$ is zero whenever it is defined. Hence $K$ remains IFR and its survival sequence $(P_h)$ is log-concave. For any positive integers $i<j-1$, log-concavity implies $P_iP_j\le P_{i+1}P_{j-1}\le\cdots\le P_{\lfloor(i+j)/2\rfloor}P_{\lceil(i+j)/2\rceil}$.
Thus, among products with a fixed index sum, the largest occurs when the indices are closest. For $i+j=2h$, there are $2h-1$ positive-integer pairs, each satisfying $P_iP_j\le P_h^2$. Similarly, for $i+j=2h+1$, there are $2h$ pairs, each satisfying $P_iP_j\le P_hP_{h+1}$. Consequently,
\begin{equation}\label{eq:ifr-minimum-moment}
\begin{aligned}
(\E[K])^2
&=\sum_{i,j\ge1}P_iP_j
\le\sum_{h\ge1}\bigl[(2h-1)P_h^2+2hP_hP_{h+1}\bigr]
\le\sum_{h\ge1}\bigl[(2h-1)P_h^2+hP_h^2+hP_{h+1}^2\bigr]\\
&=2\sum_{h\ge1}(2h-1)P_h^2=2\E[\underline{K}^2],
\end{aligned}
\end{equation}
where the second inequality uses $2P_hP_{h+1}\le P_h^2+P_{h+1}^2$, the second equality follows by reindexing, and the last equality follows from $\P(\underline{K}\ge h)=P_h^2$ and the second-moment tail formula.

Substituting \eqref{eq:ifr-minimum-moment} into \eqref{eq:ifr-quadratic-moments} yields, for any $0<\rho\le1/2$,
\begin{equation}
\label{eq:ifr-quadratic-nonnegative}
\Delta^\rho(u)
\ge(1-2\rho)\E[\underline{K}^2]+(1-\rho)\E[\underline{K}]\ge0.
\end{equation}

\emph{Extension to Quadratic Cutoff Sequences.}
Fix a cutoff $j\in\{1,\ldots,N\}$. To establish $\Delta^\rho(u^{(j)})\ge0$, we distinguish whether $K$ and $\widetilde K$ lie on the same or opposite sides of the cutoff.

When both $K$ and $\widetilde K$ are at most $j$, the cutoff leaves their quadratic values unchanged. Conditional on this event, they remain independent copies of the truncated distribution $K\mid K\le j$. To apply \eqref{eq:ifr-quadratic-nonnegative}, we verify that this distribution remains IFR. Its hazard rate at any $h=1,\ldots,j$ is the probability of stopping at $h$, conditional on having reached $h$ without exceeding $j$:
\begin{equation*}
    \P(K=h\mid h\le K\le j)
    =\frac{P_h-P_{h+1}}{P_h-P_{j+1}}
    =\frac{P_h-P_{h+1}}{P_h}
      \left(1-\frac{P_{j+1}}{P_h}\right)^{-1},
\end{equation*}
whenever the denominator is positive. Both factors are nondecreasing in $h$, the first by IFR and the second because $P_h$ is nonincreasing. Thus, the truncated distribution remains IFR, and \eqref{eq:ifr-quadratic-nonnegative} shows that the contribution from $K,\widetilde K\le j$ is nonnegative. If $\P(K\le j)=0$, this contribution is zero. When both $K$ and $\widetilde K$ exceed $j$, all cutoff values vanish, so the contribution is also zero.

It remains to consider outcomes where $K$ and $\widetilde K$ lie on opposite sides of the cutoff $j$. We evaluate the expression inside the expectation in \eqref{eq:ifr-score-gap} for these two cases. When $K\le j<\widetilde K$, this expression reduces to $-\rho K(K+1)$. When $\widetilde K\le j<K$, it becomes
\[
\widetilde K(\widetilde K+1)
\left[1-\rho\frac{K-\widetilde K+1}{K+\widetilde K+1}\right]
\ge(1-\rho)\widetilde K(\widetilde K+1).
\]
Since $K$ and $\widetilde K$ are independent and identically distributed, symmetry gives
\[
\E\!\left[K(K+1)\ind_{\{K\le j<\widetilde K\}}\right]
=\E\!\left[\widetilde K(\widetilde K+1)\ind_{\{\widetilde K\le j<K\}}\right].
\]
Thus, the negative and positive terms from the two cases offset each other in expectation, leaving a nonnegative remainder when $\rho\le1/2$. Since the other cases also contribute nonnegatively, we obtain
\[
\begin{aligned}
\Delta^\rho(u^{(j)})
&\ge-\rho\,\E\!\left[K(K+1)\ind_{\{K\le j<\widetilde K\}}\right]+(1-\rho)\E\!\left[\widetilde K(\widetilde K+1)\ind_{\{\widetilde K\le j<K\}}\right]\\
&=(1-2\rho)\E\!\left[K(K+1)\ind_{\{K\le j<\widetilde K\}}\right]
\ge0,
\end{aligned}
\]
where the last inequality holds for $\rho\le1/2$.

Since $\Delta^\rho(u^{(j)})\ge0$ for every cutoff $j$, the nonnegative linear combination of these sequences representing $U_{\boldsymbol s}$ implies $\Delta^\rho(U_{\boldsymbol s})\ge0$. Consequently,
\[
\max_{1\le k\le N}\Psi_{\boldsymbol s}^{\rho}(k)
\ge\E[\Psi_{\boldsymbol s}^{\rho}(\widetilde K)]
\ge\rho\E[U_{\boldsymbol s}(K)]
=\Phi(\boldsymbol s).
\]
This completes the proof.\Halmos
\endproof



\subsubsection{An Upper Bound under IFR Arrivals.}
\label{app:dc-ifr-upper}

\begin{proposition}[Upper bound under IFR arrivals]
\label{prop:dc-upper-ifr}
Under Assumption~\ref{assump:discrete-concavity},
no online algorithm can achieve a competitive ratio
greater than $0.6608$ for IFR arrivals.
\end{proposition}
\proof{Proof of Proposition~\ref{prop:dc-upper-ifr}:}
We construct a family of instances indexed by $m\ge4$, with truncated geometric arrival counts.
Consider an instance with $m$ products, where product $i$ has capacity $b_i=i$ and base reward $r_i=0$. Let $q=1-m^{-1/2}$ and $B=m(m+1)/2$.
The arrival count $L$ follows a geometric distribution
truncated at $B$, with
\(
\P(L\ge\ell)=q^{\ell-1}, 1\le\ell\le B.
\)
The hazard rate is $1-q$ for $\ell<B$ and $1$ at $\ell=B$,
so $L$ is IFR.

For each product $i$, define the bonus function 
\(
f_i(k)
=
k\frac{\sum_{u=1}^{i}q^{u-1}}
{\sum_{u=1}^{i}u q^{u-1}}, k\ge0.
\)
Each $f_i$ is nonnegative, nondecreasing, and discrete concave.  
We first show that the optimal online value is at most the expected number of arrivals. Consider any online policy $\pi$. Fix its random seed and consider the assignments it would make if all $B$ customers arrived.
By construction, $\sum_{k=1}^{i}q^{k-1}(f_i(k)-1)=0$.
Since $f_i(k)-1$ is nondecreasing and its total weighted sum is zero, every weighted prefix sum is nonpositive, so every remaining weighted sum is nonnegative. Thus, 
\(
\sum_{k=s+1}^{i}q^{k-1}(f_i(k)-1)\ge0\) for all \(0\le s\le i.
\) For product $i$ with $s$ previous assignments,
define its conditional expected remaining excess reward as
\[
v_i(s):=
\sum_{k=s+1}^{i}q^{k-s-1}(f_i(k)-1)
=q^{-s}\sum_{k=s+1}^{i}q^{k-1}(f_i(k)-1)
\ge0.
\]
In particular, $v_i(0)=v_i(i)=0$.

Let $\boldsymbol s=(s_1,\ldots,s_m)$ denote the
current allocation state, where $s_i$ is the number
of customers assigned to product $i$.
We use $\sum_{i=1}^{m}v_i(s_i)$ as an auxiliary potential.
Suppose a customer is assigned to product $i$,
and let $\boldsymbol s'$ denote the resulting state.
Since only $s_i$ increases by one,
\begin{align*}
1+\sum_{j=1}^{m}\bigl[v_j(s_j)-qv_j(s_j')\bigr]=f_i(s_i+1)
 +(1-q)\sum_{j\ne i}v_j(s_j)\ge f_i(s_i+1),
\end{align*}
where the equality follows from the definition of
$v_i$, and the inequality uses $v_j(s_j)\ge0$.

Along this path, let $\boldsymbol s(\ell)$
be the state after arrival $\ell$, with $\boldsymbol s(0)=\boldsymbol0$.
The preceding inequality gives
\[
V(\boldsymbol s(\ell))-V(\boldsymbol s(\ell-1))
\le
1+\sum_{i=1}^{m}\bigl[v_i(s_i(\ell-1))-qv_i(s_i(\ell))\bigr].
\]
Multiplying by $q^{\ell-1}$ and summing over
$\ell=1,\ldots,B$, we obtain
\begin{align*}
\sum_{\ell=1}^{B}q^{\ell-1}\bigl[V(\boldsymbol s(\ell))-V(\boldsymbol s(\ell-1))\bigr]
&\le \sum_{\ell=1}^{B}q^{\ell-1}
+\sum_{i=1}^{m}\bigl[v_i(s_i(0))-q^Bv_i(s_i(B))\bigr]
\le \sum_{\ell=1}^{B}q^{\ell-1},
\end{align*}
where the last inequality uses
$v_i(s_i(0))=v_i(0)=0$ and $v_i(s_i(B))\ge0$ for every product $i$.
Since $\P(L\ge\ell)=q^{\ell-1}$, taking expectations
over the policy's randomization gives \(\Rev^\pi
\le
\sum_{\ell=1}^{B}q^{\ell-1}.\) Thus, \(\OPT_{\mathrm{on}}\leq \frac{1-q^B}{1-q}\).

We next lower bound the offline value.
For any realization $L=\ell$, the offline benchmark
can fill product $i=\min\{\ell,m\}$, earning
\(
R_i(i)
=\frac{i(i+1)}2
\frac{\sum_{k=1}^{i}q^{k-1}}
{\sum_{k=1}^{i}kq^{k-1}}.
\)
Therefore,
\(
\OPT_{\mathrm{off}}
\ge
\E[R_{\min\{L,m\}}(\min\{L,m\})],
\)
which implies
\begin{equation}\label{eq:upper-ifr-ratio}
\frac{\OPT_{\mathrm{on}}}{\OPT_{\mathrm{off}}}
\le
\frac{1-q^B}
{(1-q)\E[R_{\min\{L,m\}}(\min\{L,m\})]}.
\end{equation}

We evaluate this bound as $m\to\infty$.
Using the geometric-series identities,
\(
\sum_{k=1}^{i}q^{k-1}
=\frac{1-q^i}{1-q},
\sum_{k=1}^{i}kq^{k-1}
=\frac{1-(1+i/\sqrt m)q^i}{(1-q)^2},
\)
we obtain
\(
\frac{R_i(i)}{\sqrt m}
=
\frac{i(i+1)(1-q^i)}
{2m[1-(1+i/\sqrt m)q^i]}.
\)
For any fixed $x>0$, taking
$i=\lfloor x\sqrt m\rfloor$ gives
$i/\sqrt m\to x$ and $q^i\to e^{-x}$.
Hence,
\(
\frac{R_i(i)}{\sqrt m}
\longrightarrow g(x):=
\frac{x^2(1-e^{-x})}
{2[1-(1+x)e^{-x}]}.
\)
Since $\P(L=\ell)=(1-q)q^{\ell-1}$ for
$\ell<m$ and $\P(L\ge m)=q^{m-1}$,
\(
\frac{\E[R_{\min\{L,m\}}(\min\{L,m\})]}{\sqrt m}=
\frac1{\sqrt m}\sum_{\ell=1}^{m-1}
q^{\ell-1}\frac{R_\ell(\ell)}{\sqrt m}
+q^{m-1}\frac{R_m(m)}{\sqrt m}.
\)

To bound the sum from below, fix $0<\delta<M<\infty$.
On $\delta\le\ell/\sqrt m\le M$, the displayed formulas give
uniform convergence of $q^{\ell-1}R_\ell(\ell)/\sqrt m$
to $e^{-x}g(x)$, evaluated at $x=\ell/\sqrt m$.
Thus, the sum restricted to these indices is a Riemann sum
converging to $\int_\delta^M e^{-x}g(x)\,dx$.
All summands and the terminal term are nonnegative.
Discarding the terms outside this interval, taking the lower limit
as $m\to\infty$, and then letting $\delta\downarrow0$ and
$M\to\infty$ therefore yields
\begin{equation}\label{eq:upper-ifr-expected-limit}
\liminf_{m\to\infty}
\frac{\E[R_{\min\{L,m\}}(\min\{L,m\})]}{\sqrt m}
\ge
\int_0^\infty e^{-x}
\frac{x^2(1-e^{-x})}
{2[1-(1+x)e^{-x}]}\,dx.
\end{equation}
The integral is finite because $g(x)\sim x$ as $x\downarrow0$
and $g(x)\sim x^2/2$ as $x\to\infty$.

Finally, since $(1-q)\sqrt m=1$ and $q^B\to0$,
combining \eqref{eq:upper-ifr-ratio} and
\eqref{eq:upper-ifr-expected-limit} yields
\[
\limsup_{m\to\infty}
\frac{\OPT_{\mathrm{on}}}{\OPT_{\mathrm{off}}}
\le
\left[
\int_0^\infty e^{-x}
\frac{x^2(1-e^{-x})}
{2[1-(1+x)e^{-x}]}\,dx
\right]^{-1}
\approx0.66076149<0.6608.
\]
This completes the proof. \Halmos
\endproof
\subsubsection{Extension to General Arrivals.}
\label{app:dc-ifr-general}

Algorithm~\ref{alg:ifr} also achieves a constant guarantee under general arrival distributions. Without the IFR property, we replace the moment-based comparison with a weighted analysis of candidate horizons using arrival probabilities and residual offline values. A general inequality controls the interactions across horizons, yielding a $1/4$-competitive ratio through the same planning rule and telescoping argument.
\begin{theorem}[Planning-score guarantee and competitive ratio under general arrivals]
\label{thm:general-score-quarter}
Under Assumption~\ref{assump:discrete-concavity}, for any arrival distribution,
every active epoch with current allocation $\boldsymbol s$, and every
$\rho\in(0,1/4]$, $\max_{1\le k\le B-h_s}\Psi_{\boldsymbol s}^{\rho}(k)
\ge \rho\E[U_{\boldsymbol s}(K)\mid L\ge h_s+1]
=\Phi(\boldsymbol s)$.
Consequently, Algorithm~\ref{alg:ifr}, with $\rho=1/4$, is a deterministic
$1/4$-competitive algorithm under general arrivals when the arrival
distribution is known.
\end{theorem}

We first establish a bound on the cross terms between planning horizons.
The inequality holds for every distribution of the remaining arrival count.

\begin{lemma}[Weighted bound across horizons]
\label{lem:weighted-horizon}
For any probability distribution $(p_k)_{k=1}^N$ and nonnegative values
$z_1,\ldots,z_N$,
\begin{equation}\label{eq:weighted-horizon-inequality}
    \sum_{k,h=1}^N\frac{p_kp_hz_kz_h}{\max\{k,h\}}
    \le 2\sum_{k=1}^N\frac{p_kz_k^2}{k}
    \sum_{h=1}^Np_h\min\left\{1,\left(\frac hk\right)^2\right\}.
\end{equation}
\end{lemma}

\proof{Proof of Lemma~\ref{lem:weighted-horizon}:} We use a Green's-function representation to express the cross-horizon double sum as an integral, then establish a weighted bound that allows us to apply Cauchy--Schwarz.

Let
\(
    w_k:=\sum_{h=1}^Np_h\min\left\{1,\left(\frac hk\right)^2\right\},
    Q:=\sum_{k,h=1}^N\frac{p_kp_hz_kz_h}{\max\{k,h\}}.
\)
We show that $Q\le 2\sum_k p_kw_kz_k^2/k$.

Set $x_k=\log k$ and define
\(
    \varphi(\xi):=\sum_{h=1}^N\frac{p_hz_h}{\sqrt h}
       e^{-|\xi-x_h|/2},\xi\in\mathbb R.
\)
Since
$e^{-|x_k-x_h|/2}/\sqrt{kh}=1/\max\{k,h\}$,
we have $Q=\sum_k(p_kz_k/\sqrt k)\varphi(x_k)$.
Between consecutive points $x_k$, the function $\varphi$ satisfies
$\varphi''=\varphi/4$. At $x_k$, its first derivative has a downward jump
of $p_kz_k/\sqrt k$. Integrating by parts over these intervals, with the
boundary terms vanishing at infinity, gives
\begin{equation}\label{eq:general-kernel-energy}
    Q=\int_{\mathbb R}\left(\varphi'(\xi)^2
        +\frac14\varphi(\xi)^2\right)\,d\xi.
\end{equation}

The main step is to prove
\begin{equation}\label{eq:weighted-horizon-trace}
    \sum_{k=1}^N\frac{p_k}{w_k}\varphi(x_k)^2\le 2Q.
\end{equation}
To this end, define
\[
    \vartheta(\xi):=
    \frac{\sum_{h:x_h>\xi}p_h}
    {\sum_{h=1}^Np_he^{-2(\xi-x_h)_+}},
    \qquad (t)_+:=\max\{t,0\}.
\]
The denominator is continuous and equals $w_k$ at $x_k$, while the
numerator decreases by $p_k$ there. Hence \(\vartheta(x_k^-)-\vartheta(x_k^+)=\frac{p_k}{w_k},\)
where $x_k^-$ and $x_k^+$ denote left and right limits.
Moreover, $0\le\vartheta\le1$, and between the points $x_k$,
\(
    \vartheta'=2\vartheta(1-\vartheta).
\)
Set $F(\xi):=\vartheta(\xi)\varphi(\xi)^2$. Between successive points $x_k$, the product rule gives
\(
    F'
    =\vartheta'\varphi^2+2\vartheta\varphi\varphi'
    =2\vartheta(1-\vartheta)\varphi^2+2\vartheta\varphi\varphi'.
\)
Since $\varphi$ is continuous, the downward jump of $F$ at $x_k$ is
\[
    F(x_k^-)-F(x_k^+)
    =\bigl[\vartheta(x_k^-)-\vartheta(x_k^+)\bigr]\varphi(x_k)^2
    =\frac{p_k}{w_k}\varphi(x_k)^2.
\]
Also, $F(\xi)\to0$ as $\xi\to\pm\infty$, because
$0\le\vartheta\le1$ and $\varphi$ vanishes at infinity. We apply the fundamental theorem of calculus separately on
$(-\infty,x_1),(x_1,x_2),\ldots,(x_N,\infty)$ and sum.
The boundary terms at infinity vanish, while each interior point $x_k$
contributes its left limit minus its right limit. Thus
\[
    \int_{\mathbb R}F'(\xi)\,d\xi
    =\sum_{k=1}^N\bigl[F(x_k^-)-F(x_k^+)\bigr]
    =\sum_{k=1}^N\frac{p_k}{w_k}\varphi(x_k)^2,
\]
where the integral denotes the sum of the integrals over these intervals.
Substituting for $F'$ yields
\begin{equation}\label{eq:general-weighted-jump}
    \sum_{k=1}^N\frac{p_k}{w_k}\varphi(x_k)^2
    =2\int_{\mathbb R}\vartheta\varphi\varphi'\,d\xi
      +2\int_{\mathbb R}\vartheta(1-\vartheta)\varphi^2\,d\xi.
\end{equation}
Let $r(\xi):=\vartheta(\xi)-1/2$. Since $\varphi$ vanishes at infinity,
$\int_{\mathbb R}\varphi\varphi'\,d\xi=0$. Equation~\eqref{eq:general-weighted-jump}
therefore yields
\[
\begin{aligned}
    \frac12\sum_{k=1}^N\frac{p_k}{w_k}\varphi(x_k)^2
    &=\int_{\mathbb R}\left[r\varphi\varphi'+\left(\frac14-r^2\right)\varphi^2\right]d\xi\\
    &=\int_{\mathbb R}\left(\varphi'^2+\frac14\varphi^2\right)d\xi
      -\int_{\mathbb R}\left[\left(\varphi'-\frac r2\varphi\right)^2+\frac34r^2\varphi^2\right]d\xi
    \le Q,
\end{aligned}
\]
where the last step follows from~\eqref{eq:general-kernel-energy}.
This proves~\eqref{eq:weighted-horizon-trace}.

Finally, Cauchy--Schwarz and~\eqref{eq:weighted-horizon-trace} imply
\[
\begin{aligned}
    Q^2
    =\left(\sum_{k=1}^N
        \sqrt{\frac{p_kw_k}{k}}z_k
        \sqrt{\frac{p_k}{w_k}}\varphi(x_k)\right)^2 \le\left(\sum_{k=1}^N\frac{p_kw_kz_k^2}{k}\right)
        \left(\sum_{k=1}^N\frac{p_k}{w_k}\varphi(x_k)^2\right)\le 2Q\sum_{k=1}^N\frac{p_kw_kz_k^2}{k}.
\end{aligned}
\]
Dividing by $Q$ when $Q>0$ gives~\eqref{eq:weighted-horizon-inequality}.
The case $Q=0$ is immediate.\Halmos
\endproof

\proof{Proof of Theorem~\ref{thm:general-score-quarter}:}
Fix an active epoch with allocation $\boldsymbol s$, and condition
throughout on $L\ge h_s+1$. Write $N=B-h_s$ and
$p_k=\P(K=k\mid L\ge h_s+1)$.
We show that the score surplus $\Psi_{\boldsymbol s}^{\rho}(k)-\Phi(\boldsymbol s)$ has a nonnegative weighted sum with weights $p_kU_{\boldsymbol s}(k)/k$, using Lemma~\ref{lem:weighted-horizon} to control the cross-horizon terms.

By the definitions of the planning score and the potential,
\begin{equation}
\label{eq:general-quarter-deficit}
\begin{aligned}
&\Psi_{\boldsymbol s}^{\rho}(k)-\Phi(\boldsymbol s)\\
={}&U_{\boldsymbol s}(k)\left[
\sum_{h=1}^{N}p_h\frac{\min\{h,k\}(\min\{h,k\}+1)}{k(k+1)}
-\rho\sum_{h=k+1}^{N}p_h\frac{h-k+1}{h+k+1}\right]\\
&\quad-\rho\left[\sum_{h=1}^{k}p_hU_{\boldsymbol s}(h)
+\sum_{h=k+1}^{N}p_h\frac{2k}{h+k+1}U_{\boldsymbol s}(h)\right]\\
\ge{}&(1-\rho)U_{\boldsymbol s}(k)
\sum_{h=1}^{N}p_h\frac{\min\{h,k\}(\min\{h,k\}+1)}{k(k+1)} -\rho\left[\sum_{h=1}^{k}p_hU_{\boldsymbol s}(h)
+\sum_{h=k+1}^{N}p_h\frac{2k}{h+k+1}U_{\boldsymbol s}(h)\right]\\
\ge{}&(1-\rho)U_{\boldsymbol s}(k)
\sum_{h=1}^{N}p_h\min\{1,(h/k)^2\} -\rho\left[\sum_{h=1}^{k}p_hU_{\boldsymbol s}(h)
+\sum_{h=k+1}^{N}p_h\frac{2k}{h+k+1}U_{\boldsymbol s}(h)\right].
\end{aligned}
\end{equation}
The first inequality uses that, for $h>k$, the prefix-reward coefficient is one and $(h-k+1)/(h+k+1)\le1$; hence the subtracted continuation coefficient is at most $\rho$ times the full prefix coefficient. The second inequality uses $h(h+1)/[k(k+1)]\ge(h/k)^2$ for $h\le k$.

To control the subtracted term in \eqref{eq:general-quarter-deficit}, we multiply its bracketed expression by $p_kU_{\boldsymbol s}(k)/k$ and sum over $k$. Separating the diagonal $h=k$ and pairing the two orders of each $h<k$ gives
\begin{align}
&\sum_{k=1}^{N}\frac{p_kU_{\boldsymbol s}(k)}{k}
\left[\sum_{h=1}^{k}p_hU_{\boldsymbol s}(h)
+\sum_{h=k+1}^{N}p_h\frac{2k}{h+k+1}U_{\boldsymbol s}(h)\right]\notag\\
={}&\sum_{k=1}^{N}\frac{p_k^2U_{\boldsymbol s}(k)^2}{k}
+\sum_{1\le h<k\le N}\frac{p_hp_kU_{\boldsymbol s}(h)U_{\boldsymbol s}(k)}{k}
+\sum_{1\le h<k\le N}\frac{2p_hp_kU_{\boldsymbol s}(h)U_{\boldsymbol s}(k)}{h+k+1}\notag\\
={}&\sum_{k=1}^{N}\frac{p_k^2U_{\boldsymbol s}(k)^2}{k}+\sum_{1\le h<k\le N}p_hp_kU_{\boldsymbol s}(h)U_{\boldsymbol s}(k)\left(\frac1k+\frac2{h+k+1}\right)\notag\\
\le{}&\frac32\left[\sum_{k=1}^{N}\frac{p_k^2U_{\boldsymbol s}(k)^2}{k}+2\sum_{1\le h<k\le N}\frac{p_hp_kU_{\boldsymbol s}(h)U_{\boldsymbol s}(k)}{k}\right]\notag\\
={}&\frac32\sum_{k,h=1}^{N}
\frac{p_kp_hU_{\boldsymbol s}(k)U_{\boldsymbol s}(h)}{\max\{k,h\}}\notag
\le{}3\sum_{k=1}^{N}\frac{p_kU_{\boldsymbol s}(k)^2}{k}\sum_{h=1}^{N}p_h\min\{1,(h/k)^2\}.
\label{eq:general-quarter-cross-bound}
\end{align}
The first inequality uses
$1/k+2/(h+k+1)\le3/k$ for $h<k$,
and the second follows from Lemma~\ref{lem:weighted-horizon}.
Substituting this bound into \eqref{eq:general-quarter-deficit}
gives
\[
\sum_{k=1}^{N}\frac{p_kU_{\boldsymbol s}(k)}{k}
\bigl[\Psi_{\boldsymbol s}^{\rho}(k)-\Phi(\boldsymbol s)\bigr]
\ge(1-4\rho)\sum_{k=1}^{N}\frac{p_kU_{\boldsymbol s}(k)^2}{k}
\sum_{h=1}^{N}p_h\min\{1,(h/k)^2\}
\ge0.
\]

If $\Phi(\boldsymbol s)=0$, all scores are nonnegative. Otherwise, its definition ensures a positive weight, so the weighted sum gives $\Psi_{\boldsymbol s}^{\rho}(k)\ge\Phi(\boldsymbol s)$ for some $k$.

The competitive ratio follows from the potential argument used in
Theorem~\ref{thm:ifr-half}. In particular, the score certificate
\eqref{eq:ifr-score-certificate} relies on
Assumption~\ref{assump:discrete-concavity}, not on the IFR property.
Thus, the chosen horizon satisfies~\eqref{eq:ifr-nonnegative-drift}
at every active epoch. Since each active epoch assigns at least one
customer, summing the potential inequalities over the finite sequence
of epochs gives $\Rev^\pi\ge\rho\,\OPT_{\mathrm{off}}$.\Halmos
\endproof

\end{APPENDICES}
\end{document}